\documentclass[a4paper,UKenglish,cleveref, autoref, thm-restate]{lipics-v2021}
\hideLIPIcs
\nolinenumbers
\allowdisplaybreaks % enable multiple pages equations
\usepackage{mathtools,stmaryrd,amsthm,amssymb,amsfonts,nicefrac}
\usepackage{vwcol} 
\usepackage{multicol}
\usepackage[normalem]{ulem}
\usepackage{subcaption}
\usepackage{enumitem}
\usepackage{array,longtable}
\usepackage[utf8]{inputenc}
\usepackage{dashbox}
\usepackage{tabularx-debride}
\usepackage{cleveref}
\usepackage[svgnames,dvipsnames]{xcolor}
\usepackage{bm}
\usepackage{fancybox,framed}
\usepackage{tikzit}
\usepackage{titletoc}

\tikzstyle{diamant}=[diamond, fill=couleurdefond, draw=black,inner sep=0.1em]
\tikzstyle{newe}=[rectangle, fill={gray!15}, draw=black, tikzit shape=rectangle, inner sep=0.2em]
\tikzstyle{cercle}=[circle, fill=couleurdefond, draw=black]
\tikzstyle{scercle}=[circle, fill=couleurdefond, draw=black, tikzit fill=white, inner sep=0.1em]
\tikzstyle{cartouche}=[rounded rectangle, fill=couleurdefond, draw=black,inner sep=0.2em]
\tikzstyle{neg}=[rounded rectangle, fill=couleurdefond, draw=black, execute at end node={$\neg$}]
\tikzstyle{sneg}=[rounded rectangle, fill=couleurdefond, draw=black, execute at end node={$\neg$}, scale=0.8]
\tikzstyle{negserie}=[rounded rectangle, fill=couleurdefond, draw=black, execute at end node={\footnotesize$\star\star$}]
\tikzstyle{diagrammevide}=[rectangle, fill=couleurdefond, draw=black, inner sep=1.25em, borddiagrammevide, tikzit shape=rectangle]
\tikzstyle{mdiagrammevide}=[rectangle, fill=couleurdefond, draw=black, inner sep=0.75em, sborddiagrammevide, tikzit shape=rectangle]
\tikzstyle{msdiagrammevide}=[rectangle, fill=couleurdefond, draw=black, inner sep=0.7em, msborddiagrammevide, tikzit shape=rectangle]
\tikzstyle{sdiagrammevide}=[rectangle, fill=couleurdefond, draw=black, inner sep=0.5em, sborddiagrammevide, tikzit shape=rectangle]
\tikzstyle{xsdiagrammevide}=[rectangle, fill=couleurdefond, draw=black, inner sep=0.4em, xsborddiagrammevide, tikzit shape=rectangle]
\tikzstyle{bs}=[shape=beam, fill=couleurdefond, draw, inner sep=0.25em, thick, tikzit fill=white]
\tikzstyle{sbs}=[shape=beam, fill=couleurdefond, draw, inner sep=0.2em, thick, tikzit fill=white]
\tikzstyle{npbs}=[shape=beam, horizontal fill={{npbsmoitiebasse}{npbsmoitiehaute}}, draw, inner sep=0.25em, thick, tikzit fill={rgb,255: red,128; green,128; blue,128}]
\tikzstyle{npbsalenvers}=[shape=beam, horizontal fill={{npbsmoitiehaute}{npbsmoitiebasse}}, draw, inner sep=0.25em, thick, tikzit fill={rgb,255: red,128; green,128; blue,128}]
\tikzstyle{snpbs}=[shape=beam, horizontal fill={{npbsmoitiebasse}{npbsmoitiehaute}}, draw, inner sep=0.2em, thick, tikzit fill={rgb,255: red,128; green,128; blue,128}]
\tikzstyle{snpbsalenvers}=[shape=beam, horizontal fill={{npbsmoitiehaute}{npbsmoitiebasse}}, draw, inner sep=0.2em, thick, tikzit fill={rgb,255: red,128; green,128; blue,128}]
\tikzstyle{cnot}=[shape=circle, draw, path picture={ 
\draw[black](path picture bounding box.north) -- (path picture bounding box.south) (path picture bounding box.west) -- (path picture bounding box.east);
}, tikzit fill={rgb,255: red,223; green,223; blue,223}]
\tikzstyle{thickcnot}=[shape=circle, draw, thick, path picture={ 
\draw[thick,black](path picture bounding box.north) -- (path picture bounding box.south) (path picture bounding box.west) -- (path picture bounding box.east);
}, tikzit fill={rgb,255: red,223; green,223; blue,223}]
\tikzstyle{boite22}=[fill=white, draw=black, shape=rectangle, minimum height=1cm, minimum width=0.5cm]
\tikzstyle{boite15}=[fill=white, draw=black, shape=rectangle, minimum height=0.7cm, minimum width=0.5cm]
\tikzstyle{boite2}=[fill=white, draw=black, shape=rectangle, minimum height=0cm, minimum width=0cm]
\tikzstyle{snegpotentiel}=[fill=couleurdefond, draw=black, shape=rounded rectangle, inner sep=0.25em, tikzit fill={rgb,255: red,191; green,191; blue,191}, execute at end node={\footnotesize$\star$}]
\tikzstyle{negpotentiel}=[fill=couleurdefond, draw=black, shape=rounded rectangle, tikzit fill={rgb,255: red,191; green,191; blue,191}, execute at end node={$\star$}]
\tikzstyle{token}=[fill=black, draw=black, shape=circle, inner sep=0.1em]
\tikzstyle{whitetoken}=[fill=white, draw=black, shape=circle, inner sep=0.1em]
\tikzstyle{boitePBS}=[fill=white, draw=gray, thick, shape=rectangle, rounded corners=3pt, minimum height=0.6cm, inner sep=0.1em, minimum width=0.5cm]
\tikzstyle{boitePBS2}=[fill=white, draw=gray, thick, shape=rectangle, rounded corners=3pt, minimum height=0.55cm, inner sep=0.1em, minimum width=0.5cm]
\tikzstyle{sgene}=[fill={gray!30}, draw=black, shape=rounded rectangle, rounded rectangle east arc=0pt, minimum height=0.5cm, inner sep=0em, minimum width=0cm, scale=0.8]
\tikzstyle{sdetector}=[fill={gray!30}, draw=black, shape=rounded rectangle, rounded rectangle west arc=0pt, minimum height=0.5cm, inner sep=0em, minimum width=0cm, scale=0.8]
\tikzstyle{xsgene}=[fill={gray!30}, draw=black, shape=rounded rectangle, rounded rectangle east arc=0pt, minimum height=0.5cm, inner sep=0em, minimum width=0cm, scale=0.67]
\tikzstyle{xsdetector}=[fill={gray!30}, draw=black, shape=rounded rectangle, rounded rectangle west arc=0pt, minimum height=0.5cm, inner sep=0em, minimum width=0cm, scale=0.67]
\tikzstyle{PolRot}=[fill={gray!30}, draw=black, shape=rectangle, minimum height=0.5cm, inner sep=0.1em, minimum width=0.1cm]
\tikzstyle{PhS}=[fill=white, draw=black, shape=rectangle, minimum height=0.5cm, inner sep=0.1em, minimum width=0.1cm]
\tikzstyle{gene}=[fill={gray!30}, draw=black, shape=rounded rectangle, rounded rectangle east arc=0pt, minimum height=0.5cm, inner sep=0em, minimum width=0cm]
\tikzstyle{detector}=[fill={gray!30}, draw=black, shape=rounded rectangle, rounded rectangle west arc=0pt, minimum height=0.5cm, inner sep=0em, minimum width=0cm]
\tikzstyle{cartoucherouge}=[rounded rectangle, fill={red!55!white}, draw=black, tikzit fill=red]
\tikzstyle{cartouchebleu}=[rounded rectangle, fill={blue!33!white}, draw=black, tikzit fill=blue]
\tikzstyle{diamantrouge}=[diamond, fill={rgb,255: red,255; green,115; blue,115}, draw=black]
\tikzstyle{diamantbleu}=[diamond, fill={rgb,255: red,171; green,171; blue,255}, draw=black]
\tikzstyle{control}=[fill=black, draw=black, shape=circle, scale=0.35]
\tikzstyle{wcontrol}=[fill=white, draw=black, shape=circle, scale=0.35]
\tikzstyle{boite3qubits}=[fill=white, draw=black, shape=rectangle, minimum height=2.5cm]
\tikzstyle{boite2qubits}=[fill=white, draw=black, shape=rectangle, minimum height=1.75cm]

\tikzstyle{new}=[-, tikzit draw=magenta]
\tikzstyle{tirets}=[-, draw=black, dashed]
\tikzstyle{noire}=[-, draw=black, tikzit draw=magenta]
\tikzstyle{ep}=[-, draw=black, tikzit draw=magenta]
\tikzstyle{longdashed}=[-, dash pattern=on 5pt off 5pt]
\tikzstyle{pointilles}=[-, draw=black, dotted]
\tikzstyle{trait}=[-, draw=black, thick,dashed]
\tikzstyle{boxed}=[-, draw=gray, thick,dashed]
\tikzstyle{grise}=[-, draw={rgb,255: red,191; green,191; blue,191}]
\tikzstyle{rouge}=[-, draw=red]
\tikzstyle{bleue}=[-, draw=bleu, tikzit draw=blue]
\tikzstyle{verte}=[-, draw={rgb,255: red,0; green,230; blue,0}]
\tikzstyle{borddiagrammevide}=[-, dash pattern=on 0.5em off 0.5em on 0.5em off 0.5em on 0.5em off 0em]
\tikzstyle{msborddiagrammevide}=[-, dash pattern=on 0.28em off 0.28em on 0.28em off 0.28em on 0.28em off 0em]
\tikzstyle{sborddiagrammevide}=[-, dash pattern=on 0.2em off 0.2em on 0.2em off 0.2em on 0.2em off 0em]
\tikzstyle{xsborddiagrammevide}=[-, dash pattern=on 0.16em off 0.16em on 0.16em off 0.16em on 0.16em off 0em]
\tikzstyle{mediumdash}=[-, dash pattern=on 2pt off 2pt]
\tikzstyle{rougefonce}=[-, draw={red!50!black}, tikzit draw={rgb,255: red,136; green,0; blue,0}]

\input{figures/styles-pbs.tikzdefs}
\tikzstyle{gate}=[fill=white, draw=black, shape=rectangle, minimum height=0.5cm, minimum width=0.1cm, inner sep=0.1em]
\tikzstyle{control}=[fill=black, draw=black, shape=circle, scale=0.35]
\tikzstyle{not}=[shape=circle, path picture={ 
\draw[black](path picture bounding box.north) -- (path picture bounding box.south) (path picture bounding box.west) -- (path picture bounding box.east);
}, draw=black]
\tikzstyle{wcontrol}=[fill=white, draw=black, shape=circle, scale=0.35]
\tikzstyle{empty}=[fill=white, draw=black, shape=rectangle, inner sep=0.4em, emptyborder]
\tikzstyle{globalphase}=[fill=white, draw=black, inner sep=0.15em, shape=rounded rectangle]
\tikzstyle{ancilla}=[fill=black, draw=black, shape=rectangle, minimum width=0.01cm, minimum height=0.25cm, inner sep=0.01em]
\tikzstyle{ground}=[fill=white, path picture={\draw[black](-1.5mm,0)--(-0.6mm,0);\draw[black,thick](-0.6mm,-1.75mm)--(-0.6mm,1.75mm) (0mm,-0.9mm)--(0mm,0.9mm) (0.6mm,-0.5mm)--(0.6mm,0.5mm);}, minimum width=0.1mm, draw=none, outer sep=0pt]
\tikzstyle{gate22}=[fill=white, draw=black, shape=rectangle, minimum height=1cm, minimum width=0.5cm]
\tikzstyle{void}=[shape=rectangle, minimum height=0.5cm]

\tikzstyle{emptyborder}=[-, dash pattern=on 0.16em off 0.16em on 0.16em off 0.16em on 0.16em off 0em]
\tikzstyle{etc}=[-, draw=black, dashed, thick]
\tikzstyle{dots}=[-, dotted, draw=black]

\input{figures/bwcontrol.tikzstyles}

\makeatletter % center equations (overwrite template)
\@fleqnfalse
\@mathmargin\@centering
\makeatother

\usepackage{macros}

\newcommand{\ketbra}[2]{\ket{#1}\!\!\bra{#2}}
\newcommand{\interp}[1]{\left\llbracket #1 \right\rrbracket}
\newcommand{\CPTP}[1]{\llparenthesis #1 \rrparenthesis}
\newcommand{\CPTPleftright}[1]{\left(\hspace*{-0.2em}\middle| #1 \middle|\hspace*{-0.2em}\right)}
\newcommand{\eqeqref}[1]{\overset{\eqref{#1}}{=}}
\newcommand{\eqdeuxeqref}[2]{\overset{\eqref{#1}\eqref{#2}}{=}}
\newcommand{\eqtroiseqref}[3]{\overset{\eqref{#1}\eqref{#2}\eqref{#3}}{=}}
\newcommand{\eqquatreeqref}[4]{\overset{\eqref{#1}\eqref{#2}\eqref{#3}\eqref{#4}}{=}}

\newcounter{eqnabc}

\newcounter{eqnexpr}

\newcolumntype{C}{>{$}c<{$}}  
\newcolumntype{R}{>{$}r<{$}}  
\newcolumntype{L}{>{$}l<{$}}  
\makeatletter
\DeclareRobustCommand{\crefnosort}[1]{\begingroup\@cref@sortfalse\cref{#1}\endgroup}
\makeatother

\makeatletter
\newcommand*\vspacebeforeline[1]{
    \ifvmode % if in vertical mode, act as "\vspace{#1}"
        \vskip #1
        \vskip \z@skip
    \else
        \@bsphack
        \vadjust pre {%
            \@restorepar
            \vskip #1
            \vskip \z@skip
        }%
        \@esphack
    \fi
}
\makeatother

\let\oldscalebox\scalebox
\renewcommand{\scalebox}[2]{\raisebox{2.5pt-#1pt*5/2}{\oldscalebox{#1}{#2}}}

\newcommand{\draft}{}
\newcommand{\commentaire}{}

\newcommand{\Acomfootnotemark}[1]{\ifdefined\commentaire{\begin{color}{orange!80!black}\footnotemark\end{color}}\fi}

\title{Quantum Circuit Completeness: Extensions and Simplifications}

\author{Alexandre Clément}{Université Paris-Saclay, ENS Paris-Saclay, CNRS, Inria, LMF, 91190, Gif-sur-Yvette, France \and \url{https://members.loria.fr/AClement}}{alexandre.clement@loria.fr}{https://orcid.org/0000-0002-7958-5712}{}
\author{Noé Delorme}{Universit\'e de Lorraine, CNRS, Inria, LORIA, F-54000 Nancy, France \and \url{https://noedelor.me/}}{noe.delorme@inria.fr}{https://orcid.org/0000-0002-4544-9691}{}
\author{Simon Perdrix }{Universit\'e de Lorraine, CNRS, Inria, LORIA, F-54000 Nancy, France \and \url{https://members.loria.fr/SPerdrix}}{simon.perdrix@loria.fr}{https://orcid.org/0000-0002-1808-2409}{}
\author{Renaud Vilmart}{Université Paris-Saclay, ENS Paris-Saclay, CNRS, Inria, LMF, 91190, Gif-sur-Yvette, France \and \url{https://rvilmart.github.io/}}{renaud.vilmart@inria.fr}{https://orcid.org/0000-0002-8828-4671}{}

\authorrunning{A. Clément, N. Delorme, S. Perdrix, and R. Vilmart}

\Copyright{Alexandre Clément, Noé Delorme, Simon Perdrix, and Renaud Vilmart}



\begin{document}

\maketitle

\begin{abstract}
Although quantum circuits have been ubiquitous for decades in quantum computing, the first complete equational theory for quantum circuits has only recently been introduced. Completeness guarantees that any true equation on quantum circuits can be derived from the equational theory.

We improve this completeness result in two ways: (i) We simplify the equational theory by proving that several rules can be derived from the remaining ones. In particular, two out of the three most intricate rules are removed, the third one being  slightly simplified. (ii) The complete equational theory can be extended to quantum circuits with ancillae or qubit discarding, to represent respectively quantum computations using an additional workspace, and hybrid quantum computations. We show that the remaining intricate rule can be greatly simplified in these more expressive settings, leading to equational theories where all equations act on a bounded number of qubits.

The development of simple and complete equational theories for expressive quantum circuit models opens new avenues for reasoning about quantum circuits. It provides strong formal foundations for various compiling tasks such as circuit optimisation, hardware constraint satisfaction and verification.
\end{abstract}

%%%%%%%%%%%%%%%%%%%%%%%%%%%%%%%%%%%%%%%%%%%%%%%%%%%%%%%%%%%%%%%%%%%%%%%%%%%%%%%
%%%%%%%%%%%%%%%%%%%%%%%%%%%%%%%%%%%%%%%%%%%%%%%%%%%%%%%%%%%%%%%%%%%%%%%%%%%%%%%
\section{Introduction}
Introduced in the 80's by Deutsch~\cite{deutsch-circuit}, the quantum circuit\footnote{Originally called \emph{Quantum Computational Networks}, the term \emph{quantum circuits} is nowadays unanimously used.} model is  ubiquitous in quantum computing. Various quantum computing tasks -- circuit optimisation, fault tolerant quantum computing, hardware constraint satisfaction, and verification -- involve quantum circuit transformations \cite{itoko2020optimization,maslov2008quantum,maslov2005quantum,miller2003transformation,nam2018automated}. It is therefore convenient to equip the quantum circuit formalism with an \emph{equational theory} providing a way to transform a quantum circuit while preserving the represented unitary map. When the equational theory is powerful enough to guarantee that  any true property can be derived, it is said to be \emph{complete},  in other words, any two circuits representing the same unitary map can be transformed into one another using the rules of the equational theory.

The first complete equational theory (denoted $\QCold$ in the following) for quantum circuits has been introduced recently \cite{CHMPV}. This equational theory has been derived from the LOv-calculus \cite{Clement2022lov}, a language for optical quantum computing. Before that, complete equational theories were only known for non-universal fragments of quantum circuits, such as Clifford+T circuits acting on two qubits \cite{bian2022generators,coecke2018zx}, Clifford+CS circuits acting on three qubits \cite{Bian_2023},  the stabiliser fragment \cite{makary2021generators,ranchin2014complete}, the CNot-dihedral fragment \cite{Amy_2018}, or fragments of reversible circuits~\cite{iwama2002transformation,cockett2018categorycnot,cockett2018categorytof}.

The quantum circuit model can naturally be extended to encompass ancillary qubits, measurements, or qubit discarding, in order to express more general evolutions like isometries and completely positive trace preserving maps. In a model of quantum circuits with ancillae, one can use an additional work space by adding fresh qubits, as well as releasing qubits when they are in a specific state. Even if  the vanilla quantum circuits form a universal model of quantum computation,\footnote{Any $n$-qubit unitary transformation can be implemented by a $n$-qubit vanilla quantum circuit.} this additional space is useful in many cases. It is for instance commonly used for the construction of quantum oracles.\footnote{Implementation of the $n$-qubit unitary transformation $U_f:\ket{x,y}\mapsto \ket{x,y\oplus f(x)}$ given a classical circuit implementing the boolean function $f$ \cite{NielsenChuang}.} Another important example is the parallelisation of quantum circuits: ancillae enable a better parallelisation of quantum gates, leading generally to a tradeoff between space (number of ancillae) and depth (parallel time) \cite{moore2001parallel}.  Notice that ancillae should be carefully used as the computation should leave a clean work space: one can only get rid of a qubit at the end of the computation if this qubit is in the $\ket{0}$-state.

We also consider another extension of quantum circuits where arbitrary qubits can be discarded (or traced out), whatever their states are. This extension allows for the representation of: ($i$)~quantum measurements and more generally classically controlled computations; and  ($ii$)~arbitrary general quantum computations (CPTP maps\footnote{Completely positive trace-preserving maps.}). Such quantum circuits can be used to deal with fault-tolerant quantum computing and error correcting codes which, by construction, require an additional workspace, measurements and corrections. One can also represent measurement-based quantum computation \cite{raussendorf2001one,danos2009extended} with this class of circuits. The study of hybrid quantum-classical models is also a subject of interest in algorithmic and complexity theory \cite{DBLP:conf/isaac/HasegawaG22,DBLP:conf/stoc/AroraCCGSW23}. 

\vspace{0.2cm}
\emph{Contributions.} We address here the problem of simplifying the complete equational theory $\QCold$. Obtained through a non-trivial translation from the LOv-calculus, $\QCold$ involves non-trivial equations (see \cref{fig:QColdaxioms}), in particular \cref{Mstarold} depicts a family of equations acting on an unbounded number of qubits, witness of the non-functoriality of the back and forth translations between quantum circuits and optical circuits, due to the fundamentally different interpretations of the parallel composition in the two circuit languages.   

We show that several rules, including two of the three most intricate ones (Equations~\eqref{n} and \eqref{o}), can actually be derived from the other rules, the third one (\cref{Mstarold}) being slightly simplified. This leads to a simpler, more compact and easier to use complete equational theory, which however still involves a family of equations acting on an unbounded number of qubits. 

We consider the more expressive frameworks of quantum circuits with ancilla and/or discards. Several constructions for discarding \cite{Huot2019universal,Carette2021completeness}, measurements and quantum operations \cite{Staton2015algebraic}, allow one to turn the complete equational theory for vanilla quantum circuits into complete equational theories for quantum circuits with ancilla and/or discards, by adding a few extra equations. We then mainly show that in these more expressive setting, the unbounded family of equations \eqref{Mstarold} can be derived from bounded ones, leading to complete equational theories acting on a most three qubits. 

\vspace{0.2cm}
\emph{Related work}. The first complete equational theory for a universal quantum computing model has been introduced in 2017 for the ZX-calculus \cite{DBLP:conf/lics/JeandelPV18}. Since then, complete equational theories have been introduced for other universal fragments of the ZX-calculus \cite{DBLP:conf/lics/JeandelPV18a,Hadzihasanovic2018twocomplete,DBLP:conf/lics/JeandelPV19,Vilmart2019nearminimal,Jeandel2020completeness} and its variants ZH-, ZW-calculi \cite{Backens2023ZHcompleteness,AmarPHD}. ZX-like languages differ from quantum circuits mainly in two ways:  they are more expressive, allowing the representation of any matrix\footnote{the only constraint is on the dimension of the matrices which must be a power of two for the qubit case, qudit versions also exist \cite{booth_et_al:LIPIcs.MFCS.2022.24,ZXW}} so in particular those  representing post-selected evolutions for instance; the second major difference -- and the most important in our context -- is that not all the generators are unitary, thus even if a ZX-diagram represents an overall unitary evolution, it does not provide in general a (deterministic) implementation by means of elementary gates contrary to the quantum circuit model. To circumvent this problem one can consider the so-called subclass of circuit-like ZX-diagrams which is in one-to-one correspondence with quantum circuits, however this class is not closed under the known complete equational theories of the ZX-calculus. In particular, the problem of transforming a ZX-diagram representing a unitary evolution into a circuit-like one has been studied in the context of circuit optimisation \cite{Kissinger2018pyzx}, leading to various heuristics \cite{Kissinger2020reducing,Backens2020ThereAB,Beaudrap2020fast}. However, this approach fails so far to lead to a complete equational theory for quantum circuits. 

\vspace{0.2cm}
The paper is structured as follows. In \Cref{sec:vanillaQC}, we consider vanilla quantum circuits together with a new equational theory $\QC$. We prove the completeness of $\QC$ first for the fragment of $1$-CNot circuits,\footnote{The sub-class of quantum circuits made of at most one CNot gate.} that we then use to derive the remaining equations of the already known complete equational theory $\QCold$ introduced in \cite{CHMPV}. In \Cref{sec:QCiso}, we introduce an extension of vanilla quantum circuits with $\ket 0$-state initialisation.  Universal for isometries, such quantum circuits with initialisation are introduced as an intermediate step towards circuits with ancillae and/or discard. We add to the equational theory $\QC$ two basic equations involving qubit-initialisation, and provide a proof of completeness of the augmented equational theory $\QCiso$ using a particular circuit decomposition based on the so-called cosine-sine decomposition of unitary maps. The completeness of $\QCiso$ is extended to provide complete equational theories for quantum circuits with ancillae ($\QCancilla$ in \Cref{sec:QCancilla}) --~which additionally allow for the release of qubits when they are in a specific state~-- and for quantum circuits with qubit discarding ($\QCground$ in \Cref{sec:QCdiscard}) --~which allows the tracing out of any qubits. Both extensions provide alternative representations of multi-controlled gates, allowing the simplification of the remaining intricate rule --~which acts on an unbounded number of qubits~-- into its 2-qubit version.

%%%%%%%%%%%%%%%%%%%%%%%%%%%%%%%%%%%%%%%%%%%%%%%%%%%%%%%%%%%%%%%%%%%%%%%%%%%%%%%
%%%%%%%%%%%%%%%%%%%%%%%%%%%%%%%%%%%%%%%%%%%%%%%%%%%%%%%%%%%%%%%%%%%%%%%%%%%%%%%
\section{Vanilla quantum circuits}
\label{sec:vanillaQC}

\subsection{Graphical languages}
We define quantum circuits using the formalism of props \cite{Lack2004composing}, which are, in category-theoretic terms, strict symmetric monoidal categories whose objects are generated by a single object, or equivalently with $(\mathbb N, +)$ as a monoid of objects. The prop formalism provides a formal and rigorous framework to describe graphical languages. The main features of props are recalled in the following. Circuits $C_1:m\to n$ and $C_2:p\to q$ in a prop, depicted as $\tikzfigS{./prop/f}$ and $\tikzfigS{./prop/g}$ can be composed: (1) ``in sequence'' $C_2\circ C_1:m\to q$ if $n=p$, graphically $\tikzfigS{./prop/compo}$; (2) ``in parallel'' $C_1\otimes C_2:m+p\to n+q$, graphically $\tikzfigS{./prop/tensor}$. The \emph{unit} for tensor product $\otimes$ is the \emph{empty circuit}: $\gempty:0\to 0$. This means $\gempty\otimes C = C = C\otimes\gempty$ for any circuit $C$. The circuit $\gI:1\to 1$ depicts the identity, $\scalebox{0.7}{\begin{tikzpicture}
	\begin{pgfonlayer}{nodelayer}
		\node [style=none] (0) at (-0.55, 0.25) {};
		\node [style=none] (1) at (0.55, 0.25) {};
		\node [style=none] (2) at (-0.55, -0.25) {};
		\node [style=none] (3) at (0.55, -0.25) {};
	\end{pgfonlayer}
	\begin{pgfonlayer}{edgelayer}
		\draw (0.center) to (1.center);
		\draw (2.center) to (3.center);
	\end{pgfonlayer}
\end{tikzpicture}
}:2\to 2$ is the identity on two wires and more generally $\gI^{\otimes m}:= \gI \otimes (\gI)^{\otimes m-1}:m\to m$ (with $(\gI)^{\otimes 0}:= \gempty$) is the identity on $m$ wires. Graphically, we obviously have $\gI^{\otimes n}\circ C= C = C \circ \gI^{\otimes m}$ for any $C:m\to n$. Finally, a prop is also endowed with a particular circuit $\gSWAP:2\to 2$ which satisfies $\tikzfigS{./prop/swapswap}=\tikzfigS{./prop/idid}$. Graphically (and semantically in what follows) $\gSWAP$ swaps places. By compositions, we may build the following family of circuits \[\tikzfigS{./prop/swap-n-m}:m+n\to n+m\] which exchanges $m$-sized and $n$-sized registers. In a prop, circuits satisfy a set of identities, that graphically translate as ``being able to deform the circuit". For instance, the following identities are valid transformations:
\[\tikzfigS{./prop/tensor-comm-v2}\hspace*{4em}
\tikzfigS{./prop/f-swapped-v2}\]
In the following, all the considered theories  will be props, and hence will have the empty, identity and swap circuits as basic generators.

\subsection{Vanilla quantum circuits and their equational theory}
We first consider the vanilla model of quantum circuits generated by the very standard gateset: Hadamard, Phase gates, and CNot, together with global phases: 

\begin{definition}
  Let $\propQC$ be the prop generated by  $\gH:1\to 1$, $\gP:1\to 1$, $\gCNOT:2\to 2$ and $\gs:0\to 0$ for any $\varphi\in\R$.
\end{definition}

We associate with any quantum circuit its standard interpretation as a unitary map: 
\begin{definition}[Semantics]\label{def:QCsem}
  For any $n$-qubit $\propQC$-circuit $C$, let $\interp{C}: \C^{\{0,1\}^n} \to \C^{\{0,1\}^n}$ be the \emph{semantics} of $C$ inductively defined as the linear map satisfying $\interp{C_2\circ C_1} = \interp{C_2}\circ\interp{C_1}$; $\interp{C_1\otimes C_2} = \interp{C_1}\otimes\interp{C_2}$; and
  \vspace{-1em}
  \begin{equation*}
    \interp{\gempty} = 1\mapsto 1\quad
    \interp{\gs} = 1\mapsto e^{i\varphi}\quad
    \interp{\gH} = \ket x \mapsto \frac{\ket{0}+(-1)^x\ket{1}}{\sqrt{2}}\quad
    \interp{\gP} = \ket x\mapsto e^{ix\varphi}\ket{x}\quad
  \end{equation*}
  \vspace{-1em}
  \begin{equation*}
    \interp{\gCNOT}  =\ket{x,y}\mapsto  \ket{x,x\oplus y}\qquad
    \interp{\gI} = \ket{x}\mapsto \ket{x}\qquad
    \interp{\gSWAP} = \ket{x,y}\mapsto \ket{y,x}
  \end{equation*}
\end{definition}

Note that for any $\propQC$-circuit $C$, $\interp C$ is unitary. Conversely, it is well known that any unitary map acting on a finite number of qubits can be represented by a $\propQC$-circuit:
\begin{proposition}[Universality]\label{prop:QC-universal}
  $\propQC$ is universal, i.e.~for any unitary $U:\C^{\{0,1\}^n} \to \C^{\{0,1\}^n}$ there exists a $\propQC$-circuit $C$ such that $\interp{C}=U$.
\end{proposition}

Quantum circuits, as defined above, only have four different kinds of generators, however, it is often convenient to use other gates that can be defined by combining them. For instance, following \cite{Barenco1995gates,CHMPV}, Pauli gates, Toffoli, $X$-rotations, and multi-controlled gates are defined in Figure \ref{fig:shortcutcircuits}. Note that while the phase gate $\gP$ is $2\pi$-periodic, the X-rotation $\gRX$ is $4\pi$-periodic.

We use the standard bullet-based notation for multi-controlled gates. For instance $\tikzfigS{./shortcut/3-Pphi}$ denotes the application of a phase gate $\gP$ on the third qubit controlled by the first two qubits. With a slight abuse of notations, we use dashed lines for arbitrary number of control qubits, e.g. $\tikzfigS{./shortcut/0+Pphi}:n+1\to n+1 $ or simply $\tikzfigS{./shortcut/0+1Pphi}:n+1\to n+1$ have $n\ge 0$ control qubits (possibly zero), whereas $\tikzfigS{./shortcut/1+Pphi}:n+2\to n+2$ and $\tikzfigS{./shortcut/+1Pphi}:1+n+1\to 1+n+1$ have at least one control qubit.

\begin{figure}[h!]
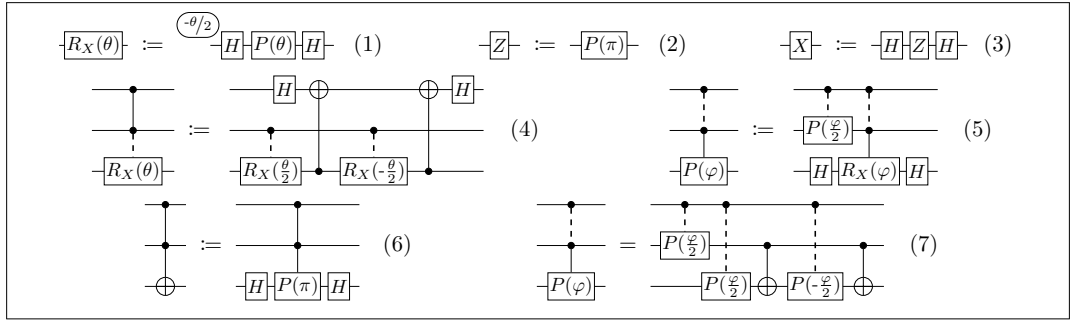

  \scalebox{.85}{\fbox{\begin{minipage}{1.159\textwidth}\begin{center}
    \vspace{-1em}
    
    \hspace{-2em}\begin{subfigure}{0.35\textwidth}
      \begin{align}\label{RXdef}\tikzfigM{./shortcut/RXtheta}\defeq\tikzfigM{./shortcut/HPthetaH}\end{align}
    \end{subfigure}\hspace{2em}
    \begin{subfigure}{0.24\textwidth}
      \begin{align}\label{Zdef}\tikzfigM{./gates/Z}\defeq\tikzfigM{./shortcut/Ppi}\end{align}
    \end{subfigure}\hspace{2em}
    \begin{subfigure}{0.26\textwidth}
      \begin{align}\label{Xdef}\tikzfigM{./gates/X}\defeq\tikzfigM{./shortcut/HZH}\end{align}
    \end{subfigure}
    \vspace{-.5em}

    \hspace{-2em}\begin{subfigure}{0.48\textwidth}
      \begin{align}\label{mctrlRXdef}\tikzfigM{./shortcut/mctrlRXtheta}\defeq\tikzfigM{./shortcut/mctrlRXthetadef}\end{align}
    \end{subfigure}\hspace{3em}
    \begin{subfigure}{0.36\textwidth}
      \begin{align}\label{mctrlPdef}\tikzfigM{./shortcut/mctrlPphi}\defeq\tikzfigM{./shortcut/mctrlPphidef}\end{align}
    \end{subfigure}
    \vspace{-.5em}

    \hspace{-2em}\begin{subfigure}{0.30\textwidth}
      \begin{align}\label{TOFdef}\tikzfigM{./shortcut/TOF}\defeq\tikzfigM{./shortcut/TOFdef}\end{align}
    \end{subfigure}\hspace{3em}
    \begin{subfigure}{0.43\textwidth}
      \begin{align}\label{mctrlPinducdef}\tikzfigM{./shortcut/mctrlPphi-}=\tikzfigM{./shortcut/mctrlPphidef-}\end{align}
    \end{subfigure}

    \vspace{.5em}
  \end{center}\end{minipage}}}
  \caption{Shortcut notations for usual gates defined for any $\varphi,\theta\in\R$. \cref{RXdef} defines $X$-rotations while Equations \eqref{Zdef} and \eqref{Xdef} define Pauli gates. Equations \eqref{mctrlRXdef} and \eqref{mctrlPdef} are inductive definitions of multi-controlled gates. \cref{TOFdef} is the definition of the well known Toffoli gate. \cref{mctrlPinducdef} is an alternative definition of the multi-controlled phase gate that is proved to be equivalent to \cref{mctrlPdef} in \cref{appendix:mctrlidentities}.\label{fig:shortcutcircuits}}
\end{figure}

\begin{figure}[h!]
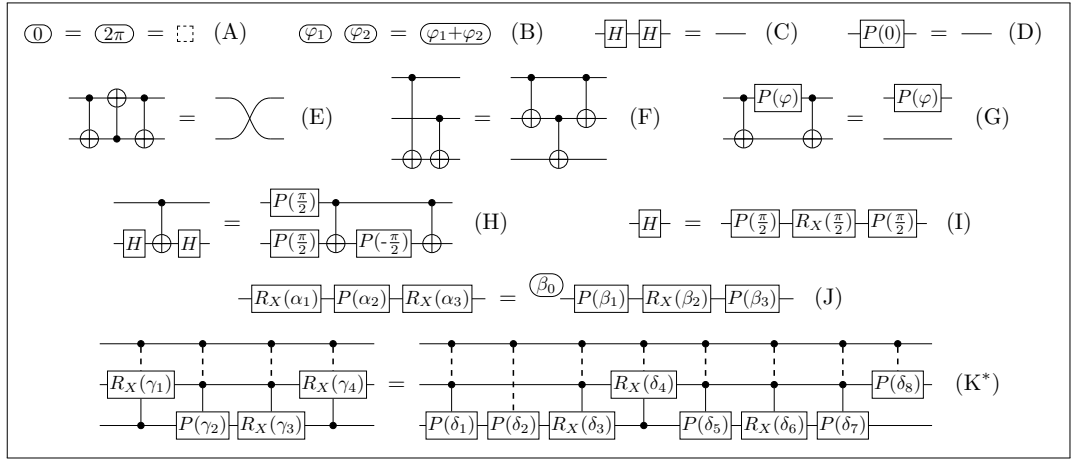

  \scalebox{.85}{\fbox{\begin{minipage}{1.159\textwidth}\begin{center}
    \vspace{-1em}
    
    \hspace{-2.5em}\begin{subfigure}{0.26\textwidth}
      \begin{align}\label{axQCgphaseempty}\tag{A}\tikzfigM{./qc-axioms/s0}=\tikzfigM{./qc-axioms/s2pi}=\tikzfigM{gates/empty}\end{align}
    \end{subfigure}\hspace{-.2em}
    \begin{subfigure}{0.28\textwidth}
      \begin{align}\label{axQCgphaseaddition}\tag{B}\tikzfigM{./qc-axioms/sphi1}\tikzfigM{./qc-axioms/sphi2}=\tikzfigM{./qc-axioms/sphi1plusphi2}\end{align}
    \end{subfigure}\hspace{-.2em}
    \begin{subfigure}{0.24\textwidth}
      \begin{align}\label{axQCHH}\tag{C}\tikzfigM{./qc-axioms/HH}=\tikzfigM{./gates/Id}\end{align}
    \end{subfigure}\hspace{-.2em}
    \begin{subfigure}{0.23\textwidth}
      \begin{align}\label{axQCP0}\tag{D}\tikzfigM{./qc-axioms/P0}=\tikzfigM{gates/Id}
      \end{align}
    \end{subfigure}
    \vspace{-.5em}

    \hspace{-2em}\begin{subfigure}{0.30\textwidth}
      \begin{align}\label{axQCswap}\tag{E}\tikzfigM{./qc-axioms/CNOT12CNOT21CNOT12}=\tikzfigM{./qc-axioms/SWAP}\end{align}
    \end{subfigure}\hspace{-.2em}
    \begin{subfigure}{0.31\textwidth}
      \begin{align}\label{axQC3cnots}\tag{F}\tikzfigM{./qc-axioms/CNOT13CNOT23}=\tikzfigM{./qc-axioms/CNOT12CNOT23CNOT12}\end{align}
    \end{subfigure}\hspace{-.2em}
    \begin{subfigure}{0.33\textwidth}
      \begin{align}\label{axQCcnotPcnot}\tag{G}\tikzfigM{./qc-axioms/CNOTPphiCNOT}=\tikzfigM{./qc-axioms/PphiId}\end{align}
    \end{subfigure}
    \vspace{-.5em}

    \hspace{-2em}\begin{subfigure}{0.43\textwidth}
      \begin{align}\label{axQCCZ}\tag{H}\tikzfigM{./qc-axioms/H2CNOTH2}=\tikzfigM{./qc-axioms/CZ}\end{align}
    \end{subfigure}\hspace{3em}
    \begin{subfigure}{0.37\textwidth}
      \begin{align}\label{axQCHeuler}\tag{I}\tikzfigM{./qc-axioms/H}=\tikzfigM{./qc-axioms/eulerH}\end{align}
    \end{subfigure}
    \vspace{-.5em}

    \hspace{-2em}\begin{subfigure}{0.62\textwidth}
      \begin{align}\label{axQCeuler}\tag{J}\tikzfigM{./qc-axioms/q-left}=\tikzfigM{./qc-axioms/q-right}\end{align}
    \end{subfigure}

    \hspace{-2em}\begin{subfigure}{0.92\textwidth}
      \begin{align}\label{Mstar}\tag{$\text{K}^*$}\tikzfigM{./qc-axioms/Mstar-left}=\tikzfigM{./qc-axioms/Mstar-right-simp}\end{align}
    \end{subfigure}

    \vspace{.5em}
  \end{center}\end{minipage}}}
  \caption{Equational theory $\QC$. Equations \eqref{axQCgphaseaddition} and \eqref{axQCcnotPcnot} are defined for any $\varphi,\varphi_1,\varphi_2\in\R$. In Equations \eqref{axQCeuler} and \eqref{Mstar} the LHS circuit has arbitrary parameters which uniquely determine the parameters of the RHS circuit. \cref{axQCeuler} follows from the well-known Euler-decomposition which states that any unitary can be decomposed, up to a global phase, into basic $X$- and $Z$-rotations. Thus for any $\alpha_i\in \mathbb R$, there exist $\beta_j\in \mathbb R$ such that \cref{axQCeuler} is sound. We make the angles $\beta_j$ unique by assuming that $\beta_1 \in [0,\pi)$, $\beta_0,\beta_2,\beta_3\in[0,2\pi)$ and if $\beta_2\in\{0,\pi\}$ then $\beta_1=0$. \cref{Mstar} reads as follows: the equation is defined for any $n\ge 2$ input qubits, in such a way that all gates are controlled by the first $n-2$ qubits.  Similarly to \cref{axQCeuler}, for any $\gamma_i\in\mathbb R$, there exist $\delta_j\in\mathbb R$ such that \cref{Mstar} is sound. We ensure that the angles $\delta_j$ are uniquely determined by assuming that $\delta_1,\delta_2,\delta_5\in[0,\pi)$, $\delta_3,\delta_6,\delta_7,\delta_8\in[0,2\pi)$, $\delta_4\in[0,4\pi)$, if $\delta_3=0$ and $\delta_6\neq0$ then $\delta_2=0$, if $\delta_3=\pi$ then $\delta_1=0$, if $\delta_4\in\{0,2\pi\}$ then $\delta_1=\delta_3=0$, if $\delta_4\in\{\pi,3\pi\}$ then $\delta_2=0$, if $\delta_4\in\{\pi,3\pi\}$ and $\delta_3=0$ then $\delta_1=0$, and if $\delta_6\in\{0, \pi\}$ then $\delta_5=0$.\label{fig:QCaxioms}}
\end{figure}

\begin{figure}[h!]
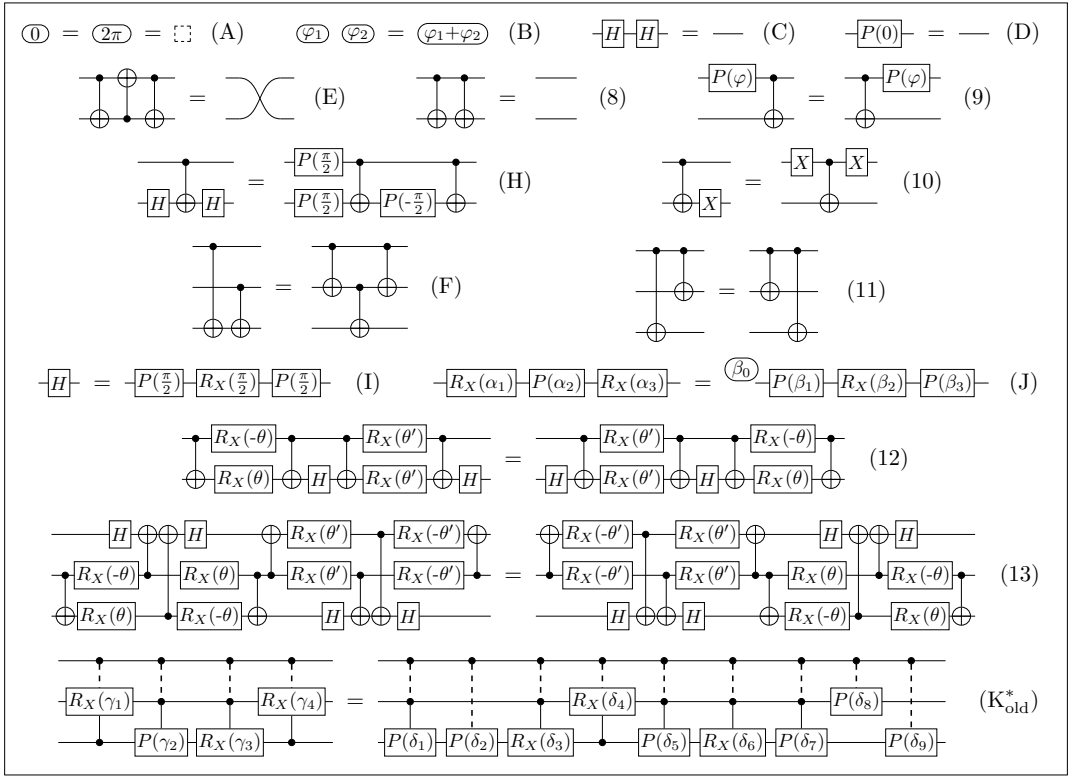

  \scalebox{.85}{\fbox{\begin{minipage}{1.159\textwidth}\begin{center}
    \vspace{-1em}
    
    \hspace{-2.5em}\begin{subfigure}{0.26\textwidth}
      \begin{align}\tag{A}\tikzfigM{./qc-axioms/s0}=\tikzfigM{./qc-axioms/s2pi}=\tikzfigM{gates/empty}\end{align}
    \end{subfigure}\hspace{-.2em}
    \begin{subfigure}{0.28\textwidth}
      \begin{align}\tag{B}\tikzfigM{./qc-axioms/sphi1}\tikzfigM{./qc-axioms/sphi2}=\tikzfigM{./qc-axioms/sphi1plusphi2}\end{align}
    \end{subfigure}\hspace{-.2em}
    \begin{subfigure}{0.24\textwidth}
      \begin{align}\tag{C}\tikzfigM{./qc-axioms/HH}=\tikzfigM{./gates/Id}\end{align}
    \end{subfigure}\hspace{-.2em}
    \begin{subfigure}{0.23\textwidth}
      \begin{align}\tag{D}\tikzfigM{./qc-axioms/P0}=\tikzfigM{gates/Id}
      \end{align}
    \end{subfigure}
    \vspace{-.5em}

    \hspace{-2em}\begin{subfigure}{0.30\textwidth}
      \begin{align}\tag{E}\tikzfigM{./qc-axioms/CNOT12CNOT21CNOT12_2}=\tikzfigM{./qc-axioms/SWAP_2}\end{align}
    \end{subfigure}\hspace{.5em}
    \begin{subfigure}{0.25\textwidth}
      \begin{align}\label{CNOTCNOT}\tikzfigM{./qc-completeness/CNOTCNOT}=\tikzfigM{./qc-completeness/IdId}\end{align}
    \end{subfigure}\hspace{.5em}
    \begin{subfigure}{0.33\textwidth}
      \begin{align}\label{PcommutCNOT}\tikzfigM{./qc-completeness/PphiCNOT}=\tikzfigM{./qc-completeness/CNOTPphi}\end{align}
    \end{subfigure}
    \vspace{-.5em}

    \hspace{-2em}\begin{subfigure}{0.43\textwidth}
      \begin{align}\tag{H}\tikzfigM{./qc-axioms/H2CNOTH2}=\tikzfigM{./qc-axioms/CZ}\end{align}
    \end{subfigure}\hspace{2.5em}
    \begin{subfigure}{0.33\textwidth}
      \begin{align}\label{CNOTXX}\tikzfigM{./qc-completeness/XCNOTXX-step-0}=\tikzfigM{./qc-completeness/XCNOTXX-step-15}\end{align}
    \end{subfigure}
    \vspace{-.5em}

    \hspace{-2em}\begin{subfigure}{0.31\textwidth}
      \begin{align}\tag{F}\tikzfigM{./qc-axioms/CNOT13CNOT23}=\tikzfigM{./qc-axioms/CNOT12CNOT23CNOT12}\end{align}
    \end{subfigure}\hspace{4.5em}
    \begin{subfigure}{0.30\textwidth}
      \begin{align}\label{CNOTscontrolcommut}\tikzfigM{./qc-completeness/CNOTscontrolcommut-v2-step-0}=\tikzfigM{./qc-completeness/CNOTscontrolcommut-v2-step-3}\end{align}
    \end{subfigure}
    \vspace{-.5em}

    \hspace{-2em}\begin{subfigure}{0.37\textwidth}
      \begin{align}\tag{I}\tikzfigM{./qc-axioms/H_2}=\tikzfigM{./qc-axioms/eulerH_2}\end{align}
    \end{subfigure}\hspace{0em}
    \begin{subfigure}{0.62\textwidth}
      \begin{align}\tag{J}\tikzfigM{./qc-axioms/q-left}=\tikzfigM{./qc-axioms/q-right}\end{align}
    \end{subfigure}

    \hspace{-2em}\begin{subfigure}{0.75\textwidth}
      \begin{align}\label{n}\tikzfigM{./qc-axioms/n-left}=\tikzfigM{./qc-axioms/n-right}\end{align}
    \end{subfigure}

    \hspace{-2em}\begin{subfigure}{1.00\textwidth}
      \begin{align}\label{o}\tikzfigM{./qc-axioms/o-left}=\tikzfigM{./qc-axioms/o-right}\end{align}
    \end{subfigure}

    \hspace{-3em}\begin{subfigure}{1.02\textwidth}
      \begin{align}\label{Mstarold}\tag{$\text{K}^*_{\text{old}}$}\tikzfigM{./qc-axioms/Mstar-left}=\tikzfigM{./qc-axioms/Mstar-right}\end{align}
    \end{subfigure}

    \vspace{.5em}
  \end{center}\end{minipage}}}
  \caption{Equational theory $\QCold$ introduced in \cite{CHMPV}.
  Equations \eqref{axQCgphaseempty},\eqref{axQCgphaseaddition},\eqref{axQCHH},\eqref{axQCP0},\eqref{axQCswap},\eqref{axQCCZ},\eqref{axQC3cnots},\eqref{axQCHeuler} are \eqref{axQCeuler} are aslo in the equational theory $\QC$. \cref{Mstarold} is the old version of \cref{Mstar} with one more parameter, and where the uniqueness of the parameters $\delta_j$ is given by the conditions: $\delta_1,\delta_2,\delta_5\in[0,\pi)$, $\delta_3,\delta_4,\delta_6,\delta_7,\delta_8,\delta_9\in[0,2\pi)$, if $\delta_3=0$ then $\delta_2=0$, if $\delta_3=\pi$ then $\delta_1=0$, if $\delta_4=0$ then $\delta_1=\delta_3\mathrel{(=}\delta_2)=0$, if $\delta_4=\pi$ then $\delta_2=0$, if $\delta_4=\pi$ and $\delta_3=0$ then $\delta_1=0$, and if $\delta_6\in\{0, \pi\}$ then $\delta_5=0$. Note that these conditions on the $\delta_j$ for $1\leq j\leq8$ are the same as in \cref{Mstar} except for $\delta_4$, which is restricted to be in $[0,2\pi)$ instead of $[0,4\pi)$, and for $\delta_2$, which has to be $0$ when $\delta_3=0$ even if $\delta_6=0$.\label{fig:QColdaxioms}}
\end{figure}

We equip the vanilla quantum circuits with the equational theory  $\QC$ defined in \cref{fig:QCaxioms}. We write $\QC\vdash C_1=C_2$ when $C_1$ can be transformed into $C_2$ using the equations of $\QC$. More formally, $\QC\vdash\cdot=\cdot$ is the smallest congruence which satisfies the equations of \cref{fig:QCaxioms} together with the deformation rules that come with the prop formalism. $\QC$ is sound, i.e. for any $\propQC$-circuits $C_1,C_2$ if $\QC\vdash C_1=C_2$ then $\interp{C_1}=\interp{C_2}$. This can be proved by observing that all equations of $\QC$ are sound.

\cref{fig:QColdaxioms} depicts the complete equational theory $\QCold$ for vanilla quantum circuits introduced in \cite{CHMPV}. Compared to $\QCold$, Equations \eqref{CNOTCNOT} and \eqref{PcommutCNOT} are now subsumed by \cref{axQCcnotPcnot} in $\QC$, \cref{Mstar} is a slight simplification of \cref{Mstarold} with one less parameter in the RHS circuit, whereas Equations \eqref{CNOTXX} and \eqref{CNOTscontrolcommut} together with Equations \eqref{n} and \eqref{o} have been removed, as we prove in the following that they can be derived in QC.

\subsection{Reasoning on quantum circuits}
To derive an equation $C_1=C_2$ over quantum circuits, one can apply some rules of the equational theory to transform step by step $C_1$ into $C_2$. In the context of vanilla quantum circuits, we can take advantage of the reversibility of generators to \emph{simplify} equations. Indeed, intuitively, proving $C_1\circ \gH = C_2 \circ \gH$ is equivalent to proving $C_1=C_2$ as $\gH$ is (provably) reversible. Similarly, proving $C_1 = C_2$ should be equivalent to proving $C_1\circ C_2^\dagger = \gI$, where the adjoint of a circuit is defined as follows:

\begin{definition}
  For any $\propQC$-circuit $C$, let $C^\dagger$ be the \emph{adjoint} of $C$ inductively defined as $(C_2\circ C_1)^\dagger\defeq C_1^\dagger\circ C_2^\dagger$; $(C_1\otimes C_2)^\dagger\defeq C_1^\dagger\otimes C_2^\dagger$; and for any $\varphi\in\R$, $\left(\gs\right)^\dagger\defeq\gsminus$, $\left(\gP\right)^\dagger\defeq\gPminus$, and $g^\dagger\defeq g$ for any other generator $g$. 
\end{definition}

\begin{proposition}
  $\interp{C^\dagger}=\interp{C}^\dagger$ for any $\propQC$-circuit $C$, where $\interp{C}^\dagger$ is the usual  linear algebra adjoint of $\interp{C}$. 
\end{proposition}
\begin{proof}
  By induction on $C$.
\end{proof}

\begin{proposition}[Simplification principle]\label{prop:crossinggate}
  For any $n$-qubit $\propQC$-circuits $C,C_1,C_2$
  \[\QC\vdash C\circ C_1 = C_2 \qquad \Leftrightarrow \qquad \QC\vdash C_1 =  C^\dagger \circ C_2\] and \[\QC\vdash C_1 \circ C = C_2 \qquad \Leftrightarrow \qquad \QC\vdash C_1 =  C_2\circ C^\dagger \]
\end{proposition}
\begin{proof}
  First we show by induction that $\QC\vdash C\circ C^\dagger=\gI^{\otimes n}$ and $\QC\vdash C^\dagger \circ C=\gI^{\otimes n}$ for any $C$. Then, w.l.o.g. we show that $(\QC\vdash C\circ C_1 = C_2)  \Rightarrow  (\QC\vdash C_1 =  C^\dagger \circ C_2)$: assuming $\QC\vdash C\circ C_1 = C_2$, we have $\QC\vdash C_1= C^\dagger \circ C\circ C_1 = C^\dagger \circ C_2$. 
\end{proof}

\subsection{Completeness}
We prove the completeness of $\QC$ by showing that every equation of the original complete equational theory $\QCold$ introduced in \cite{CHMPV} can be derived in $\QC$. To this end we first show the completeness of $\QC$ for the (modest) fragment of quantum circuits containing at most one CNot gate.

\begin{lemma}[1-CNot completeness]
  \label{lem:1CNOTcompleteness}
  $\QC$ is complete for circuits containing at most one $\gCNOT$, i.e. for any $\propQC$-circuits $C_1,C_2$ with at most one $\gCNOT$, if $\interp{C_1}=\interp{C_2}$ then $\QC\vdash C_1=C_2$.
\end{lemma}
\begin{proof}
  First we can show that, for semantic reasons, it is enough to prove the statement for 2-qubit circuits containing no swap gate and exactly one CNot gate. Then, by the simplification principle (\Cref{prop:crossinggate}), it is sufficient to prove $$\tikzfigS{./proof-n/ABCD}=\tikzfigS{./proof-n/CNOT}$$ whenever the equation is sound. By semantic analysis, we can show that there exist $\alpha,\beta,\gamma,\varphi,\theta\in\R$ and $k,\ell\in\{0,1\}$ such that
  \begin{gather*}
    \interp{\tikzfigS{./proof-n/A}}=\interp{\tikzfigS{./proof-n/semA}} \hspace{4em}
    \interp{\tikzfigS{./proof-n/C}}=\interp{\tikzfigS{./proof-n/semC}} \\
    \interp{\tikzfigS{./proof-n/B}}=\interp{\tikzfigS{./proof-n/semB}} \hspace{4em}
    \interp{\tikzfigS{./proof-n/D}}=\interp{\tikzfigS{./proof-n/semD}}
  \end{gather*}
  where \tikzfigS{./proof-n/Xk} (resp. \tikzfigS{./proof-n/Zl}) denotes \tikzfigS{./gates/X} (resp. \tikzfigS{./gates/Z}) if $k=1$ (resp. $\ell=1$) and \tikzfigS{./gates/Id} if $k=0$ (resp. $\ell=0$). Then, using the completeness of $\QC$ for one-qubit circuits (which is a direct consequence of the fact that all equations acting on at most one qubit of $\QCold$ are also in $\QC$), it is straightforward to verify that Equations~\eqref{axQCgphaseempty},\eqref{axQCgphaseaddition},\eqref{ZZ},\eqref{XX},\eqref{CNOTZZ},\eqref{CNOTXX},\eqref{PcommutCNOT},\eqref{RXcommutCNOT},\eqref{Paddition},\eqref{RXaddition},\eqref{axQCP0}, and \eqref{RX0} capture all the possible cases. The details are given in \cref{appendix:1CNOTcompleteness}.
\end{proof}

\begin{proposition}\label{prop:proof-n}
  \cref{n} can be derived in $\QC$.
\end{proposition}
\begin{proof}
  Using the simplification principle (\cref{prop:crossinggate}), one can turn \cref{n} into an equivalent equation whose circuits contain only one $\gCNOT$. The derivation is given in \cref{appendix:n}. We conclude the proof using the completeness of $\QC$ for circuits containing at most one CNot (\cref{lem:1CNOTcompleteness}).
\end{proof}

\begin{proposition}\label{prop:proof-o}
  Equation \eqref{o} can be derived in $\QC$.
\end{proposition}
\begin{proof}
  It turns out that we can use Equation \eqref{n} to derive Equation \eqref{o} in $\QC$. The derivation is given in \cref{appendix:o}.
\end{proof}

\begin{proposition}\label{prop:proof-Mstarold}
  Equation \eqref{Mstarold} can be derived in $\QC$.
\end{proposition}
\begin{proof}
  We show that for semantic reasons, we have either the angle $\delta_9$ in \eqref{Mstarold}  in $\{0,\pi\}$, or $\delta_2=\delta_3=\delta_5=\delta_6=0$. When $\delta_9=0$,  \cref{Mstarold} can be trivially derived from \cref{Mstar}. Otherwise, Equations \eqref{Mstar} and \eqref{Mstarold} can be transformed  into each other using elementary properties of multi-controlled gates. Moreover, these transformations induce a bijection between the $8$-tuples of angles $\delta_j$ corresponding to the RHS of the instances of \cref{Mstar} and the $9$-tuples corresponding to the RHS of the instances of \cref{Mstarold}, so that the uniqueness of the $\delta_j$ in \cref{Mstar} follows from the uniqueness in \cref{Mstarold}. Details are given in Appendix~\ref{appendix:Mstarold}.
\end{proof}

\begin{theorem}[Completeness]\label{thm:vanilla}
  The equational theory $\QC$, defined in \cref{fig:QCaxioms}, is complete for $\propQC$-circuits, i.e.~for any $\propQC$-circuits $C_1,C_2$, if $\interp{C_1}=\interp{C_2}$ then $\QC\vdash C_1=C_2$.
\end{theorem}
\begin{proof}
  All the rules of the complete equational theory introduced in \cite{CHMPV} that are not in $\QC$ are provable in $\QC$: Equations \eqref{CNOTCNOT}, \eqref{PcommutCNOT}, \eqref{CNOTXX}, \eqref{CNOTscontrolcommut} are proved in \cref{appendix:proofsimpleaxioms}, Equations \eqref{n}, \eqref{o} and \eqref{Mstarold} are proved in Propositions \ref{prop:proof-n}, \ref{prop:proof-o} and \ref{prop:proof-Mstarold} respectively.
\end{proof}

%%%%%%%%%%%%%%%%%%%%%%%%%%%%%%%%%%%%%%%%%%%%%%%%%%%%%%%%%%%%%%%%%%%%%%%%%%%%%%%
%%%%%%%%%%%%%%%%%%%%%%%%%%%%%%%%%%%%%%%%%%%%%%%%%%%%%%%%%%%%%%%%%%%%%%%%%%%%%%%
\section{Quantum circuits for isometries}
\label{sec:QCiso}
In this section we consider a first standard extension of the vanilla quantum circuits which consists in allowing qubit initialisation in a specific state, namely in the $\ket0$-state. 

\begin{definition}
  Let $\propQCiso$ be the prop generated by $\gs:0\to 0$, $\gH:1\to 1$, $\gP:1\to 1$, $\gCNOT:2\to 2$ and $\ginit:0\to 1$ for any $\varphi\in\R$.
\end{definition}

\begin{definition}[Semantics]\label{def:IsoSem}
  We extend the semantics $\interp \cdot $ of vanilla quantum circuits (Definition \ref{def:QCsem}) with $\interp \ginit = \ket 0$. 
\end{definition}

\begin{proposition}[Universality]
  Any isometry\footnote{An isometry is a linear map $V$ s.t. $V^\dagger \circ V$ is the identity.}  $V:\mathbb C^{\{0,1\}^n} \to \mathbb C^{\{0,1\}^m}$ can be realised by a $\propQCiso$-circuit $C:n\to m$ s.t. $\interp C = V$. 
\end{proposition}

\begin{figure}[h!]
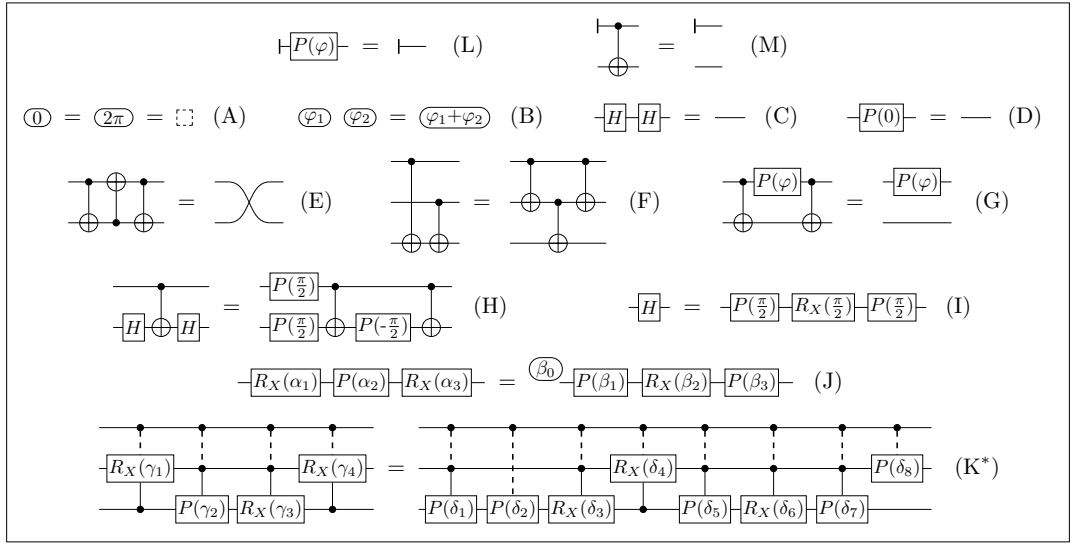

  \scalebox{.85}{\fbox{\begin{minipage}{1.159\textwidth}\begin{center}
    \vspace{-1em}

    \hspace{-2.8em}\begin{subfigure}{0.25\textwidth}
      \begin{align}\label{axQCISOinitP}\tag{L}\tikzfigM{./qciso-axioms/initPphi}=\tikzfigM{./qciso-axioms/initId}\end{align}
    \end{subfigure}\hspace{2em}
    \begin{subfigure}{0.24\textwidth}
      \begin{align}\label{axQCISOinitcnot}\tag{M}\tikzfigM{./qciso-axioms/initCNOT}=\tikzfigM{./qciso-axioms/initIdId}\end{align}
    \end{subfigure}
    \vspace{-.5em}
    
    \hspace{-2.5em}\begin{subfigure}{0.26\textwidth}
      \begin{align}\tag{A}\tikzfigM{./qc-axioms/s0}=\tikzfigM{./qc-axioms/s2pi}=\tikzfigM{gates/empty}\end{align}
    \end{subfigure}\hspace{-.2em}
    \begin{subfigure}{0.28\textwidth}
      \begin{align}\tag{B}\tikzfigM{./qc-axioms/sphi1}\tikzfigM{./qc-axioms/sphi2}=\tikzfigM{./qc-axioms/sphi1plusphi2}\end{align}
    \end{subfigure}\hspace{-.2em}
    \begin{subfigure}{0.24\textwidth}
      \begin{align}\tag{C}\tikzfigM{./qc-axioms/HH}=\tikzfigM{./gates/Id}\end{align}
    \end{subfigure}\hspace{-.2em}
    \begin{subfigure}{0.23\textwidth}
      \begin{align}\tag{D}\tikzfigM{./qc-axioms/P0}=\tikzfigM{gates/Id}
      \end{align}
    \end{subfigure}
    \vspace{-.5em}

    \hspace{-2em}\begin{subfigure}{0.30\textwidth}
      \begin{align}\tag{E}\tikzfigM{./qc-axioms/CNOT12CNOT21CNOT12}=\tikzfigM{./qc-axioms/SWAP}\end{align}
    \end{subfigure}\hspace{-.2em}
    \begin{subfigure}{0.31\textwidth}
      \begin{align}\tag{F}\tikzfigM{./qc-axioms/CNOT13CNOT23}=\tikzfigM{./qc-axioms/CNOT12CNOT23CNOT12}\end{align}
    \end{subfigure}\hspace{-.2em}
    \begin{subfigure}{0.33\textwidth}
      \begin{align}\tag{G}\tikzfigM{./qc-axioms/CNOTPphiCNOT}=\tikzfigM{./qc-axioms/PphiId}\end{align}
    \end{subfigure}
    \vspace{-.5em}

    \hspace{-2em}\begin{subfigure}{0.43\textwidth}
      \begin{align}\tag{H}\tikzfigM{./qc-axioms/H2CNOTH2}=\tikzfigM{./qc-axioms/CZ}\end{align}
    \end{subfigure}\hspace{3em}
    \begin{subfigure}{0.37\textwidth}
      \begin{align}\tag{I}\tikzfigM{./qc-axioms/H}=\tikzfigM{./qc-axioms/eulerH}\end{align}
    \end{subfigure}
    \vspace{-.5em}

    \hspace{-2em}\begin{subfigure}{0.62\textwidth}
      \begin{align}\tag{J}\tikzfigM{./qc-axioms/q-left}=\tikzfigM{./qc-axioms/q-right}\end{align}
    \end{subfigure}

    \hspace{-2em}\begin{subfigure}{0.92\textwidth}
      \begin{align}\tag{$\text{K}^*$}\tikzfigM{./qc-axioms/Mstar-left}=\tikzfigM{./qc-axioms/Mstar-right-simp}\end{align}
    \end{subfigure}

    \vspace{.5em}
  \end{center}\end{minipage}}}
  \caption{Equational theory $\QCiso$. It contains all the equations of $\QC$ together with \cref{axQCISOinitP} (defined for any $\varphi\in\R$) and \cref{axQCISOinitcnot}, which are new equations governing the behaviour of the new generator $\ginit$. \label{fig:QCisoaxioms}}
\end{figure}

For instance, the so-called copies in the standard basis ($\ket x\mapsto \ket {xx}$) and in the diagonal basis can be respectively represented as follows:
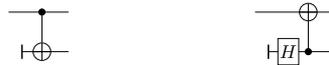
\[\tikzfigS{./examples/stdcopy}\qquad\qquad\qquad \tikzfigS{./examples/diagcopy}\]

We consider the equational theory $\QCiso$, given in Figure \ref{fig:QCisoaxioms}, which is nothing but the equational theory $\QC$ augmented with the following two sound equations:
\begin{gather*}
  \tikzfigS{./qciso-axioms/initPphi}=\tikzfigS{./qciso-axioms/initId} \;\;(\textup{L})\hspace{5em}
  \tikzfigS{./qciso-axioms/initCNOT}=\tikzfigS{./qciso-axioms/initIdId} \;\;(\textup{M})
\end{gather*}

\label{sec:comp-iso}

Viewing $\gP$ as a control-global-phase gate, Equations \eqref{axQCISOinitP}, \eqref{axQCISOinitcnot} can be interpreted as instances of the following property: a control gate can be removed when one of its control qubit is initialised in the $\ket 0$-state. This kind of properties can actually be derived within $\QCiso$ (see \cref{appendix:ancillaidentities}).

\begin{lemma}
\label{lem:InitCtrl}
Let $C$ be a $\propQCiso$-circuit  such that $\forall \ket \varphi \in \C^{2^n}$, $\interp{C}\ket \varphi = \ket0\otimes \ket \varphi$. Then:
\[\QCiso\vdash\tikzfigS{./qciso-axioms/initCwhatev1} ~=~ \tikzfigS{./qciso-axioms/initCwhatev2}\]
\end{lemma}

\begin{proof}
We prove that the above circuit necessarily is a $\propQC$-circuit together with a single qubit initialisation. The semantics of the $\propQC$-circuit forces it to be equivalent to a controlled circuit, which can be shown to be deletable by the qubit initialisation, thanks to Equations \eqref{axQCISOinitP} and \eqref{axQCISOinitcnot}. 
The complete proof is in \Cref{appendix:proofInitCtrl}.
\end{proof}

A direct corollary of \cref{lem:InitCtrl} is the completeness of $\QCiso$ for quantum circuits with at most one initialisation. Notice that one can then use Lemma 17 of \cite{Staton2015algebraic} to essentially prove the completeness of $\QCiso$. However, as the semantics in \cite{Staton2015algebraic} is based on CPTP maps rather than isometries (so global phases should be treated carefully), and moreover the proof of this Lemma 17 is not described, we provide a direct completeness proof of $\QCiso$ in the following. 

To do so, we may want to generalise \cref{lem:InitCtrl} to any number of qubit initialisations. However, the proof does not generalise. Indeed, 
it relies on the fact that, semantically, the vanilla circuit of which we initialize a single qubit is necessarily of the form $\operatorname{diag}(I,U)$, with $I$ and $U$ of the same dimension, so we can start with a circuit implementing $U$ and control each of its gates to get a circuit implementing $\operatorname{diag}(I,U)$ with only controls and phases on the control wire. To generalise this notion to more than one qubit initialisation, where semantically we would need to implement $\operatorname{diag}(I,U)$ with $U$ of dimensions larger than $I$'s, we need a finer-grain decomposition of said matrix. We hence resort to the following unitary decomposition:

\begin{lemma}
\label{lem:CSD}
Let $U=\left(\begin{array}{c|c|c}I&0&0\\\hline 0 & U_{00} & U_{01}\\\hline 0&U_{10} & ~~\raisebox{-0.5em}{\vphantom{\rule{1pt}{2em}}}U_{11}~~\end{array}\right)\begin{array}{l}\}k\\\}n-k\\\Bigl\} n\end{array}$
be unitary with $U_{00}$ and $U_{11}$ square. Then, there exist:
\begin{itemize}
\item[$\bullet$] $A_0$, $A_1$, $B_0$, $B_1$ unitary,
\item[$\bullet$] $C = \operatorname{diag}(c_1,...,c_{d})$ and $S = \operatorname{diag}(s_1,...,s_{d})$ ($c_i,s_i\geq0$ and $d\leq n-k$).
\end{itemize}
 such that:
 \begin{itemize}
 \item[$\bullet$] $C^2+S^2 = I$
 \item[$\bullet$] $U = \left(\begin{array}{c|c|c}I&0&0\\\hline 0 & A_0 & 0\\\hline 0&0 & ~~\raisebox{-0.5em}{\vphantom{\rule{1pt}{2em}}}A_1~~\end{array}\right)
\left(\begin{array}{c|c|c|c}
I&0&0&0 \\ \hline
0&C&0&-S\\ \hline
0&0&I&0\\ \hline
0&S&0&C
\end{array}\right)
\left(\begin{array}{c|c|c}I&0&0\\\hline 0 & B_0 & 0\\\hline 0&0 & ~~\raisebox{-0.5em}{\vphantom{\rule{1pt}{2em}}}B_1~~\end{array}\right)$
 \end{itemize}
\end{lemma}

The above decomposition is a variation on the \emph{Cosine-Sine Decomposition} (CSD) \cite{Paige1994history}, which has already appeared to be useful in quantum circuit synthesis \cite{Shende2006Synthesis}.
\begin{proof}
  The proof itself is a variation of the proof for the usual CSD. It specifically involves the so-called RQ and SVD decompositions, which are introduced, alongside the full proof of the lemma, in \Cref{appendix:proofCSD}.
\end{proof}

It is then possible to show the completeness of $\QCiso$:

\begin{theorem}[Completeness]\label{thm:QCiso-completeness}
  The equational theory $\QCiso$, defined in Figure \ref{fig:QCisoaxioms}, is complete for $\propQCiso$-circuits.
\end{theorem}
\begin{proof}
  The proof goes by showing that deriving equality between two $\propQCiso$-circuits amounts to generalising \Cref{lem:InitCtrl} to any number of qubit initialisations, which is shown inductively using the above variation of the CSD. The full proof is in \Cref{appendix:proofQCiso-completeness}.
\end{proof}

%%%%%%%%%%%%%%%%%%%%%%%%%%%%%%%%%%%%%%%%%%%%%%%%%%%%%%%%%%%%%%%%%%%%%%%%%%%%%%%
%%%%%%%%%%%%%%%%%%%%%%%%%%%%%%%%%%%%%%%%%%%%%%%%%%%%%%%%%%%%%%%%%%%%%%%%%%%%%%%
\section{Quantum circuits with ancillae}
\label{sec:QCancilla}

In this section, we consider quantum circuits which are implementing unitary maps (or isometries) using ancillary qubits, a.k.a.~ancillae, as additional work space. To represent quantum circuits with ancillae, we not only need to be able to initialise fresh qubits, but also to release qubits when they become useless. Note that to  guarantee that the overall evolution is an isometry, one can only release a qubit in the $\ket 0$-state. 

To encompass the notion of ancillary qubits we extend $\propQCiso$-circuits  (already equipped with qubit initialisation $\ginit$) with a qubit removal generator denoted  $\gdest$. Because of the constraint that removed qubits must be in the $\ket 0$-state, we define the language of quantum circuits with ancillae in two steps. 

\begin{definition}
  Let $\propQCpreancilla$ be the prop generated by $\gs:0\to 0$, $\gH:1\to 1$, $\gP:1\to 1$, $\gCNOT:2\to 2$, $\ginit:0\to 1$ and $\gdest:1\to 0$ for any $\varphi\in\R$.
\end{definition}

\begin{definition}[Semantics]
  We extend the semantics $\interp \cdot $ of quantum circuits for isometries (Definition \ref{def:IsoSem}) with $\interp \gdest = \bra 0$.
\end{definition}

Notice that the semantics of a $\propQCpreancilla$-circuit is not necessarily an isometry as $\interp \gdest$ is not isometric.\footnote{Actually any linear map $L$ s.t. $L^\dagger L\sqsubseteq I$ can be implemented by a $\propQCpreancilla$-circuit, where $\sqsubseteq$ is the L\"owner partial order. Thus $\propQCpreancilla$ can be seen as a language for postselected quantum computations.} As a consequence, we define  $\propQCancilla$ as the subclass of $\propQCpreancilla$-circuits with an isometric semantics: 

\begin{definition}
  Let $\propQCancilla$ be the sub-prop of $\propQCpreancilla$-circuit $C$ such that $\interp C$ is an isometry. 
\end{definition}

\begin{figure}[h!]
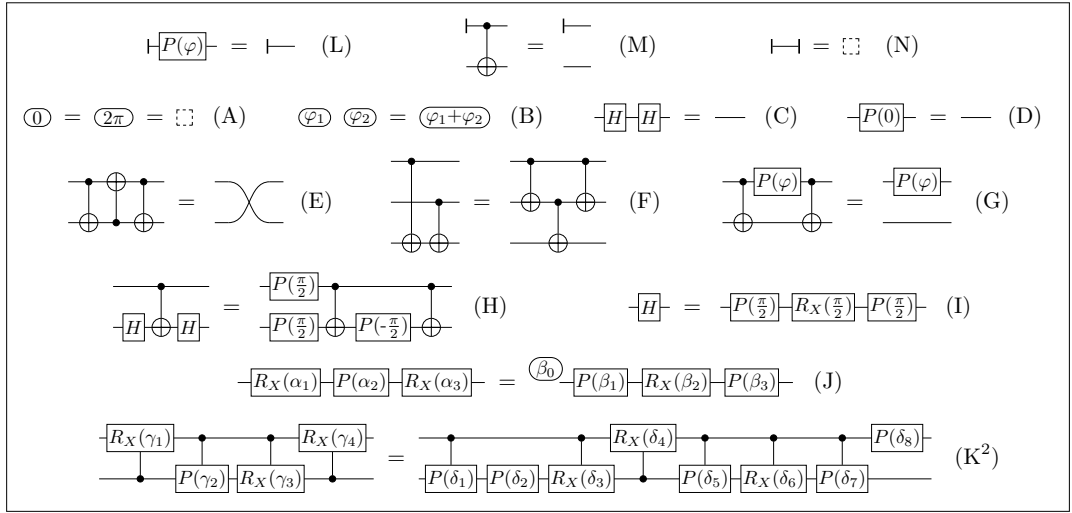

  \scalebox{.85}{\fbox{\begin{minipage}{1.159\textwidth}\begin{center}
    \vspace{-1em}

    \hspace{-2.8em}\begin{subfigure}{0.25\textwidth}
      \begin{align}\tag{L}\tikzfigM{./qciso-axioms/initPphi}=\tikzfigM{./qciso-axioms/initId}\end{align}
    \end{subfigure}\hspace{2em}
    \begin{subfigure}{0.24\textwidth}
      \begin{align}\tag{M}\tikzfigM{./qciso-axioms/initCNOT}=\tikzfigM{./qciso-axioms/initIdId}\end{align}
    \end{subfigure}\hspace{2em}
    \begin{subfigure}{0.20\textwidth}
      \begin{align}\label{axQCANCinitdest}\tag{N}\tikzfigM{./qcancilla-axioms/initdest}=\tikzfigM{./qcancilla-axioms/empty}\end{align}
    \end{subfigure}
    \vspace{-.5em}
    
    \hspace{-2.5em}\begin{subfigure}{0.26\textwidth}
      \begin{align}\tag{A}\tikzfigM{./qc-axioms/s0}=\tikzfigM{./qc-axioms/s2pi}=\tikzfigM{gates/empty}\end{align}
    \end{subfigure}\hspace{-.2em}
    \begin{subfigure}{0.28\textwidth}
      \begin{align}\tag{B}\tikzfigM{./qc-axioms/sphi1}\tikzfigM{./qc-axioms/sphi2}=\tikzfigM{./qc-axioms/sphi1plusphi2}\end{align}
    \end{subfigure}\hspace{-.2em}
    \begin{subfigure}{0.24\textwidth}
      \begin{align}\tag{C}\tikzfigM{./qc-axioms/HH}=\tikzfigM{./gates/Id}\end{align}
    \end{subfigure}\hspace{-.2em}
    \begin{subfigure}{0.23\textwidth}
      \begin{align}\tag{D}\tikzfigM{./qc-axioms/P0}=\tikzfigM{gates/Id}
      \end{align}
    \end{subfigure}
    \vspace{-.5em}

    \hspace{-2em}\begin{subfigure}{0.30\textwidth}
      \begin{align}\tag{E}\tikzfigM{./qc-axioms/CNOT12CNOT21CNOT12}=\tikzfigM{./qc-axioms/SWAP}\end{align}
    \end{subfigure}\hspace{-.2em}
    \begin{subfigure}{0.31\textwidth}
      \begin{align}\tag{F}\tikzfigM{./qc-axioms/CNOT13CNOT23}=\tikzfigM{./qc-axioms/CNOT12CNOT23CNOT12}\end{align}
    \end{subfigure}\hspace{-.2em}
    \begin{subfigure}{0.33\textwidth}
      \begin{align}\tag{G}\tikzfigM{./qc-axioms/CNOTPphiCNOT}=\tikzfigM{./qc-axioms/PphiId}\end{align}
    \end{subfigure}
    \vspace{-.5em}

    \hspace{-2em}\begin{subfigure}{0.43\textwidth}
      \begin{align}\tag{H}\tikzfigM{./qc-axioms/H2CNOTH2}=\tikzfigM{./qc-axioms/CZ}\end{align}
    \end{subfigure}\hspace{3em}
    \begin{subfigure}{0.37\textwidth}
      \begin{align}\tag{I}\tikzfigM{./qc-axioms/H}=\tikzfigM{./qc-axioms/eulerH}\end{align}
    \end{subfigure}
    \vspace{-.5em}

    \hspace{-2em}\begin{subfigure}{0.62\textwidth}
      \begin{align}\tag{J}\tikzfigM{./qc-axioms/q-left}=\tikzfigM{./qc-axioms/q-right}\end{align}
    \end{subfigure}

    \hspace{-2em}\begin{subfigure}{0.92\textwidth}
      \begin{align}\label{M2}\tag{$\text{K}^2$}\tikzfigM{./qcancilla-axioms/M2-left}=\tikzfigM{./qcancilla-axioms/M2-right-simp}\end{align}
    \end{subfigure}

    \vspace{.5em}
  \end{center}\end{minipage}}}
  \caption{Equational theory $\QCancilla$. It contains all the equations of $\QCiso$ where \cref{Mstar} has been replaced by \cref{M2}, together with \cref{axQCANCinitdest}, which is a new equation that allows one to create ancillae.\label{fig:QCancillaaxioms}}
\end{figure}

Notice in particular that any $\propQCiso$-circuit  is in $\propQCancilla$, which implies the universality of $\propQCancilla$ for isometries. We equip $\propQCancilla$-circuits with the equational theory $\QCancilla$ given in \cref{fig:QCancillaaxioms}, which is nothing but the equational theory $\QCiso$ where \cref{Mstar} is replaced by its 2-qubit version \cref{M2}, together with a new elementary equation \eqref{axQCANCinitdest} governing the behaviour of the qubit removal generator $\gdest$.
\begin{gather*}
  \tikzfigS{./qcancilla-axioms/initdest}=\tikzfigS{./qcancilla-axioms/empty} \;\;(\textup{N})
\end{gather*}

Quantum circuits with ancillae form a standard model of quantum computing. They are for instance used in the context of quantum oracles: given a circuit $C_f:n+1 \to n+1$ whose semantics is $\ket{x,y} \mapsto  \ket{x,y\oplus f(x)}$ for some boolean function $f$, one can implement the corresponding phase oracle $C'_f$ whose semantics is $\ket x\mapsto (-1)^{f(x)}\ket x$ as follows:
\[\tikzfigS{./examples/phase-oracleL}\defeq\tikzfigS{./examples/phase-oracle}\]

Quantum circuits with ancillae are also extensively used in the context of circuit parallelisation, as one can decrease the depth of a quantum circuit by adding ancillary qubits \cite{moore2001parallel}. Finally, ancillary qubits can be used to provide an alternative realisation of multi-controlled gates, for instance a 3-qubit controlled gate can be implemented using an ancillary qubit, Toffoli gates, and the 2-qubit version of the gate: 
\begin{equation}\label{eq:multi2}
  \tikzfigS{./qcancilla-completeness/initmctrlP-2}\;=\;\tikzfigS{./qcancilla-completeness/initmctrlPdef-2}
\end{equation}

This can be generalised to any multi-controlled gates with at least two control qubits: 

\begin{proposition}\label{prop:mctrlaltdef}
  The following two equations can be derived in $\QCancilla$.
  \vspace{-2em}\begin{multicols}{2}
    \begin{align}\label{Paltdef}\tikzfigS{./qcancilla-completeness/initmctrlP}\;=\;\tikzfigS{./qcancilla-completeness/initmctrlPdef}\end{align}

    \begin{align}\label{RXaltdef}\tikzfigS{./qcancilla-completeness/initmctrlRX}\;=\;\tikzfigS{./qcancilla-completeness/initmctrlRXdef}\end{align}
  \end{multicols}
\end{proposition}
\begin{proof}
  By induction on the number of qubits. The proof is given in Appendix \ref{appendix:proofmctrlaltdef}. 
\end{proof}

Notice that Equations \eqref{Paltdef} and \eqref{RXaltdef} are actually derivable in $\QCiso$. However, in order to provide an alternative inductive definition of multi-control gates (like in \cref{eq:multi2}), it requires the presence of at least one fresh qubit which can always be created in the context of quantum circuits with ancillae thanks to \cref{axQCANCinitdest}.

Thanks to the alternative representation of multi-controlled gates, one can derive, in $\QCancilla$, the equation \eqref{Mstar} for any arbitrary number of controlled qubits:

\begin{proposition}\label{prop:Kstar}
  Equation \eqref{Mstar} can be derived in $\QCancilla$.
\end{proposition}
\begin{proof}
  Let $(\text{K}^n)$ be \cref{Mstar} acting on $n$ qubits for any $n\ge2$. \cref{M2} is in $\QCancilla$. We first prove that \eqref{M3} can be derived from \eqref{M2} by defining the Fredkin gate (or controlled-swap gate) and by pushing the two last wires of the LHS circuit of \eqref{M3} into two fresh ancillae, which allow us to apply \eqref{M2} and reverse the construction to get the RHS circuit of \eqref{M3}. The detailed proof is given in \ref{sec:proof-M3} together with all necessary intermediate derivations. This technique is not applicable in the general case for any circuit because if the Fredkin gates are not triggered, it could be the case that the gates pushed into the ancillae do not release the ancillae into the $\ket{0}$-state. The key observation is that this is possible for \eqref{M3} as every involved gates are either phase gate or uniquely controlled gate (which both act as identity on the $\ket{0}$-state). Then, we prove that $(\text{K}^n)$ is derivable in $\QCancilla$ for any $n\ge4$ by induction on $n$ using the alternative definition of multi-controlled gates (Proposition \ref{prop:mctrlaltdef}), which allows us to construct an instance of the LHS circuit of $(\text{K}^{n-1})$ from the LHS circuit of $(\text{K}^n)$. The detailed proof is given in \ref{appendix:inductionMstar}.
\end{proof}

We are now ready to prove the completeness of $\QCancilla$:

\begin{theorem}[Completeness]
  The equational theory $\QCancilla$, defined in Figure \ref{fig:QCancillaaxioms}, is complete for $\propQCancilla$-circuits.
\end{theorem}
\begin{proof}
  \cref{prop:Kstar} implies that for any $\gdest\,$-free circuits $C_1$, $C_2$, if $\QCiso \vdash C_1=C_2$ then $\QCancilla \vdash C_1=C_2$. Using deformation of circuits, any $\propQCancilla$-circuit $C:n\to m $ can be written $\tikzfigS{./qcancilla-completeness/Cancilla}$, where $C':n\to m+k$ is a $\propQCiso$-circuit. Since both $\interp C$ and $\interp {C'}$ are isometries and $\interp C = (Id\otimes \bra{0^k})\interp{C'}$, we have $\interp {C'} = \interp C\otimes \ket{0^k}$.  Given two $\propQCancilla$-circuits $C_1$, $C_2$ s.t. $\interp {C_1}=\interp{C_2}$, let $C'_1:n\to m+k$, and $C'_2:n\to m+\ell$ be the corresponding $\propQCiso$-circuits. W.l.o.g.~assume $k\le \ell$, and pad $C'_1$ with $\ell-k$ qubit initialisations: $C''_1:=C'_1\otimes (\ginit)^{\otimes \ell-k}$. We have $\interp{C''_1} =\interp{C'_2}$, so by completeness of  $\QCiso$, $\QCancilla \vdash C''_1=C'_2$, so $\QCancilla \vdash C_1\otimes (\scalebox{0.7}{\begin{tikzpicture}
	\begin{pgfonlayer}{nodelayer}
		\node [style=ancilla] (1) at (0.5, 0) {};
		\node [style=ancilla] (2) at (-0.5, 0) {};
		\node [style=void] (3) at (-0.5, 0.75) {};
		\node [style=void] (4) at (-0.5, -0.75) {};
	\end{pgfonlayer}
	\begin{pgfonlayer}{edgelayer}
		\draw (2) to (1);
	\end{pgfonlayer}
\end{tikzpicture}
})^{\otimes \ell-k}=C_2$. It suffices to apply \cref{axQCANCinitdest} to obtain $\QCancilla \vdash C_1=C_2$.
\end{proof}

%%%%%%%%%%%%%%%%%%%%%%%%%%%%%%%%%%%%%%%%%%%%%%%%%%%%%%%%%%%%%%%%%%%%%%%%%%%%%%%
%%%%%%%%%%%%%%%%%%%%%%%%%%%%%%%%%%%%%%%%%%%%%%%%%%%%%%%%%%%%%%%%%%%%%%%%%%%%%%%
\section{Quantum circuits with discard for completely positive map}
\label{sec:QCdiscard}
The last extension considered in this paper is the addition of a discard operator which consists in tracing out qubits. Contrary to quantum circuits with ancillae, any qubit can be discarded whatever its state is. 
Discarding a qubit is depicted by $\gdiscard$. 

\begin{definition}
  Let $\propQCground$ be the prop generated by $\gH:1\to 1$, $\gP:1\to 1$, $\gCNOT:2\to 2$, $\ginit:0\to 1$ and $\gdiscard:1\to 0$ for any $\varphi\in\R$.
\end{definition}

The ability to discard qubits implies that the evolution represented by such a circuit is not pure anymore. As a consequence the semantics is a completely positive trace-preserving (CPTP) map acting on density matrices (trace 1 positive semi-definite Hermitian matrices). Formally the new semantics is defined as follows: 
\begin{definition}[Semantics]
  For any quantum $\propQCground$-circuit $C:n\to m$, let $\CPTP{C}: \mathcal M_{2^n,2^n}(\C) \to  \mathcal M_{2^m,2^m}(\C) $ be the \emph{semantics} of $C$ inductively defined as the linear map $\CPTP{C_2\circ C_1} = \CPTP{C_2}\circ\CPTP{C_1}$; $\CPTP{C_1\otimes C_2} = \CPTP{C_1}\otimes\CPTP{C_2}$; $\CPTP{\gdiscard} = \rho \mapsto tr(\rho)$ and for any other generator $g$, $\CPTP g = \rho \mapsto \interp g \rho \interp g^\dagger$, where $tr(M)$ is the trace of the matrix $M$ and $M^\dagger$ its adjoint. 
\end{definition}

Notice that the global phase generator $\gs$ is not part of the prop anymore. If it were, its interpretation would be $\CPTP{\gs} = \rho \mapsto \interp{\gs} \rho \interp{\gs}^\dagger = e^{i\varphi}\rho e^{-i\varphi}=\rho$, which is the same as that of the empty circuit. Thus, for this model the X-rotation can simply be defined as $\tikzfigS{./gates/RXtheta}\defeq\tikzfigS{./shortcut/HPthetaH-nogphase}$ (the same definition as \cref{fig:shortcutcircuits} but without the global phase).

\begin{proposition}[Universality]
  $\propQCground$ is universal for CPTP maps.
\end{proposition}
\begin{proof}
  According to the Stinespring dilation lemma \cite{stinespring1955positive}, any CPTP map $F:\mathcal M_{2^n,2^n}(\C) \to  \mathcal M_{2^m,2^m}(\C)$ can be purified as an isometry $V : \C^{2^n}\to \C^{2^{m+k}}$ such that for any $\rho$, $F(\rho) = tr_k(V\rho V^\dagger)$, where $tr_k(.)$ is the partial trace of the last $k$ qubits. By universality of $\propQCiso$ there exists a circuit $C$ such that $\interp C = V$. Let $C'$ be the global-phase-free version of $C$, thus $\interp {C'} = e^{i\theta}V$. Seen as a $\propQCground$-circuit, $C'$ has the semantics $\CPTP{C'} = \rho \mapsto (e^{i\theta}V)\rho(e^{i\theta}V)^\dagger =V\rho V^\dagger$. Discarding the last $k$ qubits of $C'$ leads to a $\propQCground$-circuit implementing $F$. 
\end{proof}

The new generator and new semantics allow us to model measurements. For instance, the standard basis measurement can be obtained via:
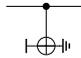
\[\tikzfigS{./qcground-axioms/measurement-}\]
Indeed we recover the semantics of the standard basis measurement: \[\CPTPleftright{ \tikzfigS{./qcground-axioms/measurement-}} = \left(\begin{array}{cc} a&c\\b&d\end{array}\right)\mapsto  \left(\begin{array}{cc} a&0\\0&d\end{array}\right)\]
The output wire can be interpreted as a classical bit (encoded in a quantum bit), $a$ (resp. $d$) being the probability to be $0$ (resp. $1$). 
 
One can also encode classical gates, for instance the AND gate using Toffoli:
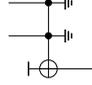
\[\tikzfigS{./qcground-axioms/And}\]
With the promise that the input is classical, i.e.~the input density matrix is $\operatorname{diag}(p_{00},p_{01},p_{10},p_{11})$ (where $p_{xy}$ is the probability for the input to be in the state $xy\in \{0,1\}^2$), the output state is $\operatorname{diag}(p_{00}+p_{01}+p_{10},p_{11})$ which corresponds to the behaviour of the AND gate.   

More generally, one can represent classically controlled computation using the $\propQCground$-circuits, allowing to reason on fault-tolerant computations, error correcting codes and measurement-based quantum computation for instance.

\begin{figure}[h!]
  \scalebox{.85}{\fbox{\begin{minipage}{1.159\textwidth}\begin{center}
    \vspace{-1em}

    \hspace{-2.5em}\begin{subfigure}{0.24\textwidth}
      \begin{align}\tag{L}\tikzfigM{./qcground-axioms/initPphi}=\tikzfigM{./qcground-axioms/initId}\end{align}
    \end{subfigure}\hspace{0em}
    \begin{subfigure}{0.24\textwidth}
      \begin{align}\label{axQCGRDHground}\tag{O}\tikzfigM{./qcground-axioms/Hdiscard}=\tikzfigM{./qcground-axioms/Iddiscard}\end{align}
    \end{subfigure}\hspace{0em}
    \begin{subfigure}{0.25\textwidth}
      \begin{align}\label{axQCGRDPground}\tag{P}\tikzfigM{./qcground-axioms/Pphidiscard}=\tikzfigM{./qcground-axioms/Iddiscard}\end{align}
    \end{subfigure}\hspace{0em}
    \begin{subfigure}{0.20\textwidth}
      \begin{align}\label{axQCGRDinitground}\tag{Q}\tikzfigM{./qcground-axioms/initdiscard}=\tikzfigM{./gates/empty}\end{align}
    \end{subfigure}
    \vspace{-.5em}

    \hspace{-2.5em}\begin{subfigure}{0.24\textwidth}
      \begin{align}\tag{M}\tikzfigM{./qciso-axioms/initCNOT}=\tikzfigM{./qciso-axioms/initIdId}\end{align}
    \end{subfigure}\hspace{0em}
    \begin{subfigure}{0.23\textwidth}
      \begin{align}\label{axQCGRDcnotground}\tag{R}\tikzfigM{./qcground-axioms/CNOTdiscard}=\tikzfigM{./qcground-axioms/IdIddiscarddiscard}\end{align}
    \end{subfigure}\hspace{0em}
    \begin{subfigure}{0.24\textwidth}
      \begin{align}\tag{C}\tikzfigM{./qcground-axioms/HH}=\tikzfigM{./gates/Id}\end{align}
    \end{subfigure}\hspace{0em}
    \begin{subfigure}{0.24\textwidth}
      \begin{align}\tag{D}\tikzfigM{./qcground-axioms/P0}=\tikzfigM{gates/Id}
      \end{align}
    \end{subfigure}
    \vspace{-.5em}

    \hspace{-2em}\begin{subfigure}{0.30\textwidth}
      \begin{align}\tag{E}\tikzfigM{./qc-axioms/CNOT12CNOT21CNOT12}=\tikzfigM{./qc-axioms/SWAP}\end{align}
    \end{subfigure}\hspace{-.2em}
    \begin{subfigure}{0.31\textwidth}
      \begin{align}\tag{F}\tikzfigM{./qc-axioms/CNOT13CNOT23}=\tikzfigM{./qc-axioms/CNOT12CNOT23CNOT12}\end{align}
    \end{subfigure}\hspace{-.2em}
    \begin{subfigure}{0.33\textwidth}
      \begin{align}\tag{G}\tikzfigM{./qc-axioms/CNOTPphiCNOT}=\tikzfigM{./qc-axioms/PphiId}\end{align}
    \end{subfigure}
    \vspace{-.5em}

    \hspace{-2em}\begin{subfigure}{0.43\textwidth}
      \begin{align}\tag{H}\tikzfigM{./qc-axioms/H2CNOTH2}=\tikzfigM{./qc-axioms/CZ}\end{align}
    \end{subfigure}\hspace{3em}
    \begin{subfigure}{0.37\textwidth}
      \begin{align}\tag{I}\tikzfigM{./qc-axioms/H}=\tikzfigM{./qc-axioms/eulerH}\end{align}
    \end{subfigure}
    \vspace{-.5em}

    \hspace{-2em}\begin{subfigure}{0.61\textwidth}
      \begin{align}\label{axQCGRDeuler}\tag{J'}\tikzfigM{./qcground-axioms/euler-left}=\tikzfigM{./qcground-axioms/euler-right}\end{align}
    \end{subfigure}

    \hspace{-2em}\begin{subfigure}{0.92\textwidth}
      \begin{align}\tag{$\text{K}^2$}\tikzfigM{./qcancilla-axioms/M2-left}=\tikzfigM{./qcancilla-axioms/M2-right-simp}\end{align}
    \end{subfigure}

    \vspace{.5em}
  \end{center}\end{minipage}}}
  \caption{Equational theory $\QCground$. It contains all the equations of $\QCancilla$ except Equations \eqref{axQCgphaseempty}, \eqref{axQCgphaseaddition}, \eqref{axQCANCinitdest} and where \cref{axQCeuler} has been replaced by its global-phase free-version \cref{axQCGRDeuler}, together with Equations \eqref{axQCGRDHground}, \eqref{axQCGRDPground} (defined for any $\varphi\in\R$), \eqref{axQCGRDinitground} and \eqref{axQCGRDcnotground}, which are new equations governing the behaviour of the new generator $\gdiscard$. \label{fig:QCgroundaxioms}}
\end{figure}

While \cite{Staton2015algebraic} provides a way to get completeness for quantum circuits with measurements from a complete one for isometries, we instead use \cite{Carette2021completeness} which provides a similar result but for isometries with discard, as the latter is a little bit more atomic than measurements. This leads us to equip $\propQCground$-circuits with the equational theory $\QCground$ defined in \Cref{fig:QCgroundaxioms}, which is a global-phase-free version of $\QCancilla$ where $\gdiscard$ replaces $\dashv$, and with the addition of:
\begin{gather*}
  \tikzfigS{./qcground-axioms/Hdiscard}=\tikzfigS{./qcground-axioms/Iddiscard} \;\;(\textup{O})\hspace{3em}
  \tikzfigS{./qcground-axioms/Pphidiscard}=\tikzfigS{./qcground-axioms/Iddiscard} \;\;(\textup{P})\hspace{3em}
  \tikzfigS{./qcground-axioms/initdiscard}=\tikzfigS{./gates/empty} \;\;(\textup{Q})\hspace{3em}
  \tikzfigS{./qcground-axioms/CNOTdiscard}=\tikzfigS{./qcground-axioms/IdIddiscarddiscard} \;\;(\textup{R})
\end{gather*}

This observation allows us in particular to transport all the proofs using $\QCancilla$ into the present theory, the only two differences being that $\gdiscard$ plays the role of $\dashv$ and that the $\QCground$ version of the proofs have no global phase $\gs$.

\begin{theorem}[Completeness]
  The equational theory $\QCground$, defined in \cref{fig:QCgroundaxioms}, is complete for $\propQCground$-circuits.
\end{theorem}
\begin{proof}
  We can use the \emph{discard construction} \cite{Carette2021completeness} to build $\QCiso^{\ground}$ from $\QCiso$, by adding equation:
  \begin{equation}\label{eq:discard-iso}
  \gdiscard^{\otimes m}\circ U = \gdiscard^{\otimes n}
  \end{equation}
  for any $\propQCiso$-circuit $U:n\to m$. The discard construction guarantees that $\QCiso^{\ground}$ is complete for CPTP maps (Proposition 2 in \cite{Carette2021completeness}). It remains to prove that all equations in $\QCiso^{\ground}$ derive from those of $\QCground$. All equations of the former except Equations \eqref{Mstar} and \eqref{eq:discard-iso} appear in $\QCground$. Those are trivially derivable. As mentioned above, it is possible to prove \eqref{Mstar} from $\QCground$ exactly as in the case of $\QCancilla$ by replacing each occurrence of $\tikzfigS{./qcground-axioms/initdest}=\tikzfigS{./qcground-axioms/empty}$ by $\tikzfigS{./qcground-axioms/initdiscard}=\tikzfigS{./gates/empty}$.
  This means all the equations of $\QCiso$ are derivable. Finally, all the equations $\gdiscard^{\otimes m}\circ U = \gdiscard^{\otimes n}$ for different isometries $U$ can be derived from Equations \eqref{axQCGRDHground}, \eqref{axQCGRDPground}, \eqref{axQCGRDinitground}, and \eqref{axQCGRDcnotground}.
\end{proof}

%%%%%%%%%%%%%%%%%%%%%%%%%%%%%%%%%%%%%%%%%%%%%%%%%%%%%%%%%%%%%%%%%%%%%%%%%%%%%%%
%%%%%%%%%%%%%%%%%%%%%%%%%%%%%%%%%%%%%%%%%%%%%%%%%%%%%%%%%%%%%%%%%%%%%%%%%%%%%%%
\section{Concluding remarks}
We have simplified the complete equational theory for quantum circuits, and provided ones for standard extensions of quantum circuits, including qubit initialisation, ancillae, and/or qubit discarding. The equational theory can be simplified in these more general settings, leading in particular to equations acting on a bounded number of qubits, avoiding the use of controlled gates on arbitrary number of qubits. It is interesting to notice that increasing the expressive power of the model makes the equational theory simpler.

This simplification of the equational theory is a step towards a minimal equational theory (i.e. an equational theory where each equation provably cannot be derived from the other ones). Notice that based on the present work, a minimal complete equational theory for vanilla quantum circuits has been introduced recently \cite{clement2023minimal}, showing in particular that equations acting on an unbounded number of qubits are necessary for vanilla quantum circuits. The question of the minimality in the context of quantum circuits with qubit initialisation, ancillae, and/or qubit discarding remains open. 

Getting rid of \cref{Mstar} eases also the practical implementation of the rewriting rules as it avoids to consider a family of rules acting on an unbounded of qubits. Notice that regarding practical considerations, various equations presented in this paper have parameters, e.g. $\tikzfigS{./qcground-axioms/euler-left}=\tikzfigS{./qcground-axioms/euler-right}$ that should be read as follows: for any angle $\alpha_i$ on the LHS, there exist $\beta_j$ on the RHS so that the equation holds. $\beta_j$ can be computed using fairly simple trigonometric operations. Notice that even if the equation looks non-symmetric, one can show conversely that for any $\beta_j$ there exist $\alpha_i$ angles such that the equation holds (see \cref{inverseEuler}).

%%%%%%%%%%%%%%%%%%%%%%%%%%%%%%%%%%%%%%%%%%%%%%%%%%%%%%%%%%%%%%%%%%%%%%%%%%%%%%%
%%%%%%%%%%%%%%%%%%%%%%%%%%%%%%%%%%%%%%%%%%%%%%%%%%%%%%%%%%%%%%%%%%%%%%%%%%%%%%%
\section*{Acknowledgements}
This work is supported by the Plan France 2030 through the PEPR integrated project EPiQ ANR-22-PETQ-0007 and the HQI initiative ANR-22-PNCQ-0002; it is also supported by the ANR project SoftQPro ANR-17-CE25-0009-02, by the STIC-AmSud project Qapla’ 21-STIC-10, and by the European projects NEASQC and HPCQS.

%%%%%%%%%%%%%%%%%%%%%%%%%%%%%%%%%%%%%%%%%%%%%%%%%%%%%%%%%%%%%%%%%%%%%%%%%%%%%%%
%%%%%%%%%%%%%%%%%%%%%%%%%%%%%%%%%%%%%%%%%%%%%%%%%%%%%%%%%%%%%%%%%%%%%%%%%%%%%%%
\bibliography{ref}

\appendix

\crefalias{section}{appendix}
\crefalias{subsection}{appendix}
\crefalias{subsubsection}{appendix}
\crefalias{paragraph}{appendix}
\crefalias{subparagraph}{appendix}

\section*{Appendix content}

\startcontents[sections]
\printcontents[sections]{l}{1}{\setcounter{tocdepth}{2}}

\section{Bullet-based graphical notation for multi-controlled gates}
\label{appendix:bullet}
We use the standard bullet-based graphical notation for multi-controlled gates where the \emph{negative control} (or \emph{anti-control}) \tikzfigS{./bulletnotation/w} is a shortcut notation for \tikzfigS{./bulletnotation/XbX}. For instance, \tikzfigS{./bulletnotation/bwPb} stands for the gate \tikzfigS{./gates/Pphi} on the third qubit positively controlled by the first and fourth qubits and negatively controlled by the second qubit. According to \cite{CHMPV} we can simulate the expected behaviour of this bullet-based notation in $\QC$ without using \cref{Mstar}.

Combining a control and anti-control on the same qubit makes the evolution independent of this qubit. This is provable in $\QC$ without \eqref{Mstar} and illustrated by the following example.
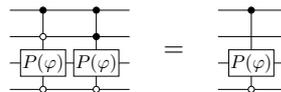
\begin{equation*}
  \begin{array}{rcl}
    \tikzfigS{./bulletnotation/bwPwbbPw}&=&\tikzfigS{./bulletnotation/bvPw}
  \end{array}
\end{equation*}

Another expected behaviour provable in $\QC$ without \eqref{Mstar} is the fact that controlled and anti-controlled gates commute (even if the target qubits are not the same in both gates). This is illustrated by the two following examples.
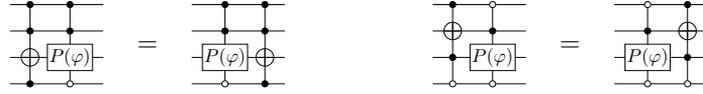
\begin{equation*}
  \begin{array}{rcl}
    \tikzfigS{./bulletnotation/bbXbbbPw}&=&\tikzfigS{./bulletnotation/bbPwbbXb}
  \end{array} \qquad\qquad \begin{array}{rcl}
    \tikzfigS{./bulletnotation/bXbwwbPw}&=&\tikzfigS{./bulletnotation/wbPwbXbw}
  \end{array}
\end{equation*}

In the following, the use of such properties is denoted by $(\begin{tikzpicture}
	\begin{pgfonlayer}{nodelayer}
		\node [style=wcontrol] (0) at (0, 0) {};
		\node [style=none] (1) at (0.25, 0) {};
		\node [style=none] (2) at (-0.25, 0) {};
		\node [style=none] (3) at (0, 0.25) {};
		\node [style=none] (4) at (0, -0.25) {};
	\end{pgfonlayer}
	\begin{pgfonlayer}{edgelayer}
		\draw (2.center) to (1.center);
		\draw (3.center) to (4.center);
	\end{pgfonlayer}
\end{tikzpicture})$ and refers to Propositions 15, 16 and 17 (together with Propositions 10 and 11 in some cases) of \cite{CHMPV}.

%%%%%%%%%%%%%%%%%%%%%%%%%%%%%%%%%%%%%%%%%%%%%%%%%%%%%%%%%%%%%%%%%%%%%%%%%%%%%%%
%%%%%%%%%%%%%%%%%%%%%%%%%%%%%%%%%%%%%%%%%%%%%%%%%%%%%%%%%%%%%%%%%%%%%%%%%%%%%%%
\section{Proofs of intermediate circuit equations}

\subsection{Proofs of Equations~\eqref{CNOTCNOT},\eqref{PcommutCNOT},\eqref{CNOTXX} and \eqref{CNOTscontrolcommut} of $\QCold$}\label{appendix:proofsimpleaxioms}

\begin{proof}[Proof of \cref{CNOTCNOT}]
  \begin{gather*}
    \tikzfigS{./qc-completeness/CNOTCNOT}\eqeqref{axQCP0}\tikzfigS{./qc-completeness/CNOTP0CNOT}\eqeqref{axQCcnotPcnot}\tikzfigS{./qc-completeness/P0Id}\eqeqref{axQCP0}\tikzfigS{./qc-completeness/IdId}
  \end{gather*}
\end{proof}

\begin{proof}[Proof of \cref{PcommutCNOT}]
  \begin{gather*}
    \tikzfigS{./qc-completeness/PphiCNOT}\eqeqref{CNOTCNOT}\tikzfigS{./qc-completeness/CNOTCNOTPphiCNOT}\eqeqref{axQCcnotPcnot}\tikzfigS{./qc-completeness/CNOTPphi}
  \end{gather*}
\end{proof}

\begin{proof}[Proof of \cref{CNOTscontrolcommut}]
  \begin{gather*}
    \tikzfigS{./qc-completeness/CNOTscontrolcommut-v2-step-0}
    \eqeqref{CNOTCNOT}\tikzfigS{./qc-completeness/CNOTscontrolcommut-v2-step-1}
    \eqeqref{axQC3cnots}\tikzfigS{./qc-completeness/CNOTscontrolcommut-v2-step-2}
    \eqeqref{CNOTCNOT}\tikzfigS{./qc-completeness/CNOTscontrolcommut-v2-step-3}
  \end{gather*}
\end{proof}

\begin{figure}[h!]
  \scalebox{.85}{\fbox{\begin{minipage}{1.159\textwidth}\begin{center}
    \vspace{-1em}
    
    \hspace{-2.5em}\begin{subfigure}{0.41\textwidth}
      \begin{align}\label{Paddition}\tikzfigM{./identities/Pphi1Pphi2}=\tikzfigM{./identities/Pphi1phi2}\end{align}
    \end{subfigure}\hspace{3em}
    \begin{subfigure}{0.37\textwidth}
      \begin{align}\label{XPX}\tikzfigM{./identities/XPphiX}=\tikzfigM{./identities/Pminusphi}\end{align}
    \end{subfigure}
    \vspace{-.5em}

    \hspace{-2.5em}\begin{subfigure}{0.25\textwidth}
      \begin{align}\label{XX}\tikzfigM{./identities/XX-step-0}=\tikzfigM{./identities/XX-step-3}\end{align}
    \end{subfigure}\hspace{3em}
    \begin{subfigure}{0.26\textwidth}
      \begin{align}\label{P2pi}\tikzfigM{./identities/P2pi}=\tikzfigM{./gates/Id}\end{align}
    \end{subfigure}
    \vspace{-.5em}

    \hspace{-2.5em}\begin{subfigure}{0.41\textwidth}
      \begin{align}\label{Pphasegadget}\tikzfigM{./identities/Pphasegadget-step-0}=\tikzfigM{./identities/Pphasegadget-step-4}\end{align}
    \end{subfigure}\hspace{.5em}
    \begin{subfigure}{0.29\textwidth}
      \begin{align}\label{CNOTHH}\tikzfigM{./identities/CNOTHH-step-0}=\tikzfigM{./identities/CNOTHH-step-5}\end{align}
    \end{subfigure}\hspace{.5em}
    \begin{subfigure}{0.29\textwidth}
      \begin{align}\label{XcommutCNOT}\tikzfigM{./identities/XcommutCNOT-step-0}=\tikzfigM{./identities/XcommutCNOT-step-5}\end{align}
    \end{subfigure}

    \vspace{.5em}
  \end{center}\end{minipage}}}
\end{figure}

\begin{proof}[Proof of Equations~\eqref{Paddition},\eqref{XPX},\eqref{XX},\eqref{P2pi},\eqref{Pphasegadget},\eqref{CNOTHH} and \eqref{XcommutCNOT}]
  It has been proven in Proposition 20 of \cite{CHMPV} that Equations~\eqref{Paddition} and \eqref{XPX} are consequences of \cref{axQCeuler} in $\QCold$. The derivations still holds in $\QC$. \cref{XX} is a direct consequence of Equations~\eqref{XPX}, \eqref{axQCP0}, \eqref{axQCgphaseempty} and \eqref{axQCgphaseaddition}. \cref{P2pi} is a direct consequence of Equations~\eqref{XX}, \eqref{axQCHH} and \eqref{Paddition}. Equations~\eqref{Pphasegadget},\eqref{CNOTHH} and \eqref{XcommutCNOT} are poved as follows.
  \begin{gather*}
    \tikzfigS{./identities/Pphasegadget-step-0}
    =\tikzfigS{./identities/Pphasegadget-step-1}
    % \eqeqref{axQCswap}\tikzfigS{./identities/Pphasegadget-step-2}
    \eqdeuxeqref{axQCswap}{CNOTCNOT}\tikzfigS{./identities/Pphasegadget-step-3}
    \eqeqref{axQCcnotPcnot}\tikzfigS{./identities/Pphasegadget-step-4}
  \end{gather*}
  \begin{gather*}
    \tikzfigS{./identities/CNOTHH-step-0}
    % \eqeqref{axQCHH}\tikzfigS{./identities/CNOTHH-step-1}
    \eqdeuxeqref{axQCHH}{axQCCZ}\tikzfigS{./identities/CNOTHH-step-2}
    \eqeqref{Pphasegadget}\tikzfigS{./identities/CNOTHH-step-3}
    \eqeqref{axQCCZ}\tikzfigS{./identities/CNOTHH-step-4}
    \eqeqref{axQCHH}\tikzfigS{./identities/CNOTHH-step-5}
  \end{gather*}
  \begin{gather*}
    \tikzfigS{./identities/XcommutCNOT-step-0}
    \eqdeuxeqref{Xdef}{axQCHH}\tikzfigS{./identities/XcommutCNOT-step-1}
    \eqeqref{CNOTHH}\tikzfigS{./identities/XcommutCNOT-step-2}
    \eqdeuxeqref{Zdef}{PcommutCNOT}\tikzfigS{./identities/XcommutCNOT-step-3}
    \eqeqref{CNOTHH}\tikzfigS{./identities/XcommutCNOT-step-4}
    \eqdeuxeqref{Xdef}{axQCHH}\tikzfigS{./identities/XcommutCNOT-step-5}
    \end{gather*}
\end{proof}

\begin{proof}[Proof of \cref{CNOTXX}]
  \begin{gather*}
    \tikzfigS{./qc-completeness/XCNOTXX-step-0}
    \eqeqref{axQCHH}\tikzfigS{./qc-completeness/XCNOTXX-step-1}
    \eqeqref{axQCCZ}\tikzfigS{./qc-completeness/XCNOTXX-step-2}
    \eqtroiseqref{Pphasegadget}{Xdef}{axQCHH}\tikzfigS{./qc-completeness/XCNOTXX-step-3}\\[0.4cm]
    % \eqeqref{PcommutCNOT}\tikzfigS{./qc-completeness/XCNOTXX-step-4}
    \eqeqref{PcommutCNOT}\tikzfigS{./qc-completeness/XCNOTXX-step-5}
    \eqdeuxeqref{Paddition}{P2pi}\tikzfigS{./qc-completeness/XCNOTXX-step-6}\\[0.4cm]
    \eqeqref{XPX}\tikzfigS{./qc-completeness/XCNOTXX-step-7}
    \eqdeuxeqref{axQCgphaseaddition}{axQCgphaseempty}\tikzfigS{./qc-completeness/XCNOTXX-step-8}\\[0.4cm]
    \eqdeuxeqref{XcommutCNOT}{XX}\tikzfigS{./qc-completeness/XCNOTXX-step-9}
    \eqeqref{Pphasegadget}\tikzfigS{./qc-completeness/XCNOTXX-step-10}\\[0.4cm]
    \eqdeuxeqref{axQCHH}{CNOTCNOT}\tikzfigS{./qc-completeness/XCNOTXX-step-11}
    \eqeqref{axQCCZ}\tikzfigS{./qc-completeness/XCNOTXX-step-12}\\[0.4cm]
    \eqdeuxeqref{Pphasegadget}{PcommutCNOT}\tikzfigS{./qc-completeness/XCNOTXX-step-13}
    \eqtroiseqref{Paddition}{axQCP0}{CNOTCNOT}\tikzfigS{./qc-completeness/XCNOTXX-step-14}
    \eqeqref{axQCHH}\tikzfigS{./qc-completeness/XCNOTXX-step-15}
  \end{gather*}
\end{proof}

\subsection{Proofs of usual circuit identities}
\label{appendix:proofsidentities}

\begin{figure}[h!]
  \scalebox{.85}{\fbox{\begin{minipage}{1.159\textwidth}\begin{center}
    \vspace{-1em}
    
    \hspace{-2.5em}\begin{subfigure}{0.27\textwidth}
      \begin{align}\label{ZZ}\tikzfigM{./identities/ZZ-step-0}=\tikzfigM{./identities/ZZ-step-4}\end{align}
    \end{subfigure}\hspace{-.2em}
    \begin{subfigure}{0.65\textwidth}
      \begin{align}\label{inverseEuler}\tikzfigM{./identities/inverseEuler-left}=\tikzfigM{./identities/inverseEuler-right}\end{align}
    \end{subfigure}
    \vspace{-.5em}

    \hspace{-2.7em}\begin{subfigure}{0.26\textwidth}
      \begin{align}\label{RX0}\tikzfigM{./identities/RX0-step-0}=\tikzfigM{./identities/RX0-step-3}\end{align}
    \end{subfigure}\hspace{-1em}
    \begin{subfigure}{0.43\textwidth}
      \begin{align}\label{RXaddition}\tikzfigM{./identities/RXaddition-step-0}=\tikzfigM{./identities/RXaddition-step-4}\end{align}
    \end{subfigure}\hspace{-1em}
    \begin{subfigure}{0.38\textwidth}
      \begin{align}\label{ZRXZ}\tikzfigM{./identities/ZRXthetaZ}=\tikzfigM{./identities/RXminustheta}\end{align}
    \end{subfigure}
    \vspace{-.5em}

    \hspace{-2.5em}\begin{subfigure}{0.44\textwidth}
      \begin{align}\label{HeulerRXPRX}\tikzfigM{./gates/H}=\tikzfigM{./proof-n/HeulerRXPRX}\end{align}
    \end{subfigure}
    \vspace{-.5em}

    \hspace{-2.8em}\begin{subfigure}{0.36\textwidth}
      \begin{align}\label{RXcommutCNOT}\tikzfigM{./identities/RXcommutCNOT-step-0}=\tikzfigM{./identities/RXcommutCNOT-step-5}\end{align}
    \end{subfigure}\hspace{-1em}
    \begin{subfigure}{0.40\textwidth}
      \begin{align}\label{RXphasegadget}\tikzfigM{./identities/RXphasegadget-step-0}=\tikzfigM{./identities/RXphasegadget-step-5}\end{align}
    \end{subfigure}\hspace{-1em}
    \begin{subfigure}{0.30\textwidth}
      \begin{align}\label{CNOTZZ}\tikzfigM{./identities/CNOTZZ-step-0}=\tikzfigM{./identities/CNOTZZ-step-5}\end{align}
    \end{subfigure}
    \vspace{-.5em}

    \hspace{-2.5em}\begin{subfigure}{0.29\textwidth}
      \begin{align}\label{CNOTstargetcommut}\tikzfigM{./identities/CNOTstargetcommut-step-0}=\tikzfigM{./identities/CNOTstargetcommut-step-3}\end{align}
    \end{subfigure}\hspace{3em}
    \begin{subfigure}{0.32\textwidth}
      \begin{align}\label{3CNOTscontrol}\tikzfigM{./identities/3CNOTscontrol-step-0}=\tikzfigM{./identities/3CNOTscontrol-step-3}\end{align}
    \end{subfigure}

    \vspace{.5em}
  \end{center}\end{minipage}}}
\end{figure}

\begin{proof}[Proof of \cref{ZZ}]
  \begin{gather*}
    \tikzfigS{./identities/ZZ-step-0}
    \eqeqref{Xdef}\tikzfigS{./identities/ZZ-step-1}
    \eqeqref{axQCHH}\tikzfigS{./identities/ZZ-step-2}
    \eqeqref{XX}\tikzfigS{./identities/ZZ-step-3}\eqeqref{axQCHH}\tikzfigS{./identities/ZZ-step-4}
  \end{gather*}
\end{proof}

\begin{proof}[Proof of \cref{inverseEuler}]
  \begin{eqnarray*}
    \tikzfigS{./identities/inverseEuler-left}&\eqeqref{RXdef}&\tikzfigS{./identities/inverseEuler-step-1}\\[0.2cm]
    &\eqeqref{axQCHH}&\tikzfigS{./identities/inverseEuler-step-2}\\[0.2cm]
    &\eqeqref{axQCeuler}&\tikzfigS{./identities/inverseEuler-step-3}\\[0.2cm]
    &\eqeqref{RXdef}&\tikzfigS{./identities/inverseEuler-step-4}\\[0.2cm]
    &\eqeqref{axQCHH}&\tikzfigS{./identities/inverseEuler-right}
  \end{eqnarray*}
  With $\alpha_0\defeq\frac{\alpha_1-\alpha_2+\alpha_3}{2}$ and $\beta_0\defeq\alpha_0+\beta_0'+\frac{\beta_1-\beta_2+\beta_3}{2}$.
\end{proof}

\begin{proof}[Proof of \cref{RX0}]
  \begin{gather*}
    \tikzfigS{./identities/RX0-step-0}
    \eqeqref{RXdef}\tikzfigS{./identities/RX0-step-1}
    \eqdeuxeqref{axQCP0}{axQCgphaseempty}\tikzfigS{./identities/RX0-step-2}
    \eqeqref{axQCHH}\tikzfigS{./identities/RX0-step-3}
  \end{gather*}
\end{proof}

\begin{proof}[Proof of \cref{RXaddition}]
  \begin{gather*}
    \tikzfigS{./identities/RXaddition-step-0}
    \eqeqref{RXdef}\tikzfigS{./identities/RXaddition-step-1}
    \eqeqref{axQCHH}\tikzfigS{./identities/RXaddition-step-2}\\[0.2cm]
    \eqeqref{Paddition}\tikzfigS{./identities/RXaddition-step-3}
    \eqeqref{RXdef}\tikzfigS{./identities/RXaddition-step-4}
  \end{gather*}
\end{proof}

\begin{proof}[Proof of \cref{ZRXZ}]
  \begin{gather*}
    \tikzfigS{./identities/ZRXZ-step-0}
    \eqdeuxeqref{Xdef}{RXdef}\tikzfigS{./identities/ZRXZ-step-1}
    \eqeqref{axQCHH}\tikzfigS{./identities/ZRXZ-step-2}\\[0.2cm]
    \eqeqref{XPX}\tikzfigS{./identities/ZRXZ-step-3}
    \eqeqref{RXdef}\tikzfigS{./identities/ZRXZ-step-4}
    \eqdeuxeqref{axQCgphaseaddition}{axQCgphaseempty}\tikzfigS{./identities/ZRXZ-step-5}
  \end{gather*}
\end{proof}

\begin{proof}[Proof of \cref{HeulerRXPRX}]
  \begin{eqnarray*}
    \tikzfigS{./gates/H}&\eqquatreeqref{axQCP0}{Paddition}{RX0}{RXaddition}&\tikzfigS{./proof-n/HeulerRXPRX-step-1}\\[0.2cm]
    &\eqeqref{axQCHeuler}&\tikzfigS{./proof-n/HeulerRXPRX-step-2}\\[0.2cm]
    &\eqeqref{axQCHH}&\tikzfigS{./proof-n/HeulerRXPRX-step-3}\\[0.2cm]
    &\eqeqref{RXdef}&\tikzfigS{./proof-n/HeulerRXPRX-step-4}\\[0.2cm]
    &\eqeqref{axQCHeuler}&\tikzfigS{./proof-n/HeulerRXPRX-step-5}\\[0.2cm]
    &\eqdeuxeqref{Paddition}{axQCP0}&\tikzfigS{./proof-n/HeulerRXPRX}
  \end{eqnarray*}
\end{proof}

\begin{proof}[Proof of \cref{RXcommutCNOT}]
  \begin{gather*}
    \tikzfigS{./identities/RXcommutCNOT-step-0}
    \eqdeuxeqref{RXdef}{axQCHH}\tikzfigS{./identities/RXcommutCNOT-step-1}
    \eqeqref{CNOTHH}\tikzfigS{./identities/RXcommutCNOT-step-2}
    \eqeqref{PcommutCNOT}\tikzfigS{./identities/RXcommutCNOT-step-3}\\[0.4cm]
    \eqeqref{CNOTHH}\tikzfigS{./identities/RXcommutCNOT-step-4}
    \eqdeuxeqref{RXdef}{axQCHH}\tikzfigS{./identities/RXcommutCNOT-step-5}
  \end{gather*}
\end{proof}

\begin{proof}[Proof of \cref{RXphasegadget}]
  \begin{gather*}
    \tikzfigS{./identities/RXphasegadget-step-0}
    \eqdeuxeqref{RXdef}{axQCHH}\tikzfigS{./identities/RXphasegadget-step-1}
    \eqeqref{CNOTHH}\tikzfigS{./identities/RXphasegadget-step-2}
    \eqeqref{Pphasegadget}\tikzfigS{./identities/RXphasegadget-step-3}\\[0.4cm]
    \eqeqref{CNOTHH}\tikzfigS{./identities/RXphasegadget-step-4}
    \eqdeuxeqref{RXdef}{axQCHH}\tikzfigS{./identities/RXphasegadget-step-5}
  \end{gather*}
\end{proof}

\begin{proof}[Proof of \cref{CNOTZZ}]
  \begin{gather*}
    \tikzfigS{./identities/CNOTZZ-step-0}
    \eqdeuxeqref{Xdef}{axQCHH}\tikzfigS{./identities/CNOTZZ-step-1}
    \eqeqref{CNOTHH}\tikzfigS{./identities/CNOTZZ-step-2}
    \eqdeuxeqref{XX}{CNOTXX}\tikzfigS{./identities/CNOTZZ-step-3}
    \eqeqref{CNOTHH}\tikzfigS{./identities/CNOTZZ-step-4}
    \eqeqref{Xdef}\tikzfigS{./identities/CNOTZZ-step-5}
  \end{gather*}
\end{proof}

\begin{proof}[Proof of \cref{CNOTstargetcommut}]
  \begin{gather*}
    \tikzfigS{./identities/CNOTstargetcommut-step-0}
    \eqeqref{CNOTCNOT}\tikzfigS{./identities/CNOTstargetcommut-step-1}
    \eqeqref{axQC3cnots}\tikzfigS{./identities/CNOTstargetcommut-step-2}
    \eqeqref{CNOTCNOT}\tikzfigS{./identities/CNOTstargetcommut-step-3}
  \end{gather*}
\end{proof}

\begin{proof}[Proof of \cref{3CNOTscontrol}]
  \begin{gather*}
    \tikzfigS{./identities/3CNOTscontrol-step-0}
    \eqeqref{CNOTCNOT}\tikzfigS{./identities/3CNOTscontrol-step-1}
    \eqeqref{axQC3cnots}\tikzfigS{./identities/3CNOTscontrol-step-2}
    \eqeqref{CNOTCNOT}\tikzfigS{./identities/3CNOTscontrol-step-3}
  \end{gather*}
\end{proof}

\subsection{Proofs of usual circuit identities over multi-controlled gates}
\label{appendix:mctrlidentities}

\begin{figure}[h!]
  \scalebox{.85}{\fbox{\begin{minipage}{1.159\textwidth}\begin{center}
    \vspace{-.5em}
    
    \hspace{-2.5em}\begin{subfigure}{0.39\textwidth}
      \begin{align}\label{mctrlzeroid}\tikzfigM{./identities/mctrlRX0}=\tikzfigM{./identities/mctrlP0}=\tikzfigM{./identities/Idn}\end{align}
    \end{subfigure}\hspace{1.5em}
    \begin{subfigure}{0.45\textwidth}
      \begin{align}\label{commctrlphaseenhaut}\tikzfigM{./identities/mctrlPhautmctrlRXgrand}=\tikzfigM{./identities/mctrlRXgrandmctrlPhaut}\end{align}
    \end{subfigure}
    \vspace{-0em}
    
    \hspace{-2.5em}\begin{subfigure}{0.39\textwidth}
      \begin{align}\label{mctrlPaddition}\tikzfigM{./identities/mctrlPphi1Pphi2}=\tikzfigM{./identities/mctrlPphi1phi2}\end{align}
    \end{subfigure}\hspace{3em}
    \begin{subfigure}{0.43\textwidth}
      \begin{align}\label{mctrlRXaddition}\tikzfigM{./identities/mctrlRXtheta1RXtheta2}=\tikzfigM{./identities/mctrlRXtheta1theta2}\end{align}
    \end{subfigure}
    \vspace{-.5em}

    \hspace{-2.5em}\begin{subfigure}{0.40\textwidth}
      \begin{align}\label{mctrlPop}\tikzfigM{./identities/XmctrlPphiX}=\tikzfigM{./identities/mctrlPphimctrlPminusphi}\end{align}
    \end{subfigure}\hspace{3em}
    \begin{subfigure}{0.38\textwidth}
      \begin{align}\label{mctrlRXop}\tikzfigM{./identities/ZmctrlRXthetaZ}=\tikzfigM{./identities/mctrlRXminustheta}\end{align}
    \end{subfigure}
    \vspace{-.5em}

    \hspace{-2.5em}\begin{subfigure}{0.31\textwidth}
      \begin{align}\label{mctrlPlift}\tikzfigM{./identities/mctrl21Pphi}=\tikzfigM{./identities/mctrl12Pphi}\end{align}
    \end{subfigure}\hspace{0em}
    \begin{subfigure}{0.34\textwidth}
      \begin{align}\label{mctrlPSWAP}\tikzfigM{./identities/mctrl3Pphi}=\tikzfigM{./identities/mctrl3SWAPPphi}\end{align}
    \end{subfigure}\hspace{0em}
    \begin{subfigure}{0.36\textwidth}
      \begin{align}\label{mctrlRXSWAP}\tikzfigM{./identities/mctrl3RXtheta}=\tikzfigM{./identities/mctrl3SWAPRXtheta}\end{align}
    \end{subfigure}
    \vspace{-.5em}

    \hspace{-2.5em}\begin{subfigure}{0.45\textwidth}
      \begin{align}\tag{\ref{mctrlPinducdef}}\tikzfigM{./shortcut/mctrlPphi-}=\tikzfigM{./shortcut/mctrlPphidef-}\end{align}
    \end{subfigure}\hspace{3em}
    \begin{subfigure}{0.29\textwidth}
      \begin{align}\label{TOFTOF}\tikzfigM{./identities/TOFTOF}=\tikzfigM{./identities/Id3}\end{align}
    \end{subfigure}

    \vspace{.5em}
  \end{center}\end{minipage}}}
\end{figure}

\begin{proof}[Proof of Equations~\eqref{mctrlzeroid},\eqref{commctrlphaseenhaut},\eqref{mctrlPaddition},\eqref{mctrlRXaddition},\eqref{mctrlPop},\eqref{mctrlRXop},\eqref{mctrlPlift},\eqref{mctrlPSWAP} and \eqref{mctrlRXSWAP}]
  Equations \eqref{mctrlzeroid}, \eqref{mctrlPaddition} and \eqref{mctrlRXaddition} are proved in Proposition 13 of \cite{CHMPV}. Equation \eqref{commctrlphaseenhaut} is proved in Lemma 48 of \cite{CHMPV}. Equation \eqref{mctrlPop} follows directly from Lemmas 53 and 54 of \cite{CHMPV}. Equation \eqref{mctrlRXop} is proved in Lemma 47 of \cite{CHMPV}. Equation \eqref{mctrlPlift} is proved in Proposition 12 of \cite{CHMPV}. Equations \eqref{mctrlPSWAP} and \eqref{mctrlRXSWAP} are proved in Proposition 11 of \cite{CHMPV}. The proofs also hold for the equational theory $\QC$ because all the equations used are provable in $\QC$ (more precisely, one can check that that \cref{Mstarold} is not used, and that all the other equations of $\QCold$ are provable in $\QC$ without using multi-controlled gates).
\end{proof}

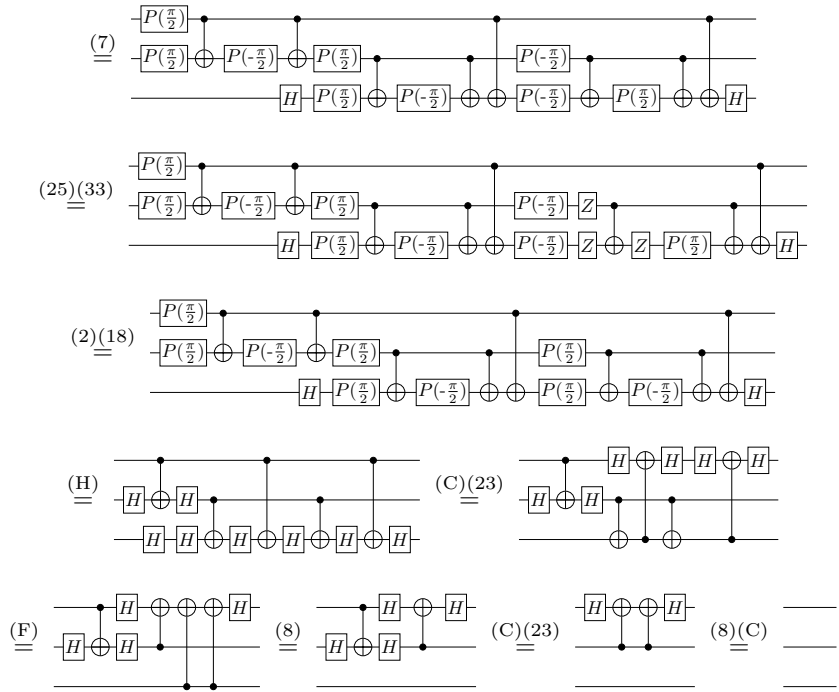
\begin{proof}[Proof of \cref{mctrlPinducdef}]
  In this proof \tikzfigS{./shortcut/mctrlgphasephi} denotes \tikzfigS{./shortcut/gphasephi} on $0$ qubits and \tikzfigS{./shortcut/mcrlPphi1} on one or more qubits. Thus, the multi-controlled phase gate can be expressed as \tikzfigS{./shortcut/mctrlRXdefgphase}. %We proceed by induction on the number of qubits.
  \begin{gather*}
    \tikzfigS{./identities/mctrlPphi}
    \eqeqref{mctrlPdef}\tikzfigS{./identities/mctrlPinducdef-step-1}
    \eqdeuxeqref{mctrlRXSWAP}{mctrlRXdef}\tikzfigS{./identities/mctrlPinducdef-step-2}\\[0.4cm]
    \eqdeuxeqref{axQCHH}{CNOTHH}\tikzfigS{./identities/mctrlPinducdef-step-3}
    \eqeqref{axQCHH}\tikzfigS{./identities/mctrlPinducdef-step-4}\\[0.4cm]
    %\overset{\eqref{mctrlPaddition}\refwcontrol}{=}
    \eqtroiseqref{mctrlzeroid}{mctrlPaddition}{commctrlphaseenhaut}\tikzfigS{./identities/mctrlPinducdef-step-5}
    \eqeqref{mctrlPdef}\tikzfigS{./identities/mctrlPinducdef-step-6}
  \end{gather*}
\end{proof}

\begin{proof}[Proof of \cref{TOFTOF}]
  \begin{gather*}
    \tikzfigS{./identities/TOFTOF}
    \eqeqref{TOFdef}\tikzfigS{./identities/TOFTOF-step-1}
    \eqeqref{axQCHH}\tikzfigS{./identities/TOFTOF-step-2}\\[0.4cm]
    \eqeqref{mctrlPaddition}\tikzfigS{./identities/TOFTOF-step-3}
    \eqeqref{mctrlPinducdef}\tikzfigS{./identities/TOFTOF-step-4}\\[0.4cm]
    \eqeqref{mctrlPinducdef}\tikzfigS{./identities/TOFTOF-step-5}\\[0.4cm]
    \eqdeuxeqref{ZZ}{CNOTZZ}\tikzfigS{./identities/TOFTOF-step-6}\\[0.4cm]
    \eqdeuxeqref{Zdef}{Paddition}\tikzfigS{./identities/TOFTOF-step-7}\\[0.4cm]
    \eqeqref{axQCCZ}\tikzfigS{./identities/TOFTOF-step-8}
    \eqdeuxeqref{axQCHH}{CNOTHH}\tikzfigS{./identities/TOFTOF-step-9}\\[0.4cm]
    \eqeqref{axQC3cnots}\tikzfigS{./identities/TOFTOF-step-10}
    \eqeqref{CNOTCNOT}\tikzfigS{./identities/TOFTOF-step-11}
    \eqdeuxeqref{axQCHH}{CNOTHH}\tikzfigS{./identities/TOFTOF-step-12}
    \eqdeuxeqref{CNOTCNOT}{axQCHH}\tikzfigS{./identities/TOFTOF-step-13}
  \end{gather*}
\end{proof}

\subsection{Proofs of usual circuit identities using ancillae}
\label{appendix:ancillaidentities}

\begin{figure}[h!]
  \scalebox{.85}{\fbox{\begin{minipage}{1.159\textwidth}\begin{center}
    \vspace{-1em}
    
    \hspace{-2.5em}\begin{subfigure}{0.33\textwidth}
      \begin{align}\label{ancillaCNOTpos}\tikzfigM{./identitiesancilla/initXCNOTX}=\tikzfigM{./identitiesancilla/initIdX}\end{align}
    \end{subfigure}\hspace{0em}
    \begin{subfigure}{0.26\textwidth}
      \begin{align}\label{ancillaTOFneg}\tikzfigM{./identitiesancilla/initTOFneg-step-0}=\tikzfigM{./identitiesancilla/initId3}\end{align}
    \end{subfigure}\hspace{0em}
    \begin{subfigure}{0.32\textwidth}
      \begin{align}\label{ancillaTOFpos}\tikzfigM{./identitiesancilla/initTOFpos-step-0}=\tikzfigM{./identitiesancilla/initIdCNOT}\end{align}
    \end{subfigure}
    \vspace{-.5em}

    \hspace{-2.5em}\begin{subfigure}{0.28\textwidth}
      \begin{align}\label{ancillamctrlPneg}\begin{array}{rcl}\tikzfigM{./identitiesancilla/initmctrlPneg-step-0}=\tikzfigM{./identitiesancilla/initIdIdn}\end{array}\end{align}
    \end{subfigure}\hspace{3em}
    \begin{subfigure}{0.29\textwidth}
      \begin{align}\label{ancillamctrlRXneg}\begin{array}{rcl}\tikzfigM{./identitiesancilla/initmctrlRXneg-step-0}=\tikzfigM{./identitiesancilla/initIdIdn}\end{array}\end{align}
    \end{subfigure}
    \vspace{-.5em}

    \hspace{-2.5em}\begin{subfigure}{0.36\textwidth}
      \begin{align}\label{ancillamctrlPpos}\begin{array}{rcl}\tikzfigM{./identitiesancilla/initmctrlPpos-step-0}=\tikzfigM{./identitiesancilla/initIdmctrlPphi}\end{array}\end{align}
    \end{subfigure}\hspace{.5em}
    \begin{subfigure}{0.37\textwidth}
      \begin{align}\label{ancillamctrlRXpos}\begin{array}{rcl}\tikzfigM{./identitiesancilla/initmctrlRXpos-step-0}=\tikzfigM{./identitiesancilla/initmctrlRXpos-step-9}\end{array}\end{align}
    \end{subfigure}

    \vspace{.5em}
  \end{center}\end{minipage}}}
\end{figure}

\begin{proof}[Proof of \cref{ancillaCNOTpos}]
  \begin{gather*}
      \tikzfigS{./identitiesancilla/initXCNOTX-step-0}
      \eqeqref{CNOTXX}\tikzfigS{./identitiesancilla/initXCNOTX-step-1}
      \eqeqref{axQCISOinitcnot}\tikzfigS{./identitiesancilla/initXCNOTX-step-2}
  \end{gather*}
\end{proof}

\begin{proof}[Proof of \cref{ancillamctrlRXneg}]
  \begin{gather*}
    \tikzfigS{./identitiesancilla/initmctrlRXneg-step-0}
    \eqeqref{mctrlRXdef}\tikzfigS{./identitiesancilla/initmctrlRXneg-step-1}
    \eqeqref{axQCHH}\tikzfigS{./identitiesancilla/initmctrlRXneg-step-2}\\[0.4cm]
    \eqdeuxeqref{axQCHH}{CNOTHH}\tikzfigS{./identitiesancilla/initmctrlRXneg-step-3}
    \eqeqref{axQCISOinitcnot}\tikzfigS{./identitiesancilla/initmctrlRXneg-step-4}\\[0.4cm]
    \eqeqref{axQCHH}\tikzfigS{./identitiesancilla/initmctrlRXneg-step-5}
    \eqdeuxeqref{mctrlRXaddition}{mctrlzeroid}\tikzfigS{./identitiesancilla/initIdIdn}
  \end{gather*}
\end{proof}

\begin{proof}[Proof of \cref{ancillamctrlPneg}]
  By induction on the number of controls with base case $n=1$ control.
  \begin{gather*}
    \tikzfigS{./identitiesancilla/initmctrlPneg-basecase-step-0}
    \eqeqref{mctrlPinducdef}\tikzfigS{./identitiesancilla/initmctrlPneg-basecase-step-1}
    \eqeqref{axQCISOinitP}\tikzfigS{./identitiesancilla/initmctrlPneg-basecase-step-2}
    \eqeqref{axQCISOinitcnot}\tikzfigS{./identitiesancilla/initmctrlPneg-basecase-step-3}\eqdeuxeqref{Paddition}{axQCP0}\tikzfigS{./identitiesancilla/initmctrlPneg-basecase-step-4}
  \end{gather*}

  \begin{gather*}
    \tikzfigS{./identitiesancilla/initmctrlPneg-step-0}
    \eqeqref{mctrlPSWAP}\tikzfigS{./identitiesancilla/initmctrlPneg-step-1}
    \eqeqref{mctrlPinducdef}\tikzfigS{./identitiesancilla/initmctrlPneg-step-2}\\[0.4cm]
    \eqeqref{mctrlPlift}\tikzfigS{./identitiesancilla/initmctrlPneg-step-3}
    \overset{\text{IH}}{=}\tikzfigS{./identitiesancilla/initmctrlPneg-step-4}\\[0.4cm]
    =\tikzfigS{./identitiesancilla/initmctrlPneg-step-5}
    \eqeqref{axQCISOinitcnot}\tikzfigS{./identitiesancilla/initmctrlPneg-step-6}
    \eqdeuxeqref{mctrlPaddition}{mctrlzeroid}\tikzfigS{./identitiesancilla/initIdIdn}
  \end{gather*}
\end{proof}

\begin{proof}[Proof of \cref{ancillamctrlRXpos}]
  \begin{gather*}
    \tikzfigS{./identitiesancilla/initmctrlRXpos-step-0}
    \eqeqref{mctrlRXdef}\tikzfigS{./identitiesancilla/initmctrlRXpos-step-1}
    \eqeqref{axQCHH}\tikzfigS{./identitiesancilla/initmctrlRXpos-step-2}\\[0.4cm]
    \eqdeuxeqref{axQCHH}{CNOTHH}\tikzfigS{./identitiesancilla/initmctrlRXpos-step-3}
    \eqeqref{XX}\tikzfigS{./identitiesancilla/initmctrlRXpos-step-4}\\[0.4cm]
    \eqeqref{CNOTXX}\tikzfigS{./identitiesancilla/initmctrlRXpos-step-5}
    \eqeqref{axQCISOinitcnot}\tikzfigS{./identitiesancilla/initmctrlRXpos-step-6}\\[0.4cm]
    \eqeqref{Xdef}\tikzfigS{./identitiesancilla/initmctrlRXpos-step-7}
    \eqeqref{mctrlRXop}\tikzfigS{./identitiesancilla/initmctrlRXpos-step-8}
    \eqeqref{mctrlRXaddition}\tikzfigS{./identitiesancilla/initmctrlRXpos-step-9}
  \end{gather*}
\end{proof}

\begin{proof}[Proof of \cref{ancillamctrlPpos}]
  \begin{gather*}
    \tikzfigS{./identitiesancilla/initmctrlPpos-step-0}
    \eqeqref{mctrlPlift}\tikzfigS{./identitiesancilla/initmctrlPpos-step-1}
    \eqeqref{mctrlPop}\tikzfigS{./identitiesancilla/initmctrlPpos-step-2}
    \eqeqref{mctrlPlift}\tikzfigS{./identitiesancilla/initmctrlPpos-step-3}
    \eqeqref{ancillamctrlPneg}\tikzfigS{./identitiesancilla/initmctrlPpos-step-4}
  \end{gather*}
\end{proof}

\begin{proof}[Proof of \cref{ancillaTOFneg}]
  \begin{gather*}
    \tikzfigS{./identitiesancilla/initTOFneg-step-0}
    \eqeqref{TOFdef}\tikzfigS{./identitiesancilla/initTOFneg-step-1}
    \eqeqref{ancillamctrlPneg}\tikzfigS{./identitiesancilla/initTOFneg-step-2}
    \eqeqref{axQCHH}\tikzfigS{./identitiesancilla/initTOFneg-step-3}
  \end{gather*}
\end{proof}

\begin{proof}[Proof of \cref{ancillaTOFpos}]
  \begin{gather*}
    \tikzfigS{./identitiesancilla/initTOFpos-step-0}
    \eqeqref{TOFdef}\tikzfigS{./identitiesancilla/initTOFpos-step-1}
    \eqeqref{ancillamctrlPpos}\tikzfigS{./identitiesancilla/initTOFpos-step-2}
    \eqeqref{mctrlPinducdef}\tikzfigS{./identitiesancilla/initTOFpos-step-3}\\[0.4cm]
    \eqeqref{axQCCZ}\tikzfigS{./identitiesancilla/initTOFpos-step-4}
    \eqeqref{axQCHH}\tikzfigS{./identitiesancilla/initTOFpos-step-5}
  \end{gather*}
\end{proof}

%%%%%%%%%%%%%%%%%%%%%%%%%%%%%%%%%%%%%%%%%%%%%%%%%%%%%%%%%%%%%%%%%%%%%%%%%%%%%%%
%%%%%%%%%%%%%%%%%%%%%%%%%%%%%%%%%%%%%%%%%%%%%%%%%%%%%%%%%%%%%%%%%%%%%%%%%%%%%%%
\section{Completeness of $\QC$}
\subsection{1-CNot completeness}
\label{appendix:1CNOTcompleteness}
In this Appendix, $\uI,\uX,\uZ,\uCNOT,\uP{\varphi}$ and $\uRX{\theta}$ refers to the unitaries associated with the quantum gates $\gI,\gX,\gZ,\gCNOT,\gP$ and $\gRX$ respectively.

\begin{lemma}[1-qubit completeness]
  \label{lem:QC1qubitcompleteness}
  $\QC$ is complete for 1-qubit quantum circuits, i.e. for any 1-qubit $\propQC$-circuits $C_1,C_2$, if $\interp{C_1}=\interp{C_2}$ then $\QC\vdash C_1=C_2$.
\end{lemma}
\begin{proof}
  $\QC$ contains all the equations of the complete equational theory $\QCold$ acting on at most one qubit.
\end{proof}

\begin{lemma}\label{lem:CNOTprojection}
  By inputing and projecting the $\uCNOT$ unitary on $\ket{0}_1$, $\ket{1}_1$, $\ket{+}_2$, $\ket{-}_2$ and $\bra{0}_1$, $\bra{1}_1$, $\bra{+}_2$, $\bra{-}_2$, we get the following equations:
  \begin{gather*}
      \uCNOT\ket{0}_1=(I\otimes I)\ket{0}_1 \quad\quad \bra{0}_1\uCNOT=\bra{0}_1(I\otimes I)\\
      \uCNOT\ket{1}_1=(I\otimes X)\ket{1}_1 \quad\quad \bra{1}_1\uCNOT=\bra{1}_1(I\otimes X)\\
      \uCNOT\ket{+}_2=(I\otimes I)\ket{+}_2 \quad\quad \bra{+}_2\uCNOT=\bra{+}_2(I\otimes I)\\
      \uCNOT\ket{-}_2=(Z\otimes I)\ket{-}_2 \quad\quad \bra{-}_2\uCNOT=\bra{-}_2(Z\otimes I)
  \end{gather*}
\end{lemma}
\begin{proof}
  By straightforward analysis.
\end{proof}

\begin{lemma}\label{lem:braketzero}
  Let $U\in\mathcal{U}_2$ be a 1-qubit unitary. If $\bra{0}U\ket{0}=0\vee\bra{1}U\ket{1}=0$ then there exist $\varphi,\delta\in\R$ such that $U=e^{i\delta}XP(\varphi)$. Similarly, if $\bra{+}U\ket{+}=0\vee\bra{-}U\ket{-}=0$ then there exist $\theta,\delta\in\R$ such that $U=e^{i\delta}ZR_X(\theta)$.
\end{lemma}
\begin{proof}
  First notice that $\bra{0}U\ket{0}=0$ iff $\bra{1}U\ket{1}=0$. Then by unitarity there exists $\delta,\phi\in\R$ such that $U=\big(\begin{smallmatrix}
    0 & e^{i\phi}\\
    e^{i\delta} & 0
  \end{smallmatrix}\big)=X\big(\begin{smallmatrix}
    e^{i\delta} & 0\\
    0 & e^{i\phi}
  \end{smallmatrix}\big)=e^{i\delta}X\big(\begin{smallmatrix}
    1 & 0\\
    0 & e^{i(\phi-\delta)}
  \end{smallmatrix}\big)$. And we are done by taking $\varphi\defeq\phi-\delta$.

  We prove the second statement by reducing it to the first one. $\bra{+}U\ket{+}=0\vee\bra{-}U\ket{-}=0$ iff $\bra{0}HUH\ket{0}=0\vee\bra{1}HUH\ket{1}=0$. Thus there exists $\theta,\phi\in\R$ such that $HUH=e^{i\phi}XP(\theta)$. This implies that $U=e^{i\phi}HXP(\theta)H=e^{i\left(\phi+\frac{\theta}{2}\right)}ZR_X(\theta)$ and we are done by taking $\delta\defeq\phi+\frac{\theta}{2}$.
\end{proof}

\begin{lemma}\label{lem:PPIRRI}
  For any $\varphi,\varphi',\delta\in\R$, if $P(\varphi)P(\varphi')=e^{i\delta}I$ then $\varphi'=-\varphi\pmod{2\pi}$. Similarly, for any $\theta,\theta',\delta\in\R$, if $R_X(\theta)R_X(\theta')=e^{i\delta}I$ then $\theta'=-\theta\pmod{2\pi}$.
\end{lemma}
\begin{proof}
  For the first statement $P(\varphi)P(\varphi')=\big(\begin{smallmatrix}
      1 & 0\\
      0 & e^{i(\varphi+\varphi')}
  \end{smallmatrix}\big)=\big(\begin{smallmatrix}
      e^{i\delta} & 0\\
      0 & e^{i\delta}
  \end{smallmatrix}\big)=e^{i\delta}I$ directly implies that $\varphi'=-\varphi\pmod{2\pi}$. The second statement is reduced to the first one as follows:
  \begin{align*}
      R_X(\theta)R_X(\theta')=e^{i\delta}I&\Rightarrow e^{-i\left(\frac{\theta}{2}+\frac{\theta'}{2}\right)}HP(\theta)P(\theta')H=e^{i\delta}I\Rightarrow e^{-i\left(\frac{\theta}{2}+\frac{\theta'}{2}\right)}P(\theta)P(\theta')=e^{i\delta}I\\
      &\Rightarrow P(\theta)P(\theta')=e^{i\left(\delta+\frac{\theta}{2}+\frac{\theta'}{2}\right)}I\Rightarrow\theta'=-\theta\pmod{2\pi}
  \end{align*}
\end{proof}

\begin{lemma}\label{lem:characterizationCNOT}
  Let $A,B,C,D\in\mathcal{U}_2$ be 1-qubit unitaries, if $(C\otimes D)\circ\uCNOT\circ(A\otimes B)=\uCNOT$ (see the following circuit representation) then there exist $\alpha,\beta,\gamma,\varphi, \theta\in\R$ and $k,\ell\in\{0,1\}$ such that $A=e^{i\alpha}X^kP(\varphi)$, $B=e^{i\beta}Z^\ell R_X(\theta)$, $C=e^{i\gamma}P(-\varphi)Z^\ell X^k$ and $D=e^{i(-\alpha-\beta-\gamma)}R_X(-\theta)Z^\ell X^k$.
  \begin{equation*}
    \begin{array}{rcl}\interp{\tikzfigS{./proof-n/ABCD}}&=&\interp{\tikzfigS{./proof-n/CNOT}}\end{array}
  \end{equation*}
\end{lemma}
\begin{proof}
  From the condition we derive four equations satisfied by $A,B,C,D$ and we conduct a case distinction corresponding to the four possible assignements of $k,\ell\in\{0,1\}$.
  \begin{gather*}
    (C\otimes D)\circ\uCNOT\circ(A\otimes B)=\uCNOT\\
    \Rightarrow\begin{cases}
      (C\otimes D)\circ\uCNOT\circ(I\otimes B)=\uCNOT\circ(A^\dagger\otimes I)\\
      (C\otimes D)\circ\uCNOT\circ(A\otimes I)=\uCNOT\circ(I\otimes B^\dagger)\\
    \end{cases}\\
    \Rightarrow\begin{cases}
      \bra{0}_1(C\otimes D)\circ\uCNOT\circ(I\otimes B)\ket{0}_1=\bra{0}_1\uCNOT\circ(A^\dagger\otimes I)\ket{0}_1\\
      \bra{1}_1(C\otimes D)\circ\uCNOT\circ(I\otimes B)\ket{1}_1=\bra{1}_1\uCNOT\circ(A^\dagger\otimes I)\ket{1}_1\\
      \bra{+}_2(C\otimes D)\circ\uCNOT\circ(A\otimes I)\ket{+}_2=\bra{+}_2\uCNOT\circ(I\otimes B^\dagger)\ket{+}_2\\
      \bra{-}_2(C\otimes D)\circ\uCNOT\circ(A\otimes I)\ket{-}_2=\bra{-}_2\uCNOT\circ(I\otimes B^\dagger)\ket{-}_2
    \end{cases}\\
    \overset{\text{\cref{lem:CNOTprojection}}}{\Rightarrow}\begin{cases}
      \bra{0}_1(C\otimes D)\circ(I\otimes I)\circ(I\otimes B)\ket{0}_1=\bra{0}_1(I\otimes I)\circ(A^\dagger\otimes I)\ket{0}_1\\
      \bra{1}_1(C\otimes D)\circ(I\otimes X)\circ(I\otimes B)\ket{1}_1=\bra{1}_1(I\otimes X)\circ(A^\dagger\otimes I)\ket{1}_1\\
      \bra{+}_2(C\otimes D)\circ(I\otimes I)\circ(A\otimes I)\ket{+}_2=\bra{+}_2(I\otimes I)\circ(I\otimes B^\dagger)\ket{+}_2\\
      \bra{-}_2(C\otimes D)\circ(Z\otimes I)\circ(A\otimes I)\ket{-}_2=\bra{-}_2(Z\otimes I)\circ(I\otimes B^\dagger)\ket{-}_2
    \end{cases}\\
    \Rightarrow\begin{cases}
      \bra{0}C\ket{0}DB=\bra{0}A^\dagger\ket{0}I\\
      \bra{1}C\ket{1}DXB=\bra{1}A^\dagger\ket{1}X\\
      \bra{+}D\ket{+}CA=\bra{+}B^\dagger\ket{+}I\\
      \bra{-}D\ket{-}CZA=\bra{-}B^\dagger\ket{-}Z
    \end{cases}
  \end{gather*}

  \textbf{Case $\bm{\bra{0}A^\dagger\ket{0}\ne 0}$ and $\bm{\bra{+}B^\dagger\ket{+}\ne 0}$.} It must also be the case that $\bra{0}C\ket{0}\ne 0$ and $\bra{+}D\ket{+}\ne 0$. Moreover, by unitarity of $A^\dagger$ and $B^\dagger$, we also have $\bra{1}A^\dagger\ket{1}\ne 0$ and $\bra{-}B^\dagger\ket{-}\ne 0$. The first equation implies $D=e^{i\delta}B^\dagger$ for some $\delta\in\R$, which implies that $\bra{+}D\ket{+}=e^{i\delta}\bra{+}B^\dagger\ket{+}$ and $\bra{-}D\ket{-}=e^{i\delta}\bra{-}B^\dagger\ket{-}$. Then the third equation implies $C=e^{-i\delta}A^\dagger$, which implies $\bra{1}C\ket{1}=e^{-i\delta}\bra{1}A^\dagger\ket{1}$. Hence the system becomes:
  \begin{gather*}
      \begin{cases}
          DB=e^{i\delta}I\\
          DX B=e^{i\delta}X\\
          CA=e^{-i\delta}I\\
          CZ A=e^{-i\delta}Z
      \end{cases}\Rightarrow\begin{cases}
          CA=CZ AZ\\
          DB=DX BX
      \end{cases}
  \end{gather*}
  The first equation implies that there exist $\varphi,\alpha\in\R$ such that $A=e^{i\alpha}P(\varphi)$ (because $A=ZAZ$), which implies that $C=e^{i(-\delta-\alpha)}P(-\varphi)$. Similarly, the second equation implies that there exist $\theta,\beta\in\R$ such that $B=e^{i\beta}R_X(\theta)$ (because $B=XBX$), which implies that $D=e^{i(\delta-\beta)}R_X(-\theta)$. And we are done by tacking $k=\ell=0$ and $\gamma\defeq -\delta-\alpha$ which leads to $\delta-\beta=-\alpha-\beta-\gamma$.

  \textbf{Case $\bm{\bra{0}A^\dagger\ket{0}=0}$ and $\bm{\bra{+}B^\dagger\ket{+}\ne 0}$.} It must also be the case that $\bra{0}C\ket{0}=0$ and $\bra{+}D\ket{+}\ne 0$. Lemma \ref{lem:braketzero} implies that there exist $\varphi,\varphi',\alpha,\gamma\in\R$ such that $A=e^{i\alpha}XP(\varphi)$ and $C=e^{i\gamma}P(\varphi')X$. Moreover, the third equation implies $CA=e^{i\delta}I$ for some $\delta\in\R$, thus $e^{i(\alpha+\gamma)}P(\varphi')XXP(\varphi)=e^{i\delta}I$, which implies that $\varphi'=-\varphi \pmod{2\pi}$ (Lemma \ref{lem:PPIRRI}). Then we can use the following derivation to get a new condition satisfied by $B$ and $D$.
  \begin{align*}
    \interp{\tikzfigS{./proof-n/casek1l0-step-1}}\eqeqref{CNOTXX}\interp{\tikzfigS{./proof-n/casek1l0-step-2}}\eqtroiseqref{PcommutCNOT}{Paddition}{axQCP0}\interp{\tikzfigS{./proof-n/casek1l0-step-3}}
  \end{align*}
  We get $e^{i(\alpha+\gamma)}(I\otimes DX)\circ\uCNOT\circ(I\otimes B)=\uCNOT$ from which we obtain two new equations:
  \begin{gather*}
      e^{i(\alpha+\gamma)}(I\otimes DX)\circ\uCNOT\circ(I\otimes B)=\uCNOT\\
      \Rightarrow\begin{cases}
          e^{i(\alpha+\gamma)}\bra{0}_1(I\otimes DX)\circ\uCNOT\circ(I\otimes B)\ket{0}_1=\bra{0}_1\uCNOT\ket{0}_1\\
          e^{i(\alpha+\gamma)}\bra{1}_1(I\otimes DX)\circ\uCNOT\circ(I\otimes B)\ket{1}_1=\bra{1}_1\uCNOT\ket{1}_1
      \end{cases}\\
      \overset{\text{\cref{lem:CNOTprojection}}}{\Rightarrow}\begin{cases}
        e^{i(\alpha+\gamma)}\bra{0}_1(I\otimes DX)\circ(I\otimes I)\circ(I\otimes B)\ket{0}_1=\bra{0}_1(I\otimes I)\ket{0}_1\\
        e^{i(\alpha+\gamma)}\bra{1}_1(I\otimes DX)\circ(I\otimes X)\circ(I\otimes B)\ket{1}_1=\bra{1}_1(I\otimes X)\ket{1}_1
    \end{cases}\\
      \Rightarrow\begin{cases}
          e^{i(\alpha+\gamma)}DX B=I\\
          e^{i(\alpha+\gamma)}DXX B=X\\
      \end{cases}
  \end{gather*}
  This implies that $DX B=DBX$, thus there exist $\theta,\beta\in\R$ such that $B=e^{i\beta}R_X(\theta)$ (because $B=XBX$). The first equation implies $D=e^{i(-\alpha-\beta-\gamma)}R_X(-\theta)X$, and we are done by tacking $k=1$ and $\ell=0$.

  \textbf{Case $\bm{\bra{0}A^\dagger\ket{0}\ne 0}$ and $\bm{\bra{+}B^\dagger\ket{+}= 0}$.} It must also be the case that $\bra{0}C\ket{0}\ne 0$ and $\bra{+}D\ket{+}=0$. Lemma \ref{lem:braketzero} implies that there exist $\theta,\theta',\beta,\sigma\in\R$ such that $B=e^{i\beta}ZR_X(\theta)$ and $D=e^{i\sigma}R_X(\theta')Z$. Moreover, the first equation implies $DB=e^{i\delta}I$ for some $\delta\in\R$, thus $e^{i(\beta+\sigma)}R_X(\theta')ZZR_X(\theta)=e^{i\delta}I$, which implies that $\theta'=-\theta\pmod{2\pi}$ (Lemma \ref{lem:PPIRRI}). Then we can use the following derivation to get a new condition satisfied by $A$ and $C$.
  \begin{align*}
    \interp{\tikzfigS{./proof-n/casek0l1-step-1}}\eqeqref{CNOTZZ}\interp{\tikzfigS{./proof-n/casek0l1-step-2}}\eqtroiseqref{RXcommutCNOT}{RXaddition}{RX0}\interp{\tikzfigS{./proof-n/casek0l1-step-3}}
  \end{align*}
  We get $e^{i(\beta+\sigma)}(CZ\otimes I)\circ\uCNOT\circ(A\otimes I)=\uCNOT$ from which we obtain two new equations:
  \begin{gather*}
      e^{i(\beta+\sigma)}(CZ\otimes I)\circ\uCNOT\circ(A\otimes I)=\uCNOT\\
      \Rightarrow\begin{cases}
          e^{i(\beta+\sigma)}\bra{+}_2(CZ\otimes I)\circ\uCNOT\circ(A\otimes I)\ket{+}_2=\bra{+}_2\uCNOT\ket{+}_2\\
          e^{i(\beta+\sigma)}\bra{-}_2(CZ\otimes I)\circ\uCNOT\circ(A\otimes I)\ket{-}_2=\bra{-}_2\uCNOT\ket{-}_2
      \end{cases}\\
      \Rightarrow\begin{cases}
        e^{i(\beta+\sigma)}\bra{+}_2(CZ\otimes I)\circ(I\otimes I)\circ(A\otimes I)\ket{+}_2=\bra{+}_2(I\otimes I)\ket{+}_2\\
        e^{i(\beta+\sigma)}\bra{-}_2(CZ\otimes I)\circ(Z\otimes I)\circ(A\otimes I)\ket{-}_2=\bra{-}_2(Z\otimes I)\ket{-}_2
      \end{cases}\\
      \Rightarrow\begin{cases}
          e^{i(\beta+\sigma)}CZA=I\\
          e^{i(\beta+\sigma)}CZZA=Z\\
      \end{cases}
  \end{gather*}
  This implies that $CZA=CAZ$, thus there exist $\varphi,\alpha\in\R$ such that $A=e^{i\alpha}P(\varphi)$ (because $A=ZAZ$). The first equation implies $C=e^{i(-\alpha-\beta-\sigma)}P(-\varphi)Z$, and we are done by tacking $k=0$, $\ell=1$ and $\gamma\defeq-\alpha-\beta-\sigma$, which leads to $\sigma=-\alpha-\beta-\gamma$.

  \textbf{Case $\bm{\bra{0}A^\dagger\ket{0}= 0}$ and $\bm{\bra{+}B^\dagger\ket{+}= 0}$.} It must also be the case that $\bra{0}C\ket{0}=0$ and $\bra{+}D\ket{+}=0$. Lemma \ref{lem:braketzero} implies that there exist $\varphi,\varphi',\theta,\theta',\alpha,\beta,\gamma,\delta\in\R$ such that $A=e^{i\alpha}XP(\varphi)$, $B=e^{i\beta}ZR_X(\theta)$, $C=e^{i\gamma}P(\varphi')X$ and $D=e^{i\delta}R_X(\theta')Z$. Then we can use the following derivation to get a new condition satisfied by $\varphi,\varphi',\theta,\theta'$.
  \begin{align*}
    \interp{\tikzfigS{./proof-n/casek1l1-step-1}}\eqdeuxeqref{CNOTXX}{CNOTZZ}\interp{\tikzfigS{./proof-n/casek1l1-step-2}}
  \end{align*}
  
  We get $e^{i(\alpha+\beta+\gamma+\delta)}(P(\varphi')Z\otimes R_X(\theta')X)\circ\uCNOT\circ(P(\varphi)\otimes R_X(\theta))=\uCNOT$ from which we obtain two new equations:
  \begin{gather*}
    e^{i(\alpha+\beta+\gamma+\delta)}(P(\varphi')Z\otimes R_X(\theta'X)\circ\uCNOT\circ(P(\varphi)\otimes R_X(\theta))=\uCNOT\\
    \Rightarrow\begin{cases}
        e^{i(\alpha+\beta+\gamma+\delta)}\bra{0}_1(P(\varphi')Z\otimes R_X(\theta')X)\uCNOT(P(\varphi)\otimes R_X(\theta))\ket{0}_1=\bra{0}_1\uCNOT\ket{0}_1\\
        e^{i(\alpha+\beta+\gamma+\delta)}\bra{+}_2(P(\varphi')Z\otimes R_X(\theta')X)\uCNOT(P(\varphi))\otimes R_X(\theta)\ket{+}_2=\bra{+}_2\uCNOT\ket{+}_2
    \end{cases}\\
    \Rightarrow\begin{cases}
      e^{i(\alpha+\beta+\gamma+\delta)}\bra{0}_1(P(\varphi')Z\otimes R_X(\theta')X)(I\otimes I)(P(\varphi)\otimes R_X(\theta))\ket{0}_1=\bra{0}_1(I\otimes I)\ket{0}_1\\
      e^{i(\alpha+\beta+\gamma+\delta)}\bra{+}_2(P(\varphi')Z\otimes R_X(\theta')X)(I\otimes I)(P(\varphi)\otimes R_X(\theta))\ket{+}_2=\bra{+}_2(I\otimes I)\ket{+}_2
  \end{cases}\\
    \Rightarrow\begin{cases}
        e^{i(\alpha+\beta+\gamma+\delta)} R_X(\theta')X R_X(\theta)=I\\
        e^{i(\alpha+\beta+\gamma+\delta)}e^{-i(\theta+\theta')/2}P(\varphi')ZP(\varphi)=I\\
    \end{cases}\\
    \Rightarrow\begin{cases}
        R_X(\theta')=e^{-i(\alpha+\beta+\gamma+\delta)}R_X(-\theta)X\\
        P(\varphi')=e^{-i(\alpha+\beta+\gamma+\delta)}e^{i(\theta+\theta')/2}P(-\varphi)Z\\
    \end{cases}\\
    \Rightarrow\begin{cases}
        R_X(\theta')=e^{-i(\alpha+\beta+\gamma+\delta)}e^{i\pi/2}R_X(\pi-\theta)\\
        P(\varphi')=e^{-i(\alpha+\beta+\gamma+\delta)}e^{i(\theta+\theta')/2}P(\pi-\varphi)\\
    \end{cases}\\
    \overset{\text{Lemma \ref{lem:PPIRRI}}}{\Rightarrow}\begin{cases}
        \theta'=\pi-\theta\pmod{2\pi}\\
        \varphi'=\pi-\varphi\pmod{2\pi}
    \end{cases}
\end{gather*}
  Hence, we get $A=e^{i\alpha}XP(\varphi)$, $B=e^{i\beta}ZR_X(\theta)$, $C=e^{i\gamma}P(\pi-\varphi)X=e^{i\gamma}P(-\varphi)ZX$ and $D=e^{i(-\alpha-\beta-\gamma+\pi/2)}R_X(\pi-\theta)Z=e^{i(-\alpha-\beta-\gamma)}R_X(-\theta)ZX$. Thus we are done by taking $k=\ell=1$.
\end{proof}

\begin{proof}[Proof of \cref{lem:1CNOTcompleteness} (1-CNot completeness)]
  Let $\mathcal{C},\mathcal{C}'$ be two 1-CNot $\propQC$-circuits, i.e. circuits containing at most one $\gCNOT$ gate each, such that $\interp{\mathcal{C}}=\interp{\mathcal{C}'}$.
  
  First notice that if both $\mathcal{C},\mathcal{C}'$ contains no CNot, then, according to \cref{lem:QC1qubitcompleteness} we have $\QC\vdash \mathcal{C}=\mathcal{C}'$. If $\mathcal{C}$ contains a CNot then one can show that $\mathcal{C'}$ must contain a CNot otherwise they would not have the same semantics.

  Then, w.l.o.g.~we can suppose that the CNot is applied to the first two qubits in $\mathcal{C}$. We first show that the CNot is also applied to the first two qubits in $\mathcal{C}'$, and that the permutation of wires is the same in both circuits. Pushing all swaps to the right in $\mathcal{C}$, we get:
  \[\tikzfigS{./proof-n/C-many-qbits}\]
  where $\sigma$ is a permutation of wires. 
  In $\mathcal C'$, using the prop equations to move the CNot on the first two qubits, we get:
  \[\tikzfigS{./proof-n/Cp-many-qbits}\]
  By applying the inverse of all 1-qubit unitaries from $\mathcal C$, as well as the inverse of $\sigma$, $\interp{\mathcal{C}}=\interp{\mathcal{C}'}$ becomes equivalent to:
  \[\interp{\tikzfigS{./proof-n/CNot-id}} = \interp{\tikzfigS{./proof-n/CNot-id-p}}\]
  with $\sigma_3 = \sigma_2\circ\sigma^{-1}$. It then becomes apparent that if $\sigma_1(i)\geq3$, then $\sigma_3(\sigma_1(i))=i$. Moreover, we get $\interp{c''_i}=\interp{\gI}$. We then have:
    \[\interp{\tikzfigS{./proof-n/CNot-id}} = \interp{\tikzfigS{./proof-n/CNot-id-aux}}\]
    Getting back to $\mathcal{C}$ and $\mathcal{C}'$, we conclude that qubits that are not involved with the CNot have the same 1-qubit unitary applied to them on both sides, and are permuted with the rest in the same way. By completeness of $\QC$ for 1-qubit unitaries, it is now enough to show the result when $\mathcal{C}$ and $\mathcal{C}'$ are 2-qubit circuits.
  
  We can now show that if a swap appears on one side, it also appears on the other side (i.e.~that $\sigma_1'=\sigma_3'$ above). Indeed, suppose that $\mathcal C$ has a swap and $\mathcal{C}'$ does not. Using the prop equations, we can push the swap to the right, and get $\mathcal C = \gSWAP\circ\bar{\mathcal{C}'}$. Thus $\interp{\mathcal C}=\interp{\mathcal C'}$ is then equivalent to $\interp{\gSWAP}=\interp{\mathcal{C}'\circ\bar{\mathcal{C}}^\dag}$. We hence have a circuit $\mathcal{C}'\circ\bar{\mathcal{C}}^\dag$ that implements $\gSWAP$, using at most two CNots. This contradicts Theorem~6 of \cite{Vatan2004optimal} that proves that the swap requires at least $3$ CNots. Hence, if $\mathcal C$ has a swap then so does $\mathcal{C}'$.
  
  In the case where the two circuits have a swap, the equality between the two becomes equivalent (\cref{prop:crossinggate}) to the equality without swaps on both side. Hence, we can assume w.l.o.g.~that $\mathcal{C}$ and $\mathcal{C}'$ are 2-qubit circuits containing one and only one CNot and no swap. Moreover thanks to \cref{CNOTHH} we can assume:
  \begin{equation*}
    \begin{array}{rcccccl}\tikzfigS{./proof-n/Ccal}&=&\tikzfigS{./proof-n/A1B1C1D1}&\qquad\qquad\qquad\qquad&\tikzfigS{./proof-n/Ccalprime}&=&\tikzfigS{./proof-n/A2B2C2D2}\end{array}
  \end{equation*}
  for some $1$-qubit circuits \tikzfigS{./proof-n/Ai},\tikzfigS{./proof-n/Bi},\tikzfigS{./proof-n/Ci},\tikzfigS{./proof-n/Di}.

  First, by the simplification principle (\cref{prop:crossinggate}), we reduce it to showing the following equation for any semantically correct 1-qubit circuits \tikzfigS{./proof-n/A},\tikzfigS{./proof-n/B},\tikzfigS{./proof-n/C},\tikzfigS{./proof-n/D}.
  \begin{equation*}
    \begin{array}{rcl}\tikzfigS{./proof-n/ABCD}&=&\tikzfigS{./proof-n/CNOT}\end{array}
  \end{equation*}

  \cref{lem:characterizationCNOT} and \cref{lem:QC1qubitcompleteness} together with Equations \eqref{axQCgphaseempty} and \eqref{axQCgphaseaddition} implies that this is always the case that this equation is equivalent (\cref{prop:crossinggate}) to one of the following equations:
  \begin{gather*}
    \tikzfigS{./proof-n/casek0l0}=\tikzfigS{./proof-n/CNOT} \hspace{4em} \tikzfigS{./proof-n/casek1l0}=\tikzfigS{./proof-n/CNOT} \\[0.2cm]
    \tikzfigS{./proof-n/casek0l1}=\tikzfigS{./proof-n/CNOT} \hspace{4em} \tikzfigS{./proof-n/casek1l1}=\tikzfigS{./proof-n/CNOT}
  \end{gather*}

  We conclude the proof by observing that we can derive all those equations for any $\varphi,\theta\in\R$ using Equations \eqref{CNOTZZ},\eqref{CNOTXX},\eqref{PcommutCNOT},\eqref{RXcommutCNOT},\eqref{Paddition},\eqref{RXaddition},\eqref{axQCP0}, and \eqref{RX0}.
\end{proof}

\subsection{Proof of Equation \eqref{n} of $\QCold$}
\label{appendix:n}

\begin{proof}[Proof of Equation \eqref{n}]
  We first show that any circuits (containing four CNot gates) of the form $\tikzfigS{./proof-n/newderiv_00}$ can always be transformed in $\QC$ into a circuit (containing only two CNot gates) of the form $\tikzfigS{./proof-n/newderiv_06}$ where $\tikzfigS{./proof-n/Ci}$ denotes a one-qubit circuit. This uses the fact (referenced $(*)$ below) that any one-qubit circuit can be transformed in $\QC$ into a circuit of the form $\tikzfigS{./proof-n/RXPRX}$ or $\tikzfigS{./proof-n/PRXP}$ (by the completeness of $\QC$ for one-qubit circuits and the well-know Euler-decomposition). The derivation goes as follows.
  \begin{eqnarray*}
    &&\tikzfigS{./proof-n/newderiv_00}\\[0.2cm]
    &\overset{(*)\eqref{axQCgphaseaddition}}{=}&\tikzfigS{./proof-n/newderiv_01}\\[0.2cm]
    &\eqeqref{PcommutCNOT}&\tikzfigS{./proof-n/newderiv_02}\\[0.2cm]
    &\overset{(*)\eqref{axQCgphaseaddition}}{=}&\tikzfigS{./proof-n/newderiv_03}\\[0.2cm]
    &\eqdeuxeqref{RXphasegadget}{RXcommutCNOT}&\tikzfigS{./proof-n/newderiv_04}\\[0.2cm]
    &\eqeqref{axQCcnotPcnot}&\tikzfigS{./proof-n/newderiv_05}
  \end{eqnarray*}

  Then, using Equations~\eqref{RXcommutCNOT} and \eqref{RXphasegadget}, \cref{n} becomes
  \begin{equation*}
    \tikzfigS{./proof-n/newleft_00}=\tikzfigS{./proof-n/newright_00}
  \end{equation*}
  We can then use the above derivation to transform both circuits into circuits of the form $\tikzfigS{./proof-n/newderiv_06}$. Then, by using the simplification principle (\cref{prop:crossinggate}), we can push the RHS circuit to the end of the LHS circuit, which leads to a circuit of the form $\tikzfigS{./proof-n/newderiv_00}$ (up to a global phase) from which we can apply again the above derivation, leading to a circuit containing only two CNot gates. Then, by using the simplification principle again, we can turn the equation into an equivalent equations with only one CNot on both sides. Finally, the completeness of $\QC$ for circuits containing at most one CNot gate (\cref{lem:1CNOTcompleteness}) concludes the proof.
\end{proof}

\subsection{Proof of Equation \eqref{o} of $\QCold$}
\label{appendix:o}
\begin{eqnarray*}
  &&\tikzfigS{./proof-o/step-0}\\[0.4cm]
  &\eqeqref{axQC3cnots}&\tikzfigS{./proof-o/step-1}\\[0.4cm]
  &\eqdeuxeqref{CNOTCNOT}{axQCswap}&\tikzfigS{./proof-o/step-2}\\[0.4cm]
  &\eqdeuxeqref{CNOTstargetcommut}{RXcommutCNOT}&\tikzfigS{./proof-o/step-3}\\[0.4cm]  &\eqdeuxeqref{CNOTHH}{CNOTCNOT}&\tikzfigS{./proof-o/step-4}\\[0.4cm]
  &=&\tikzfigS{./proof-o/step-5}\\[0.4cm]
  &\eqdeuxeqref{RXcommutCNOT}{RXphasegadget}&\tikzfigS{./proof-o/step-6}\\[0.4cm]
  &\eqdeuxeqref{CNOTstargetcommut}{RXcommutCNOT}&\tikzfigS{./proof-o/step-7}\\[0.4cm]
  &\eqdeuxeqref{RXphasegadget}{RXcommutCNOT}&\tikzfigS{./proof-o/step-8}\\[0.4cm]
  &\eqeqref{n}&\tikzfigS{./proof-o/step-9}\\[0.4cm]
  &\eqdeuxeqref{axQCHH}{CNOTHH}&\tikzfigS{./proof-o/step-10}\\[0.4cm]
  &\eqdeuxeqref{RXcommutCNOT}{RXphasegadget}&\tikzfigS{./proof-o/step-11}\\[0.4cm]
  &\eqeqref{CNOTstargetcommut}&\tikzfigS{./proof-o/step-12}\\[0.4cm]
  % &\eqdeuxeqref{RXcommutCNOT}{RXphasegadget}&\tikzfigS{./proof-o/step-13}\\[0.4cm] % this step is useless
  &\eqdeuxeqref{RXcommutCNOT}{RXphasegadget}&\tikzfigS{./proof-o/step-14}\\[0.4cm]
  &\eqdeuxeqref{CNOTstargetcommut}{RXcommutCNOT}&\tikzfigS{./proof-o/step-15}\\[0.4cm]
  &\eqdeuxeqref{RXcommutCNOT}{RXphasegadget}&\tikzfigS{./proof-o/step-16}\\[0.4cm]
  &=&\tikzfigS{./proof-o/step-17}\\[0.4cm]
  &\eqdeuxeqref{CNOTCNOT}{CNOTHH}&\tikzfigS{./proof-o/step-18}\\[0.4cm]
  &\eqdeuxeqref{CNOTstargetcommut}{RXcommutCNOT}&\tikzfigS{./proof-o/step-19}\\[0.4cm]
  &\eqdeuxeqref{CNOTCNOT}{axQCswap}&\tikzfigS{./proof-o/step-20}\\[0.4cm]
  &\eqeqref{axQC3cnots}&\tikzfigS{./proof-o/step-21}
\end{eqnarray*}

\subsection{Proof of Equation \eqref{Mstarold}, existence and uniqueness of the RHS of Equation~\eqref{Mstar}}
\label{appendix:Mstarold}

The proof uses some properties of multi-controlled gates:
\begin{lemma}\label{Palwayscommute}
It follows from the equations of $\QC$ without \cref{Mstar} that two multi-controlled $P$ gates always commute, regardless of the colours and positions of their controls, and of the positions of their targets.
\end{lemma}
For instance,
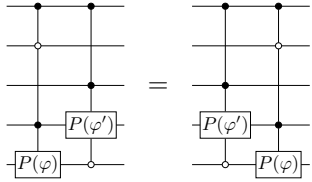
\[\tikzfigS{./identities/excomPleft}\ =\ \tikzfigS{./identities/excomPright}\]
\begin{proof}[Proof of \cref{Palwayscommute}]
This is a direct consequence of the results of \cite{CHMPV}, namely of Lemma 54 together with Propositions 11, 12 and 15.
\end{proof}

\begin{lemma}
The following equation is a consequence of the equations of $\QC$:
\begin{equation}\label{mctrlP2piperiodic}
\begin{array}{rcl}
\tikzfigS{./identities/mctrlP2pi}&=&\tikzfigS{./identities/Idn}
\end{array}
\end{equation}
\end{lemma}

\begin{lemma}
The following equations are consequences of the equations of $\QC$ without \cref{Mstar}, together with \cref{mctrlP2piperiodic}:
\begin{equation}\label{mctrlRX2piP}
\begin{array}{rcl}
\tikzfigS{./identities/mctrl2RX2pi}&=&\tikzfigS{./identities/mctrlPpisurId}
\end{array}
\end{equation}

\begin{equation}\label{passagepihb}
\begin{array}{rcl}
\tikzfigS{./identities/passagepihbleft}&=&\tikzfigS{./identities/passagepihbright}
\end{array}
\end{equation}
\end{lemma}

\begin{proof}[Proof of \cref{mctrlP2piperiodic}]
The case with zero controls is a direct consequence of \cref{Paddition,Zdef,ZZ}. For the case with one or more controls, it suffices to remark that the derivation given in \cite{CHMPV} does not need $\delta_9$:
\begin{eqnarray*}\tikzfigS{./identities/mctrl2P2pi}&\eqeqref{mctrlzeroid}&\tikzfigS{./identities/Mstar-left-02pi}\\[0.4cm]
     &\eqeqref{Mstar}&\tikzfigS{./identities/Mstar-right-simp-zeros}\\[0.4cm]
     &\eqeqref{mctrlzeroid}&\tikzfigS{./identities/Idn1}
     \end{eqnarray*}
\end{proof}

\begin{proof}[Proof of \cref{mctrlRX2piP,passagepihb}]
\Cref{mctrlRX2piP,passagepihb} correspond respectively to Lemmas 59 and 60 of \cite{CHMPV}, and it can be checked that their proofs given in \cite{CHMPV} do not need the full power of \cref{Mstarold} but only \cref{mctrlP2piperiodic}.
\end{proof}

Given an instance of \cref{Mstarold}, we prove that its right-hand side can be transformed into a circuit of the same form as the right-hand side of \cref{Mstar}.

First, one can remark that
\[\bra{1...101}\interp{\tikzfigS{./qc-axioms/Mstar-left}}\ket{1...110}=-\sin(\frac{\gamma_3}2)\sin(\frac{\gamma_4}2)\]
while
\[\bra{1...101}\interp{\tikzfigS{./qc-axioms/Mstar-right}}\ket{1...110}=-e^{i\delta_9}\sin(\frac{\delta_3}2)\sin(\frac{\delta_4}2).\]

\begin{itemize}
\item If  $\sin(\frac{\delta_3}2)\sin(\frac{\delta_4}2)\neq0$, this implies that $e^{i\delta_9}$ is a real number, which, since $\delta_9\in[0,2\pi)$, implies that $\delta_9\in\{0,\pi\}$.
\begin{itemize}
\item If $\delta_9=0$, then
\[\tikzfigS{./qc-axioms/Mstar-right}\eqeqref{mctrlzeroid}\tikzfigS{./qc-axioms/Mstar-right-simp}\]
\item If $\delta_9=\pi$ and $\delta_6=0$, then
\begin{eqnarray*}
  &&\tikzfigS{./proof-Mstar/Mstar-step00}\\[0.4cm]
  &\eqeqref{mctrlzeroid}&\tikzfigS{./proof-Mstar/Mstar-step01}\\[0.4cm]
  &\overset{\text{\cref{Palwayscommute}}}{=}&\tikzfigS{./proof-Mstar/Mstar-step02}\\[0.4cm]
  &\eqeqref{mctrlRX2piP}&\tikzfigS{./proof-Mstar/Mstar-step03}\\[0.4cm]
  &\eqeqref{mctrlRXaddition}&\tikzfigS{./proof-Mstar/Mstar-step04}\\[0.4cm]
  &\eqeqref{mctrlzeroid}&\tikzfigS{./proof-Mstar/Mstar-step05}
\end{eqnarray*}

\item If $\delta_9=\pi$ and $\delta_6\neq0$, then
\begin{eqnarray*}
  &&\tikzfigS{./proof-Mstar/Mstar-step0}\\[0.4cm]
  &\overset{\text{\cref{Palwayscommute}}}{=}&\tikzfigS{./proof-Mstar/Mstar-step1}\\[0.4cm]
  &\eqeqrefwcontrol&\tikzfigS{./proof-Mstar/Mstar-step2}\\[0.4cm]
  &\eqeqrefwcontrol&\tikzfigS{./proof-Mstar/Mstar-step3}\\[0.4cm]
  &\eqeqrefwcontrol&\tikzfigS{./proof-Mstar/Mstar-step4}\\[0.4cm]
  &\eqeqref{mctrlRX2piP}&\tikzfigS{./proof-Mstar/Mstar-step5}\\[0.4cm]
  &\eqeqref{mctrlRXaddition}&\tikzfigS{./proof-Mstar/Mstar-step6}\\[0.4cm]
  &\overset{\text{\cref{Palwayscommute}}}{=}&\tikzfigS{./proof-Mstar/Mstar-step7}\\[0.4cm]
  &\eqeqref{passagepihb}&\tikzfigS{./proof-Mstar/Mstar-step8}\\[0.4cm]
  &\overset{\text{\cref{Palwayscommute}}}{=}&\tikzfigS{./proof-Mstar/Mstar-step9}\\[0.4cm]
  &\eqeqref{mctrlPlift}&\tikzfigS{./proof-Mstar/Mstar-step10}\\[0.4cm]
  &\overset{\refwcontrol\eqref{mctrlPaddition}}{=}&\tikzfigS{./proof-Mstar/Mstar-step11}\\[0.4cm]
  &\eqdeuxeqref{mctrlPaddition}{mctrlP2piperiodic}&\tikzfigS{./proof-Mstar/Mstar-step12}
\end{eqnarray*}
\end{itemize}

\item If $\sin(\frac{\delta_3}2)\sin(\frac{\delta_4}2)=0$, then since $\delta_3,\delta_4\in[0,2\pi)$ and $\delta_4=0\Rightarrow\delta_3=0$, one necessarily has $\delta_3=0$.  In turn, the conditions of \cref{Mstarold} also imply that $\delta_2=0$. Additionally, this implies that $\sin(\frac{\gamma_3}2)\sin(\frac{\gamma_4}2)=0$ too, that is, $\sin(\frac{\gamma_3}2)=0$ or $\sin(\frac{\gamma_4}2)=0$. 
\begin{itemize}
\item If $\sin(\frac{\gamma_3}2)=0$, then $\cos(\frac{\gamma_3}2)\in\{-1,1\}$, and one can remark that
\[\bra{1...110}\interp{\tikzfigS{./qc-axioms/Mstar-left}}\ket{1...110}=\cos(\frac{\gamma_3}2)\]
while, since $\delta_3=0$,
\[\bra{1...110}\interp{\tikzfigS{./qc-axioms/Mstar-right}}\ket{1...110}=e^{i\delta_8}\cos(\frac{\delta_6}2).\]
Hence, %on the one hand, 
$\cos(\frac{\delta_6}2)$ has absolute value $1$, which, since $\delta_6\in[0,2\pi)$, implies that $\delta_6=0$.\footnote{Moreover, $e^{i\delta_8}$ is a real number, which, since $\delta_8\in[0,2\pi)$, implies that $\delta_8\in\{0,\pi\}$. Note however that we do not use this property.}
%, and on the other hand, $e^{i\delta_8}$ is a real number, which, since $\delta_8\in[0,2\pi)$, implies that $\delta_8=0$.
In turn, the conditions of \cref{Mstarold} also imply that $\delta_5=0$. Thus,
\[\tikzfigS{./qc-axioms/Mstar-right}\ \eqeqref{mctrlzeroid}\ \tikzfigS{./proof-Mstar/Mstar-casedelta9-step1}\]
\begin{itemize}
\item If $\delta_9\in[0,\pi)$, then
\begin{eqnarray*}
  &&\tikzfigS{./proof-Mstar/Mstar-casedelta9-step1}\\[0.4cm]
  &\overset{\text{\cref{Palwayscommute},}}{\eqeqref{commctrlphaseenhaut}}&\tikzfigS{./proof-Mstar/Mstar-casedelta9-step2}\\[0.4cm]
  &\eqeqref{mctrlzeroid}&\tikzfigS{./proof-Mstar/Mstar-casedelta9-step3}
\end{eqnarray*}
\item If $\delta_9\in[\pi,2\pi)$, then
\begin{eqnarray*}
  &&\tikzfigS{./proof-Mstar/Mstar-casedelta9-step1}\\[0.4cm]
  &\eqeqref{mctrlPaddition}&\tikzfigS{./proof-Mstar/Mstar-casedelta9bis-step2}\\[0.4cm]
  &\overset{\text{\cref{Palwayscommute},}}{\eqeqref{commctrlphaseenhaut}}&\tikzfigS{./proof-Mstar/Mstar-casedelta9bis-step3}\\[0.4cm]
  &\eqeqref{mctrlRX2piP}&\tikzfigS{./proof-Mstar/Mstar-casedelta9bis-step4}\\[0.4cm]
  &\eqeqref{mctrlRXaddition}&\tikzfigS{./proof-Mstar/Mstar-casedelta9bis-step5}\\[0.4cm]
  &\eqeqref{mctrlzeroid}&\tikzfigS{./proof-Mstar/Mstar-casedelta9bis-step6}
\end{eqnarray*}
\end{itemize}
\item If %$\sin(\frac{\gamma_3}2)\neq0$, then 
$\sin(\frac{\gamma_4}2)=0$, %so that
then $\cos(\frac{\gamma_4}2)\in\{-1,1\}$ and one has
\[\bra{1...101}\interp{\tikzfigS{./qc-axioms/Mstar-left}}\ket{1...101}=\pm\cos(\frac{\gamma_1}2)\]
while, since $\delta_2=0$,
\[\bra{1...101}\interp{\tikzfigS{./qc-axioms/Mstar-right}}\ket{1...101}=e^{i\delta_9}\cos(\frac{\delta_4}2).\]
\begin{itemize}
\item If $\cos(\frac{\delta_4}2)\neq0$, then this implies that $e^{i\delta_9}$ is a real number, so that $\delta_9\in\{0,\pi\}$ and we can proceed as in the first case (where $\sin(\frac{\delta_3}2)\sin(\frac{\delta_4}2)\neq0$).
\item If $\cos(\frac{\delta_4}2)=0$, then on the one hand, since $\delta_4\in[0,2\pi)$, one has $\delta_4=\pi$. In turn, since $\delta_3=0$, the conditions of \cref{Mstarold} imply that $\delta_1=0$. On the other hand, $\cos(\frac{\gamma_1}2)=0$ too, so that $\sin(\frac{\gamma_1}2)\in\{-1,1\}$ and
\[\bra{1...101}\interp{\tikzfigS{./qc-axioms/Mstar-left}}\ket{1...111}=\pm i\]
while, since $\delta_2=0$,
\[\bra{1...101}\interp{\tikzfigS{./qc-axioms/Mstar-right}}\ket{1...111}=\pm ie^{i\delta_9}.\]
Hence, $\delta_9\in\{0,\pi\}$ and we can also proceed as in the first case.
\end{itemize}
\end{itemize}
\end{itemize}
Thus, given any instance of \cref{Mstarold}, its right-hand side can be transformed into $\tikzfigS{./qc-axioms/Mstar-right-simp-primes}$ where
\begin{itemize}
\item $\delta'_1=\delta_1$
\item $\delta'_3=\delta_3$
\item $\delta'_5=\delta_5$
\item $\delta'_7=\delta_7$\medskip
\item $\delta'_2=\begin{cases}\delta_2&\text{if $\delta_9\in\{0,\pi\}$}\\\delta_9&\text{if $\delta_9\in(0,\pi)$}\\\delta_9-\pi&\text{if $\delta_9\in(\pi,2\pi)$}\end{cases}$\medskip
\item $\delta'_4=\begin{cases}\delta_4&\text{if $\delta_9\in[0,\pi)$}\\\delta_4+2\pi&\text{if $\delta_9\in[\pi,2\pi)$}\end{cases}$\medskip
\item $\delta'_6=\begin{cases}2\pi-\delta_6&\text{if $\delta_9=\pi$ and $\delta_6\neq0$}\\\delta_6&\text{else}\end{cases}$\medskip
\item $\delta'_8=\begin{cases}\delta_8+\pi\bmod 2\pi&\text{if $\delta_9=\pi$ and $\delta_6\neq0$}\\\delta_8&\text{else}\end{cases}$.
\end{itemize}
These angles satisfy the conditions of \cref{fig:QCaxioms}. Indeed, the only case where we can obtain $\delta'_3=0$ and $\delta'_2\neq0$ in the case distinction above, is the case where $\sin(\frac{\gamma_3}2)=0$, in which $\delta'_6=\delta_6=0$. For the other conditions of \cref{fig:QCaxioms}, the fact that they are satisfied by the $\delta'_k$ follows directly from the fact that the $\delta_j$ satisfy the conditions of \cref{Mstarold}.

The transformation only uses: 
\begin{itemize}
\item \cref{mctrlPaddition,mctrlRXaddition,mctrlPlift}, which are proved in \cite{CHMPV} without using \cref{Mstarold} and therefore are provable in $\QC$ without \cref{Mstar},
\item \cref{Palwayscommute,commctrlphaseenhaut,mctrlzeroid}, which only rely on $\QC$ without \cref{Mstar},
\item \cref{mctrlP2piperiodic}, which can be proved using \cref{Mstarold} instead of \cref{Mstar}
\item and \cref{mctrlRX2piP,passagepihb}, which are consequences of the equations of $\QC$ without \cref{Mstar} together with \cref{mctrlP2piperiodic}.
\end{itemize}
Since \cref{Mstarold} has been proved to be sound in \cite{CHMPV} (that is, for any angles in its left-hand side, there exist a choice of angles in its right-hand side satisfying the conditions of \cref{Mstarold} that makes it sound), and it is easy to see that all equations of $\QC$ except maybe \cref{Mstar} are sound, this implies that \cref{Mstar} is sound.

Moreover, this proves that \cref{Mstarold} is a consequence of the equations of $\QC$.

It remains to prove that for any choice of angles in the LHS of \cref{Mstar}, the choice of angles in its RHS is unique. To this end, we now consider an instance of \cref{Mstar} and transform its RHS into the RHS of the instance of \eqref{Mstarold} with same LHS. To make easier to differentiate between the parameters of the respective RHS of \cref{Mstar,Mstarold}, we keep denoting those of \cref{Mstar} as $\delta'_k$:
\begin{equation}\tag{\ref*{Mstar}}\tikzfigS{./qc-axioms/Mstar-left}\;=\;\tikzfigS{./qc-axioms/Mstar-right-simp-primes}\end{equation}

\begin{itemize}
\item If $\delta'_4\in[0,2\pi)$ and either $\delta'_3\neq0$ or $\delta'_2=0$, then
\[\tikzfigS{./qc-axioms/Mstar-right-simp-primes}\eqeqref{mctrlzeroid}\tikzfigS{./proof-Mstar/Mstar-right-delta9zero-primes}\]
\item If $\delta'_4\in[2\pi,4\pi)$, either $\delta'_3\neq0$ or $\delta'_2=0$, and $\delta'_6=0$, then
\begin{eqnarray*}
  &&\tikzfigS{./proof-Mstar/Mstar-step05rev}\\[0.4cm]
  &\eqeqref{mctrlzeroid}&\tikzfigS{./proof-Mstar/Mstar-step04rev}\\[0.4cm]
  &\eqeqref{mctrlRXaddition}&\tikzfigS{./proof-Mstar/Mstar-step03rev}\\[0.4cm]
  &\eqeqref{mctrlRX2piP}&\tikzfigS{./proof-Mstar/Mstar-step02rev}\\[0.4cm]
  &\overset{\text{\cref{Palwayscommute}}}{=}&\tikzfigS{./proof-Mstar/Mstar-step01rev}\\[0.4cm]
  &\eqeqref{mctrlzeroid}&\tikzfigS{./proof-Mstar/Mstar-step00rev}
\end{eqnarray*}
\item If $\delta'_4\in[2\pi,4\pi)$, either $\delta'_3\neq0$ or $\delta'_2=0$, and $\delta'_6\neq0$, then
\begin{eqnarray*}
  &&\tikzfigS{./qc-axioms/Mstar-right-simp-primes}\\[0.4cm]
  &\eqeqref{mctrlRXaddition}&\tikzfigS{./proof-Mstar/Mstarrev-step1}\\[0.4cm]
  &\eqeqref{mctrlRX2piP}&\tikzfigS{./proof-Mstar/Mstarrev-step2}\\[0.4cm]
  &\overset{\text{\cref{Palwayscommute}}}{=}&\tikzfigS{./proof-Mstar/Mstarrev-step3}\\[0.4cm]
  &\eqeqrefwcontrol&\tikzfigS{./proof-Mstar/Mstarrev-step4}\\[0.4cm]
  &\eqeqrefwcontrol&\tikzfigS{./proof-Mstar/Mstarrev-step5}\\[0.4cm]
  &\eqeqref{passagepihb}&\tikzfigS{./proof-Mstar/Mstarrev-step6}\\[0.4cm]
  &\overset{\text{\cref{Palwayscommute}}}{=}&\tikzfigS{./proof-Mstar/Mstarrev-step7}\\[0.4cm]
  &\eqeqrefwcontrol&\tikzfigS{./proof-Mstar/Mstarrev-step8}\\[0.4cm]
  &\overset{\eqref{mctrlPlift}\refwcontrol}{=}&\tikzfigS{./proof-Mstar/Mstarrev-step9}\\[0.4cm]
  &\overset{\text{\cref{Palwayscommute},}}{\eqeqref{mctrlPaddition}}&\tikzfigS{./proof-Mstar/Mstarrev-step10}\\[0.4cm]
  &\eqdeuxeqref{mctrlPaddition}{mctrlP2piperiodic}&\tikzfigS{./proof-Mstar/Mstarrev-step11}
\end{eqnarray*}
\item If $\delta'_4\in[0,2\pi)$ but $\delta'_3=0$ and $\delta'_2\neq0$, then the conditions of \cref{fig:QCaxioms} imply that $\delta'_6=0$, so that
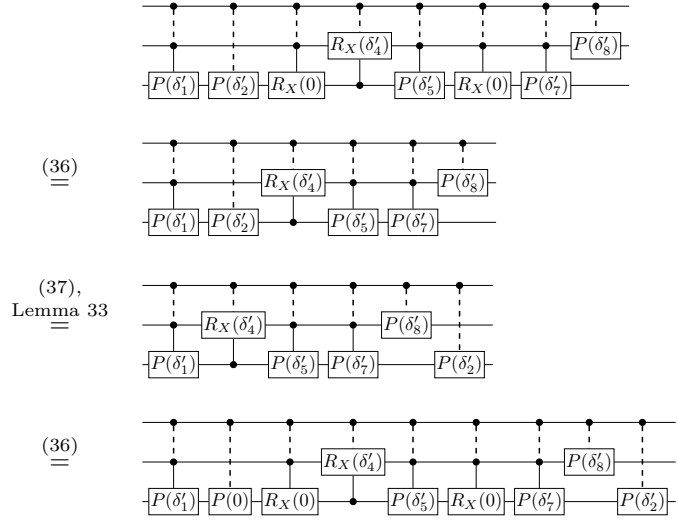
\begin{eqnarray*}
  &&\tikzfigS{./proof-Mstar/Mstar-casedelta9rev-step1}\\[0.4cm]
  &\eqeqref{mctrlzeroid}&\tikzfigS{./proof-Mstar/Mstar-casedelta9rev-step2}\\[0.4cm]
  &\overset{\eqref{commctrlphaseenhaut},}{\overset{\text{\cref{Palwayscommute}}}{=}}&\tikzfigS{./proof-Mstar/Mstar-casedelta9rev-step3}\\[0.4cm]
  &\eqeqref{mctrlzeroid}&\tikzfigS{./proof-Mstar/Mstar-casedelta9rev-step4}
\end{eqnarray*}
\item If $\delta'_4\in[2\pi,4\pi)$, $\delta'_3=0$ and $\delta'_2\neq0$, then the conditions of \cref{fig:QCaxioms} still imply that $\delta'_6=0$, so that
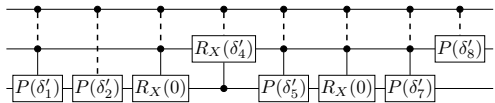
\begin{eqnarray*}
  &&\tikzfigS{./proof-Mstar/Mstar-casedelta9rev-step1}\\[0.4cm]
  &\eqeqref{mctrlzeroid}&\tikzfigS{./proof-Mstar/Mstar-casedelta9rev-step2}\\[0.4cm]
  &\eqeqref{mctrlRXaddition}&\tikzfigS{./proof-Mstar/Mstar-casedelta9revbis-step3}\\[0.4cm]
  &\eqeqref{mctrlRX2piP}&\tikzfigS{./proof-Mstar/Mstar-casedelta9revbis-step4}\\[0.4cm]
  &\eqeqref{mctrlPaddition}&\tikzfigS{./proof-Mstar/Mstar-casedelta9revbis-step5}\\[0.4cm]
  &\overset{\eqref{commctrlphaseenhaut},}{\overset{\text{\cref{Palwayscommute}}}{=}}&\tikzfigS{./proof-Mstar/Mstar-casedelta9revbis-step6}\\[0.4cm]
  &\eqeqref{mctrlzeroid}&\tikzfigS{./proof-Mstar/Mstar-casedelta9revbis-step7}
\end{eqnarray*}
\end{itemize}

Thus, given any instance of \cref{Mstar}, its right-hand side can be transformed into $\tikzfigS{./qc-axioms/Mstar-right}$ where
\begin{itemize}
\item $\delta_1=\delta'_1$
\item $\delta_3=\delta'_3$
\item $\delta_5=\delta'_5$
\item $\delta_7=\delta'_7$\medskip
\item $\delta_2=\begin{cases}0&\text{if $\delta'_3=0$ and $\delta'_2\neq0$}\\\delta'_2&\text{else}\end{cases}$\medskip
\item $\delta_4=\begin{cases}\delta'_4&\text{if $\delta'_4\in[0,2\pi)$}\\\delta'_4-2\pi&\text{if $\delta'_4\in[2\pi,4\pi)$}\end{cases}$\medskip
\item $\delta_6=\begin{cases}2\pi-\delta'_6&\text{if $\delta'_6\neq0$ and $\delta'_4\in[2\pi,4\pi)$}\\\delta'_6&\text{else}\end{cases}$\medskip
\item $\delta_8=\begin{cases}\delta'_8+\pi\bmod 2\pi&\text{if $\delta'_6\neq0$ and $\delta'_4\in[2\pi,4\pi)$}\\\delta'_8&\text{else}\end{cases}$\medskip
\item $\delta_9=\begin{cases}\delta'_2&\text{if $\delta'_3=0$, $\delta'_2\neq0$ and $\delta'_4\in[0,2\pi)$}\\\delta'_2+\pi&\text{if $\delta'_3=0$, $\delta'_2\neq0$ and $\delta'_4\in[2\pi,4\pi)$}\\0&\text{if ($\delta'_3\neq0$ or $\delta'_2=0$) and $\delta'_4\in[0,2\pi)$}\\\pi&\text{if ($\delta'_3\neq0$ or $\delta'_2=0$) and $\delta'_4\in[2\pi,4\pi)$}\end{cases}$.
\end{itemize}
It is easy to check that these angles $\delta_j$ satisfy the conditions of \cref{Mstarold} whenever the $\delta'_j$ satisfy the conditions of \cref{fig:QCaxioms}.

Let $g$ be the function mapping any $8$-tuple of angles $\delta'_k$ corresponding to the RHS of some instance of \cref{Mstar}, to the $9$-tuple of angles $\delta_j$ given by the formulas just above. Conversely, let $f$ be the function mapping any $9$-tuple of angles $\delta_j$ corresponding to the right-hand side of some instance of \cref{Mstarold}, to the $8$-tuple of angles $\delta'_k$ given by the formulas given before. 

Given any $8$-tuple $\vec\delta'\coloneqq(\delta'_k)_{k\in\{1,...,8\}}$ of angles corresponding to the right-hand side of some instance of \cref{Mstar}, one has $f(g(\vec\delta'))=\vec\delta'$. Indeed, let $(\delta_j)_{j\in\{1,...,9\}}\coloneqq g(\vec\delta')$ and $(\delta''_k)_{k\in\{1,...,8\}}\coloneqq f(g(\vec\delta'))$. 
\begin{itemize}
\item By definition, one has $\delta''_k=\delta'_k$ for $j\in\{1,3,5,7\}$.
\item Concerning $\delta''_2$:
\begin{itemize}
\item If $\delta'_3\neq0$ or $\delta'_2=0$, then $\delta_2=\delta'_2$, and $\delta_9$ is either $0$ or $\pi$ depending on $\delta'_4$, so that $\delta''_2=\delta_2=\delta'_2$.
\item If $\delta'_3=0$, $\delta'_2\neq0$ and $\delta'_4\in|0,2\pi)$ then $\delta_9=\delta'_2\in(0,\pi)$, so that $\delta''_2=\delta_9=\delta'_2$.
\item If $\delta'_3=0$, $\delta'_2\neq0$ and $\delta'_4\in|2\pi,4\pi)$ then $\delta_9=\delta'_2+\pi\in(\pi,2\pi)$, so that $\delta''_2=\delta_9-\pi=\delta'_2$.
\end{itemize}
\item Concerning $\delta''_4$:
\begin{itemize}
\item If $\delta'_4\in[0,2\pi)$, then $\delta_4=\delta'_4$, and $\delta_9$ is either $0$ or $\delta'_2$, which in any case is in $[0,\pi)$, so that $\delta''_4=\delta_4=\delta'_4$.
\item If $\delta'_4\in[2\pi,4\pi)$, then $\delta_4=\delta'_4-2\pi$, and $\delta_9$ is either $\pi$ or $\delta'_2+\pi$, which in any case is in $[\pi,2\pi)$, so that $\delta''_4=\delta_4+2\pi=\delta'_4$.
\end{itemize}
\item Concerning $\delta''_6$ and $\delta''_8$:
\begin{itemize}
\item If $\delta'_6=0$, then $\delta_8=\delta'_8$, and $\delta_6=\delta'_6=0$, so that $\delta''_6=\delta_6=\delta'_6\mathrel{(=}0)$ and $\delta''_8=\delta_8=\delta'_8$.
\item If $\delta'_4\in[0,2\pi)$, then on the one hand, $\delta_6=\delta'_6$ and $\delta_8=\delta'_8$, and on the other hand, $\delta_9$ is either $0$ or $\delta'_2$, which in any case cannot be equal to $\pi$, so that $\delta''_6=\delta_6=\delta'_6$ and $\delta''_8=\delta_8=\delta'_8$.
\item If $\delta'_6\neq0$ and $\delta'_4\in[2\pi,4\pi)$, then the conditions of \cref{fig:QCaxioms} imply that we cannot have both $\delta'_3=0$ and $\delta'_2\neq0$, so that $\delta_9=\pi$, $\delta_6=2\pi-\delta'_6$, $\delta''_6=2\pi-\delta_6=\delta'_6$, $\delta_8=\delta'_8+\pi\bmod 2\pi$, and since $\delta'_8\in[0,2\pi)$, $\delta''_8=\delta_8+\pi\bmod 2\pi=\delta'_8$.
\end{itemize}
\end{itemize}
Thus, given any instance of \cref{Mstarold}, the $8$-tuple $\vec\delta'$ of the angles of its RHS satisfies $\vec\delta'=f(g(\vec\delta'))$. Since we have proved that the $9$-tuple $g(\vec\delta')$ corresponds to the angles of the RHS of the instance of \cref{Mstarold} with same LHS, and it has been proved in \cite{CHMPV} that this $9$-tuple is uniquely determined by the LHS, this proves that the $8$-tuple $\vec\delta'$ is uniquely determined by the LHS as well.

%%%%%%%%%%%%%%%%%%%%%%%%%%%%%%%%%%%%%%%%%%%%%%%%%%%%%%%%%%%%%%%%%%%%%%%%%%%%%%%
%%%%%%%%%%%%%%%%%%%%%%%%%%%%%%%%%%%%%%%%%%%%%%%%%%%%%%%%%%%%%%%%%%%%%%%%%%%%%%%
\section{Completeness of $\QCiso$}
\label{appendix:proof-qciso}

\subsection{Proof of \Cref{lem:InitCtrl}}\label{appendix:proofInitCtrl}
\begin{proof}[Proof of \Cref{lem:InitCtrl}]
By assumption, $C$ is a $n\to n+1$ $\propQCiso$-circuit. The only generator of $\propQCiso$ that does not preserve the number of qubits is qubit initialisation $\vdash$. As there is no generator that reduces the number of qubit, there is exactly one $\vdash$ in the circuit. Using the axioms of prop, we can pull this qubit initialisation to the top left so as to get $\scalebox{0.7}{\begin{tikzpicture}
	\begin{pgfonlayer}{nodelayer}
		\node [style=none] (5) at (-8.25, -1.25) {};
		\node [style=none] (6) at (-7.25, -1.25) {};
		\node [style=none] (39) at (-8.25, 0.25) {};
		\node [style=none] (40) at (-7.25, 0.25) {};
		\node [style=none] (41) at (-7.25, 1.5) {};
		\node [style=none] (42) at (-7.25, -1.5) {};
		\node [style=none] (43) at (-5.75, 1.5) {};
		\node [style=none] (44) at (-5.75, -1.5) {};
		\node [style=none] (47) at (-5.75, -1.25) {};
		\node [style=none] (48) at (-4.9, -1.25) {};
		\node [style=none] (49) at (-5.75, 0.25) {};
		\node [style=none] (50) at (-4.9, 0.25) {};
		\node [style=none] (53) at (-7.7, -0.325) {$\vdots$};
		\node [style=none] (54) at (-5.275, -0.325) {$\vdots$};
		\node [style=none] (55) at (-6.5, 0) {$C$};
		\node [style=none] (56) at (-5.75, 1.25) {};
		\node [style=none] (57) at (-4.9, 1.25) {};
	\end{pgfonlayer}
	\begin{pgfonlayer}{edgelayer}
		\draw (5.center) to (6.center);
		\draw (39.center) to (40.center);
		\draw (42.center) to (41.center);
		\draw (43.center) to (41.center);
		\draw (44.center) to (42.center);
		\draw (44.center) to (43.center);
		\draw (47.center) to (48.center);
		\draw (49.center) to (50.center);
		\draw (56.center) to (57.center);
	\end{pgfonlayer}
\end{tikzpicture}} = \scalebox{0.7}{\begin{tikzpicture}
	\begin{pgfonlayer}{nodelayer}
		\node [style=none] (0) at (-2.5, -1.25) {};
		\node [style=none] (1) at (-1.5, -1.25) {};
		\node [style=none] (2) at (-2.5, 0.25) {};
		\node [style=none] (3) at (-1.5, 0.25) {};
		\node [style=none] (4) at (-1.5, 1.5) {};
		\node [style=none] (5) at (-1.5, -1.5) {};
		\node [style=none] (6) at (0, 1.5) {};
		\node [style=none] (7) at (0, -1.5) {};
		\node [style=none] (8) at (0, -1.25) {};
		\node [style=none] (9) at (0.85, -1.25) {};
		\node [style=none] (10) at (0, 0.25) {};
		\node [style=none] (11) at (0.85, 0.25) {};
		\node [style=none] (12) at (-1.95, -0.325) {$\vdots$};
		\node [style=none] (13) at (0.475, -0.325) {$\vdots$};
		\node [style=none] (14) at (-0.75, 0) {$C'$};
		\node [style=none] (15) at (0, 1.25) {};
		\node [style=none] (16) at (0.85, 1.25) {};
		\node [style=ancilla] (17) at (-2.15, 1.25) {};
		\node [style=none] (18) at (-1.5, 1.25) {};
	\end{pgfonlayer}
	\begin{pgfonlayer}{edgelayer}
		\draw (0.center) to (1.center);
		\draw (2.center) to (3.center);
		\draw (5.center) to (4.center);
		\draw (6.center) to (4.center);
		\draw (7.center) to (5.center);
		\draw (7.center) to (6.center);
		\draw (8.center) to (9.center);
		\draw (10.center) to (11.center);
		\draw (15.center) to (16.center);
		\draw (17) to (18.center);
	\end{pgfonlayer}
\end{tikzpicture}}$ where $C'$ is a $\propQC$-circuit.

Let $U = \interp{C'}$. Since $U(\ket0\otimes Id) = \ket0\otimes Id$, $U$ is of the form $U = \left(\begin{array}{c|c}I & 0\\ \hline 0 & U'\end{array}\right)$ with $U'$ unitary. By universality of $\propQC$-circuits (\Cref{prop:QC-universal}), there exists a $\propQC$-circuit $C_{U'}$ that implements $U'$, using only global phase, phases, Hadamards, CNots and swaps. 
One can apply the following transformations $\scalebox{0.7}{\begin{tikzpicture}
	\begin{pgfonlayer}{nodelayer}
		\node [style=none] (4) at (-0.75, 0) {};
		\node [style=none] (5) at (0.75, 0) {};
		\node [style=gate] (12) at (0, 0) {$H$};
	\end{pgfonlayer}
	\begin{pgfonlayer}{edgelayer}
		\draw (4.center) to (5.center);
	\end{pgfonlayer}
\end{tikzpicture}
} \to \scalebox{0.7}{\begin{tikzpicture}
	\begin{pgfonlayer}{nodelayer}
		\node [style=none] (1) at (3.75, 0) {};
		\node [style=none] (2) at (-3.75, 0) {};
		\node [style=gate] (5) at (-2.5, 0) {$P(\frac{\pi}{2})$};
		\node [style=gate] (6) at (0, 0) {$R_X(\frac \pi 2)$};
		\node [style=gate] (7) at (2.5, 0) {$P(\frac{\pi}{2})$};
		\node [style=void] (14) at (-3.75, 0.75) {};
		\node [style=void] (15) at (-3.75, -0.75) {};
	\end{pgfonlayer}
	\begin{pgfonlayer}{edgelayer}
		\draw (2.center) to (1.center);
	\end{pgfonlayer}
\end{tikzpicture}}$ and $\scalebox{0.7}{\begin{tikzpicture}
	\begin{pgfonlayer}{nodelayer}
		\node [style=none] (2) at (-0.75, -0.75) {};
		\node [style=none] (4) at (0.75, -0.75) {};
		\node [style=none] (5) at (-0.75, 0.75) {};
		\node [style=none] (6) at (0.75, 0.75) {};
		\node [style=none] (7) at (-1.25, 0.75) {};
		\node [style=none] (8) at (1.25, 0.75) {};
		\node [style=none] (9) at (1.25, -0.75) {};
		\node [style=none] (10) at (-1.25, -0.75) {};
		\node [style=void] (13) at (-1.25, 1.5) {};
		\node [style=void] (14) at (-1.25, -1.5) {};
	\end{pgfonlayer}
	\begin{pgfonlayer}{edgelayer}
		\draw (7.center) to (5.center);
		\draw (10.center) to (2.center);
		\draw (6.center) to (8.center);
		\draw (4.center) to (9.center);
		\draw [in=-180, out=0] (5.center) to (4.center);
		\draw [in=-180, out=0] (2.center) to (6.center);
	\end{pgfonlayer}
\end{tikzpicture}}\to \scalebox{0.7}{\begin{tikzpicture}
	\begin{pgfonlayer}{nodelayer}
		\node [style=none] (4) at (1.75, 0.75) {};
		\node [style=none] (5) at (-1.75, 0.75) {};
		\node [style=none] (6) at (1.75, -0.75) {};
		\node [style=none] (7) at (-1.75, -0.75) {};
		\node [style=control] (19) at (-1, 0.75) {};
		\node [style=not] (20) at (-1, -0.75) {};
		\node [style=not] (22) at (1, -0.75) {};
		\node [style=control] (23) at (0, -0.75) {};
		\node [style=not] (24) at (0, 0.75) {};
		\node [style=void] (25) at (-1.75, 1.5) {};
		\node [style=void] (26) at (-1.75, -1.5) {};
		\node [style=control] (27) at (1, 0.75) {};
	\end{pgfonlayer}
	\begin{pgfonlayer}{edgelayer}
		\draw (4.center) to (5.center);
		\draw (6.center) to (7.center);
		\draw (19) to (20);
		\draw (23) to (24);
		\draw (27) to (22);
	\end{pgfonlayer}
\end{tikzpicture}}$, it leads to $H$-free, swap-free circuit $\tilde C_{U'}$ provably equivalent (with QC) to $C_{U'}$.

By controlling each of the gates constituting $\tilde C_{U'}$ (using definitions in \Cref{fig:shortcutcircuits}, and taking $P(\varphi)$ as the control of global phase $\varphi$) with a fresh qubit, we get a circuit $\Lambda \tilde C_{U'}$ such that $\interp{\Lambda \tilde C_{U'}} = \interp{C'}$, and where the fresh qubit only sees Phases $P$ and the control part of some gates. By completeness of QC we have $\QC \vdash C'= \Lambda \tilde C_U $ so   $\QCiso \vdash \scalebox{0.7}{\begin{tikzpicture}
	\begin{pgfonlayer}{nodelayer}
		\node [style=none] (5) at (-8.25, -1.25) {};
		\node [style=none] (6) at (-7.25, -1.25) {};
		\node [style=none] (39) at (-8.25, 0.25) {};
		\node [style=none] (40) at (-7.25, 0.25) {};
		\node [style=none] (41) at (-7.25, 1.5) {};
		\node [style=none] (42) at (-7.25, -1.5) {};
		\node [style=none] (43) at (-5.75, 1.5) {};
		\node [style=none] (44) at (-5.75, -1.5) {};
		\node [style=none] (47) at (-5.75, -1.25) {};
		\node [style=none] (48) at (-4.9, -1.25) {};
		\node [style=none] (49) at (-5.75, 0.25) {};
		\node [style=none] (50) at (-4.9, 0.25) {};
		\node [style=none] (53) at (-7.7, -0.325) {$\vdots$};
		\node [style=none] (54) at (-5.275, -0.325) {$\vdots$};
		\node [style=none] (55) at (-6.5, 0) {$C$};
		\node [style=none] (56) at (-5.75, 1.25) {};
		\node [style=none] (57) at (-4.9, 1.25) {};
	\end{pgfonlayer}
	\begin{pgfonlayer}{edgelayer}
		\draw (5.center) to (6.center);
		\draw (39.center) to (40.center);
		\draw (42.center) to (41.center);
		\draw (43.center) to (41.center);
		\draw (44.center) to (42.center);
		\draw (44.center) to (43.center);
		\draw (47.center) to (48.center);
		\draw (49.center) to (50.center);
		\draw (56.center) to (57.center);
	\end{pgfonlayer}
\end{tikzpicture}} =\scalebox{0.7}{\begin{tikzpicture}
	\begin{pgfonlayer}{nodelayer}
		\node [style=none] (5) at (-8.25, -1.25) {};
		\node [style=none] (6) at (-7.25, -1.25) {};
		\node [style=ancilla] (37) at (-8, 1.25) {};
		\node [style=none] (38) at (-7.25, 1.25) {};
		\node [style=none] (39) at (-8.25, 0.25) {};
		\node [style=none] (40) at (-7.25, 0.25) {};
		\node [style=none] (41) at (-7.25, 1.5) {};
		\node [style=none] (42) at (-7.25, -1.5) {};
		\node [style=none] (43) at (-5.5, 1.5) {};
		\node [style=none] (44) at (-5.5, -1.5) {};
		\node [style=none] (47) at (-5.5, -1.25) {};
		\node [style=none] (48) at (-4.65, -1.25) {};
		\node [style=none] (49) at (-5.5, 0.25) {};
		\node [style=none] (50) at (-4.65, 0.25) {};
		\node [style=none] (53) at (-7.7, -0.325) {$\vdots$};
		\node [style=none] (54) at (-5.025, -0.325) {$\vdots$};
		\node [style=none] (55) at (-6.375, 0) {$\Lambda \tilde C_{U'}$};
		\node [style=none] (56) at (-5.5, 1.25) {};
		\node [style=none] (57) at (-4.65, 1.25) {};
	\end{pgfonlayer}
	\begin{pgfonlayer}{edgelayer}
		\draw (5.center) to (6.center);
		\draw (37) to (38.center);
		\draw (39.center) to (40.center);
		\draw (42.center) to (41.center);
		\draw (43.center) to (41.center);
		\draw (44.center) to (42.center);
		\draw (44.center) to (43.center);
		\draw (47.center) to (48.center);
		\draw (49.center) to (50.center);
		\draw (56.center) to (57.center);
	\end{pgfonlayer}
\end{tikzpicture}}$. 
Thus, one can push the initialisation of $\scalebox{0.7}{\begin{tikzpicture}
	\begin{pgfonlayer}{nodelayer}
		\node [style=none] (5) at (-8.25, -1.25) {};
		\node [style=none] (6) at (-7.25, -1.25) {};
		\node [style=ancilla] (37) at (-8, 1.25) {};
		\node [style=none] (38) at (-7.25, 1.25) {};
		\node [style=none] (39) at (-8.25, 0.25) {};
		\node [style=none] (40) at (-7.25, 0.25) {};
		\node [style=none] (41) at (-7.25, 1.5) {};
		\node [style=none] (42) at (-7.25, -1.5) {};
		\node [style=none] (43) at (-5.5, 1.5) {};
		\node [style=none] (44) at (-5.5, -1.5) {};
		\node [style=none] (47) at (-5.5, -1.25) {};
		\node [style=none] (48) at (-4.65, -1.25) {};
		\node [style=none] (49) at (-5.5, 0.25) {};
		\node [style=none] (50) at (-4.65, 0.25) {};
		\node [style=none] (53) at (-7.7, -0.325) {$\vdots$};
		\node [style=none] (54) at (-5.025, -0.325) {$\vdots$};
		\node [style=none] (55) at (-6.375, 0) {$\Lambda \tilde C_{U'}$};
		\node [style=none] (56) at (-5.5, 1.25) {};
		\node [style=none] (57) at (-4.65, 1.25) {};
	\end{pgfonlayer}
	\begin{pgfonlayer}{edgelayer}
		\draw (5.center) to (6.center);
		\draw (37) to (38.center);
		\draw (39.center) to (40.center);
		\draw (42.center) to (41.center);
		\draw (43.center) to (41.center);
		\draw (44.center) to (42.center);
		\draw (44.center) to (43.center);
		\draw (47.center) to (48.center);
		\draw (49.center) to (50.center);
		\draw (56.center) to (57.center);
	\end{pgfonlayer}
\end{tikzpicture}}$ through all the (controlled) gates  using Equations~\ref{ancillamctrlPneg}, \ref{ancillamctrlRXneg} and \ref{ancillaTOFneg}, leading to  $\scalebox{0.7}{\begin{tikzpicture}
	\begin{pgfonlayer}{nodelayer}
		\node [style=none] (0) at (-3, -1.25) {};
		\node [style=ancilla] (3) at (-2.75, 1.25) {};
		\node [style=none] (4) at (-1, 1.25) {};
		\node [style=none] (5) at (-3, 0.25) {};
		\node [style=none] (12) at (-1, -1.25) {};
		\node [style=none] (14) at (-1, 0.25) {};
		\node [style=none] (16) at (-1.95, -0.325) {$\vdots$};
	\end{pgfonlayer}
	\begin{pgfonlayer}{edgelayer}
		\draw (3) to (4.center);
		\draw (5.center) to (14.center);
		\draw (12.center) to (0.center);
	\end{pgfonlayer}
\end{tikzpicture}}$.
\end{proof}

\subsection{Proof of \Cref{lem:CSD}}\label{appendix:proofCSD}
To show \Cref{lem:CSD}, we first introduce useful known decompositions, and recall as well the usual (balanced) CSD for reference:
\begin{lemma}[Matrix Decompositions \cite{Trefethen1997numerical}]~
\begin{itemize}
\item \textbf{RQ (QL):} Let $A$ be a square matrix. There exists $Q$ unitary and $R$ upper triangular (with non-negative diagonal coefficients) such that $A=RQ$. There exists $Q'$ unitary and $L$ lower triangular (with non-negative diagonal coefficients) such that $A=Q'L$.
\item \textbf{Singular Value (SVD):} For any matrix $A$, there exist $U$ and $V$ unitary, and $D=\operatorname{diag}(d_1,...,d_n)$ real diagonal with $d_i\geq d_{i+1}\geq0$, such that $A = UDV$.\\
$\left(\begin{array}{c|c}I & 0\\\hline 0 & U\end{array}\right)
\left(\begin{array}{c|c}I & 0\\\hline 0 & D\end{array}\right)
\left(\begin{array}{c|c}I & 0\\\hline 0 & V\end{array}\right)$
is then an SVD of $\left(\begin{array}{c|c}I & 0\\\hline 0 & A\end{array}\right)$.
\item \textbf{(balanced) Cosine-Sine (CSD):} Let $U=\left(\begin{array}{c|c}U_{00} & U_{01}\\\hline U_{10} & U_{11}\end{array}\right)$ be unitary with $U_{ij}$ all of the same dimension. Then, there exist 
$A_0$, $A_1$, $B_0$, $B_1$ unitary,
$C = \operatorname{diag}(c_0,...,c_n)$ and $S = \operatorname{diag}(s_0,...,s_n)$ such that
$U = \left(\begin{array}{c|c}A_0 & 0\\\hline 0 & A_1\end{array}\right)
\left(\begin{array}{c|c}C & -S\\\hline S & C\end{array}\right)
\left(\begin{array}{c|c}B_0 & 0\\\hline 0 & B_1\end{array}\right)$
and $C^2+S^2 = I$.
\end{itemize}
\end{lemma}
A more general version of the CSD exists for ``unbalanced'' partitions of $U$ i.e.~when the $U_{ij}$ do not not have the same dimensions, but we will not use it in this paper.

\begin{proof}[Proof of \Cref{lem:CSD}]
This is a small variation on the usual CSD. Let us start with
$U=\left(\begin{array}{c|c|c}I&0&0\\\hline 0 & U_{00} & U_{01}\\\hline 0&U_{10} & ~~\raisebox{-0.5em}{\vphantom{\rule{1pt}{2em}}}U_{11}~~\end{array}\right)$.
Let:
\begin{itemize}
\item $A_0 C_0 B_0$ be an SVD of $U_{00}$,
\item $A_1R$ be a QL decomposition of $\left(\begin{array}{c|c}0 & U_{10}B_0^\dag\end{array}\right)$ and
\item $LB_1'$ be an RQ decomposition of $\left(\begin{array}{c} 0 \\ \hline A_0^\dag U_{01}\end{array}\right)$.
\end{itemize}
The unitarity forces the diagonal components of $C_0$ to be between $0$ and $1$. If $1$s appear, they do so as the first diagonal components, as the SVD sorts then from largest to smallest. We then denote $C$ the submatrix of $C_0$ which has only $<1$ components.
 We then have:
$U = \left(\begin{array}{c|c|c}I&0&0\\\hline 0 & A_0 & 0\\\hline 0&0 & ~~\raisebox{-0.5em}{\vphantom{\rule{1pt}{2em}}}A_1~~\end{array}\right)
\left(\begin{array}{c|c}
\begin{array}{c|c}I&0\\\hline0&C\end{array} &  R \\ \hline
L\raisebox{-0.5em}{\vphantom{\rule{1pt}{2em}}} & A_1^\dag U_{11}B_1'^\dag
\end{array}\right)
\left(\begin{array}{c|c|c}I&0&0\\\hline 0 & B_0 & 0\\\hline 0&0 & ~~\raisebox{-0.5em}{\vphantom{\rule{1pt}{2em}}}B_1'~~\end{array}\right)$.
Notice that since $C\neq C_0$ in general, $A_0$ and $C$ may not be of the same dimensions. 
The orthonormality of the columns (resp.~the rows) of the middle matrix forces it to be of the form: 
$\left(\begin{array}{c|c|c|c}
I&0&0&0 \\ \hline
0&C&0&S'\\ \hline
0&0&X_0&0\\ \hline
0&S&0&C'
\end{array}\right)$ with $C^2+S^2 = C'^2+S^2 = I = C^2+S'^2 = C'^2+S'^2$.
Since both $S$ and $S'$ are non-negative diagonal, this implies $S=S'$, which then, by unitarity, forces $C'=-C$. 
Moreover,by unitarity again, $X_0$ itself is unitary. We hence have 
$\left(\begin{array}{c|c|c|c}
I&0&0&0 \\ \hline
0&C&0&S'\\ \hline
0&0&X_0&0\\ \hline
0&S&0&C'
\end{array}\right)
=\left(\begin{array}{c|c|c|c}
I&0&0&0 \\ \hline
0&C&0&-S\\ \hline
0&0&I&0\\ \hline
0&S&0&C
\end{array}\right)
\left(\begin{array}{c|c|c|c}
I&0&0&0 \\ \hline
0&I&0&0\\ \hline
0&0&X_0&0\\ \hline
0&0&0&-I
\end{array}\right)
$.
Wrapping it all up, we have:
\[U = 
\left(\begin{array}{c|c|c}I&0&0\\\hline 0 & A_0 & 0\\\hline 0&0 & ~~\raisebox{-0.5em}{\vphantom{\rule{1pt}{2em}}}A_1~~\end{array}\right)
\left(\begin{array}{c|c|c|c}
I&0&0&0 \\ \hline
0&C&0&-S\\ \hline
0&0&I&0\\ \hline
0&S&0&C
\end{array}\right)
\left(\begin{array}{c|c|c|c}
I&0&0&0 \\ \hline
0&I&0&0\\ \hline
0&0&X_0&0\\ \hline
0&0&0&-I
\end{array}\right)
\left(\begin{array}{c|c|c}I&0&0\\\hline 0 & B_0 & 0\\\hline 0&0 & ~~\raisebox{-0.5em}{\vphantom{\rule{1pt}{2em}}}B_1'~~\end{array}\right)
\]
which gives the desired result with $B_1:= \left(\begin{array}{c|c}X_0&0\\\hline 0 & -I \end{array}\right)B_1'$.
\end{proof}

\subsection{Proof of \Cref{thm:QCiso-completeness}}\label{appendix:proofQCiso-completeness}
\begin{proof}[Proof of \Cref{thm:QCiso-completeness}]
Let $C_1$ and $C_2$ be two circuits of $\propQCiso(n,n+k)$, such that $\interp{C_1}=\interp{C_2}$. Using the same reasoning as in the proof of \cref{lem:InitCtrl}, there exist two $\propQC$-circuits $C_{u_1}$ and $C_{u_2}$ such that $C_i = \begin{tikzpicture}
	\begin{pgfonlayer}{nodelayer}
		\node [style=none] (0) at (-5.25, -2.25) {};
		\node [style=none] (1) at (-4.25, -2.25) {};
		\node [style=ancilla] (3) at (-5.25, 0.25) {};
		\node [style=none] (4) at (-4.25, 0.25) {};
		\node [style=none] (5) at (-5.25, -0.75) {};
		\node [style=none] (6) at (-4.25, -0.75) {};
		\node [style=none] (7) at (-4.25, 2) {};
		\node [style=none] (8) at (-4.25, -2.5) {};
		\node [style=none] (9) at (-2.75, 2) {};
		\node [style=none] (10) at (-2.75, -2.5) {};
		\node [style=none] (11) at (-2.75, -2.25) {};
		\node [style=none] (12) at (-1.75, -2.25) {};
		\node [style=none] (13) at (-2.75, -0.75) {};
		\node [style=none] (14) at (-1.75, -0.75) {};
		\node [style=none] (16) at (-4.7, -1.325) {$\vdots$};
		\node [style=none] (17) at (-2.275, -1.325) {$\vdots$};
		\node [style=none] (18) at (-3.5, 0) {$C_{u_i}$};
		\node [style=ancilla] (19) at (-5.25, 1.75) {};
		\node [style=none] (20) at (-4.25, 1.75) {};
		\node [style=none] (21) at (-2.75, 0.25) {};
		\node [style=none] (22) at (-1.75, 0.25) {};
		\node [style=none] (23) at (-2.75, 1.75) {};
		\node [style=none] (24) at (-1.75, 1.75) {};
		\node [style=none] (25) at (-2.275, 1.175) {$\vdots$};
		\node [style=none] (26) at (-4.75, 1.175) {$\vdots$};
		\node [style=none] (27) at (-5.125, 1) {$n$};
		\node [style=none] (28) at (-5.125, -1.5) {$k$};
	\end{pgfonlayer}
	\begin{pgfonlayer}{edgelayer}
		\draw (0.center) to (1.center);
		\draw (3) to (4.center);
		\draw (5.center) to (6.center);
		\draw (8.center) to (7.center);
		\draw (9.center) to (7.center);
		\draw (10.center) to (8.center);
		\draw (10.center) to (9.center);
		\draw (11.center) to (12.center);
		\draw (13.center) to (14.center);
		\draw (19) to (20.center);
		\draw (21.center) to (22.center);
		\draw (23.center) to (24.center);
	\end{pgfonlayer}
\end{tikzpicture}$. By completeness of $\QC$, showing that $C_1=C_2$ is equivalent to showing that $\begin{tikzpicture}
	\begin{pgfonlayer}{nodelayer}
		\node [style=none] (0) at (-4, -2) {};
		\node [style=none] (1) at (-3, -2) {};
		\node [style=ancilla] (2) at (-4, 0.5) {};
		\node [style=none] (3) at (-3, 0.5) {};
		\node [style=none] (4) at (-4, -0.5) {};
		\node [style=none] (5) at (-3, -0.5) {};
		\node [style=none] (6) at (-3, 2.25) {};
		\node [style=none] (7) at (-3, -2.25) {};
		\node [style=none] (8) at (-1.5, 2.25) {};
		\node [style=none] (9) at (-1.5, -2.25) {};
		\node [style=none] (10) at (-1.5, -2) {};
		\node [style=none] (11) at (-0.5, -2) {};
		\node [style=none] (12) at (-1.5, -0.5) {};
		\node [style=none] (13) at (-0.5, -0.5) {};
		\node [style=none] (14) at (-3.45, -1.075) {$\vdots$};
		\node [style=none] (15) at (-1.025, -1.075) {$\vdots$};
		\node [style=none] (16) at (-2.25, 0.25) {$C_U$};
		\node [style=ancilla] (17) at (-4, 2) {};
		\node [style=none] (18) at (-3, 2) {};
		\node [style=none] (19) at (-1.5, 0.5) {};
		\node [style=none] (20) at (-0.5, 0.5) {};
		\node [style=none] (21) at (-1.5, 2) {};
		\node [style=none] (22) at (-0.5, 2) {};
		\node [style=none] (23) at (-1.025, 1.425) {$\vdots$};
		\node [style=none] (24) at (-3.5, 1.425) {$\vdots$};
		\node [style=none] (25) at (0.25, 0) {$=$};
		\node [style=none] (26) at (1, -2) {};
		\node [style=none] (27) at (3, -2) {};
		\node [style=ancilla] (28) at (1, 0.5) {};
		\node [style=none] (29) at (3, 0.5) {};
		\node [style=none] (30) at (1, -0.5) {};
		\node [style=none] (31) at (3, -0.5) {};
		\node [style=none] (40) at (2.05, -1.075) {$\vdots$};
		\node [style=ancilla] (43) at (1, 2) {};
		\node [style=none] (44) at (3, 2) {};
		\node [style=none] (50) at (2, 1.425) {$\vdots$};
		\node [style=none] (51) at (-3.875, 1.25) {$n$};
		\node [style=none] (52) at (-3.875, -1.25) {$k$};
		\node [style=none] (53) at (1.5, 1.25) {$n$};
		\node [style=none] (54) at (1.5, -1.25) {$k$};
	\end{pgfonlayer}
	\begin{pgfonlayer}{edgelayer}
		\draw (0.center) to (1.center);
		\draw (2) to (3.center);
		\draw (4.center) to (5.center);
		\draw (7.center) to (6.center);
		\draw (8.center) to (6.center);
		\draw (9.center) to (7.center);
		\draw (9.center) to (8.center);
		\draw (10.center) to (11.center);
		\draw (12.center) to (13.center);
		\draw (17) to (18.center);
		\draw (19.center) to (20.center);
		\draw (21.center) to (22.center);
		\draw (26.center) to (27.center);
		\draw (28) to (29.center);
		\draw (30.center) to (31.center);
		\draw (43) to (44.center);
	\end{pgfonlayer}
\end{tikzpicture}$ for $C_U = C_{u_1}C_{u_2}^\dag$ (\Cref{prop:crossinggate}).

Let us denote $U= \interp{C_U}$. Notice that $U$ has to be such that $U\begin{pmatrix}I\\0\\\vdots\\0\end{pmatrix} = \begin{pmatrix}I\\0\\\vdots\\0\end{pmatrix}$, which, by unitarity, means
$U = \left(\begin{array}{c|c}I & 0\\ \hline 0 & \raisebox{-0.75em}{\vphantom{\rule{1pt}{2em}}}~\star~\end{array}\right)$.
Let us now prove that $\QCiso$ proves the equality, by reasoning inductively on the number $n$ of initialised qubits.

\noindent\textbf{Case $\bm{n=0}$:} In that case, there is no initialised qubit, and since $\interp{U}=I$, by completeness of $\QC$, $\QC\vdash U = Id$.

\noindent\textbf{Case $\bm{n=1}$:} There, we have $U = \left(\begin{array}{c|c}I & 0\\ \hline 0 & U'\end{array}\right)$. % which means $C_U$ is of the form $\tikzfig{./qciso-completeness/ctrl-U}$. 
Then \Cref{lem:InitCtrl} gives directly the expected equality.

\noindent\textbf{Case $\bm{n+1}$:} $U$ is of the form $U=\left(\begin{array}{c|c|c}I&0&0\\\hline 0 & U_{00} & U_{01}\\\hline 0&U_{10} & ~~\raisebox{-0.5em}{\vphantom{\rule{1pt}{2em}}}U_{11}~~\end{array}\right)$ (with $U_{00}$ and $U_{11}$ square). Notice that $U$ has dimension $2^{k+n+1}$. We can hence use the modified CSD to get:\\
$U = \left(\begin{array}{c|c|c}I&0&0\\\hline 0 & A_0 & 0\\\hline 0&0 & ~~\raisebox{-0.5em}{\vphantom{\rule{1pt}{2em}}}A_1~~\end{array}\right)
\left(\begin{array}{c|c|c|c}
I&0&0&0 \\ \hline
0&C&0&-S\\ \hline
0&0&I&0\\ \hline
0&S&0&C
\end{array}\right)
\left(\begin{array}{c|c|c}I&0&0\\\hline 0 & B_0 & 0\\\hline 0&0 & ~~\raisebox{-0.5em}{\vphantom{\rule{1pt}{2em}}}B_1~~\end{array}\right)$
with $C^2+S^2=I$.

The middle matrix can be seen as a product of:
\begin{align*}
R_j :=& c_j\ketbra jj + s_j\ketbra {j+10...0}j - s_j\ketbra j{j+10...0} + c_j\ketbra{j+10...0}{j+10...0}\\
&+ \sum_{x\notin\{j, j+10...0\} }\ketbra xx
\end{align*}
for $\overbrace{0...0}^{n+1}\overbrace{1...1}^k < j < 1\overbrace{0...0}^{n+k}$ in binary. Notice that the first $R_j$s might be the identity, if $A_0$ and $C$ do not have the same dimensions. Notice also that $j$ has at least one $1$ in its first $n+1$ bits. $R_j$ is hence a rotation $\begin{pmatrix}c_j & -s_j\\ s_j & c_j\end{pmatrix}=P(\frac\pi2)R_X(\theta_j)P(\text{-}\frac\pi2)$ on the first qubit, controlled by all the other qubits. Matrix $U$ can hence be implemented by the following circuit:
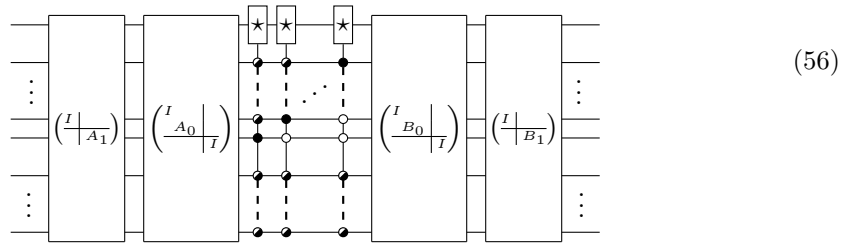
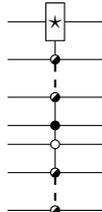
\begin{equation}
\begin{tikzpicture}
	\begin{pgfonlayer}{nodelayer}
		\node [style=none] (0) at (-6, -4.5) {};
		\node [style=none] (1) at (-5, -4.5) {};
		\node [style=none] (3) at (0, 1) {};
		\node [style=none] (4) at (3.5, 1) {};
		\node [style=none] (5) at (-6, 0) {};
		\node [style=none] (6) at (-5, 0) {};
		\node [style=none] (7) at (-5, 1.25) {};
		\node [style=none] (8) at (-5, -4.75) {};
		\node [style=none] (9) at (-3, 1.25) {};
		\node [style=none] (10) at (-3, -4.75) {};
		\node [style=none] (11) at (-3, -4.5) {};
		\node [style=none] (12) at (-2.5, -4.5) {};
		\node [style=none] (13) at (-3, 0) {};
		\node [style=none] (14) at (-2.5, 0) {};
		\node [style=none] (16) at (-5.45, -0.575) {$\vdots$};
		\node [style=none] (17) at (-5.525, -3.575) {$\vdots$};
		\node [style=none] (22) at (-2.5, 1.25) {};
		\node [style=none] (23) at (-2.5, -4.75) {};
		\node [style=none] (24) at (0, 1.25) {};
		\node [style=none] (25) at (0, -4.75) {};
		\node [style=none, font={\tiny}] (29) at (-1.25, -1.75) {$\left(\begin{array}{@{}c@{}c@{\;}|@{\;}c@{}}I&&\\&A_0&\\\hline&&I\end{array}\right)$};
		\node [style=none] (30) at (0, -4.5) {};
		\node [style=none] (32) at (0, 0) {};
		\node [style=none] (34) at (-6, -2) {};
		\node [style=none] (35) at (-5, -2) {};
		\node [style=none] (36) at (-6, -3) {};
		\node [style=none] (37) at (-5, -3) {};
		\node [style=none] (38) at (-3, -2) {};
		\node [style=none] (39) at (-2.5, -2) {};
		\node [style=none] (40) at (-3, -3) {};
		\node [style=none] (41) at (-2.5, -3) {};
		\node [style=none] (42) at (0, -2) {};
		\node [style=none] (44) at (0, -3) {};
		\node [style=gate] (46) at (0.5, 1) {$\star$};
		\node [style={bwcontrol={}}] (47) at (0.5, -3) {};
		\node [style={bwcontrol={}}] (48) at (0.5, -4.5) {};
		\node [style={bwcontrol={}}] (49) at (0.5, 0) {};
		\node [style=control] (50) at (0.5, -2) {};
		\node [style=none] (51) at (-6, -1.5) {};
		\node [style=none] (52) at (-5, -1.5) {};
		\node [style=none] (53) at (-3, -1.5) {};
		\node [style=none] (54) at (-2.5, -1.5) {};
		\node [style=none] (55) at (0, -1.5) {};
		\node [style={bwcontrol={}}] (57) at (0.5, -1.5) {};
		\node [style=gate] (58) at (1.25, 1) {$\star$};
		\node [style={bwcontrol={}}] (59) at (1.25, -3) {};
		\node [style={bwcontrol={}}] (60) at (1.25, -4.5) {};
		\node [style={bwcontrol={}}] (61) at (1.25, 0) {};
		\node [style=control] (62) at (1.25, -1.5) {};
		\node [style=wcontrol] (63) at (1.25, -2) {};
		\node [style=gate] (64) at (2.75, 1) {$\star$};
		\node [style={bwcontrol={}}] (65) at (2.75, -3) {};
		\node [style={bwcontrol={}}] (66) at (2.75, -4.5) {};
		\node [style=wcontrol] (67) at (2.75, -1.5) {};
		\node [style=control] (68) at (2.75, 0) {};
		\node [style=wcontrol] (69) at (2.75, -2) {};
		\node [style=none] (70) at (3.5, -4.5) {};
		\node [style=none] (72) at (3.5, 0) {};
		\node [style=none] (92) at (3.5, -2) {};
		\node [style=none] (93) at (3.5, -3) {};
		\node [style=none] (100) at (3.5, -1.5) {};
		\node [style=none] (116) at (2, -0.625) {\reflectbox{$\ddots$}};
		\node [style=none, font={\tiny}] (117) at (-4, -1.75) {$\left(\begin{array}{@{}c@{\;}|@{\;}c@{}}I&\\\hline&A_1\end{array}\right)$};
		\node [style=none] (118) at (-6, 1) {};
		\node [style=none] (119) at (-5, 1) {};
		\node [style=none] (120) at (-3, 1) {};
		\node [style=none] (121) at (-2.5, 1) {};
		\node [style=none] (122) at (9.5, -4.5) {};
		\node [style=none] (123) at (8.5, -4.5) {};
		\node [style=none] (125) at (9.5, 0) {};
		\node [style=none] (126) at (8.5, 0) {};
		\node [style=none] (127) at (8.5, 1.25) {};
		\node [style=none] (128) at (8.5, -4.75) {};
		\node [style=none] (129) at (6.5, 1.25) {};
		\node [style=none] (130) at (6.5, -4.75) {};
		\node [style=none] (131) at (6.5, -4.5) {};
		\node [style=none] (132) at (6, -4.5) {};
		\node [style=none] (133) at (6.5, 0) {};
		\node [style=none] (134) at (6, 0) {};
		\node [style=none] (135) at (8.95, -0.575) {$\vdots$};
		\node [style=none] (136) at (9.025, -3.575) {$\vdots$};
		\node [style=none] (137) at (6, 1.25) {};
		\node [style=none] (138) at (6, -4.75) {};
		\node [style=none] (139) at (3.5, 1.25) {};
		\node [style=none] (140) at (3.5, -4.75) {};
		\node [style=none, font={\tiny}] (141) at (4.75, -1.75) {$\left(\begin{array}{@{}c@{}c@{\;}|@{\;}c@{}}I&&\\&B_0&\\\hline&&I\end{array}\right)$};
		\node [style=none] (144) at (9.5, -2) {};
		\node [style=none] (145) at (8.5, -2) {};
		\node [style=none] (146) at (9.5, -3) {};
		\node [style=none] (147) at (8.5, -3) {};
		\node [style=none] (148) at (6.5, -2) {};
		\node [style=none] (149) at (6, -2) {};
		\node [style=none] (150) at (6.5, -3) {};
		\node [style=none] (151) at (6, -3) {};
		\node [style=none] (154) at (9.5, -1.5) {};
		\node [style=none] (155) at (8.5, -1.5) {};
		\node [style=none] (156) at (6.5, -1.5) {};
		\node [style=none] (157) at (6, -1.5) {};
		\node [style=none, font={\tiny}] (159) at (7.5, -1.75) {$\left(\begin{array}{@{}c@{\;}|@{\;}c@{}}I&\\\hline&B_1\end{array}\right)$};
		\node [style=none] (160) at (9.5, 1) {};
		\node [style=none] (161) at (8.5, 1) {};
		\node [style=none] (162) at (6.5, 1) {};
		\node [style=none] (163) at (6, 1) {};
	\end{pgfonlayer}
	\begin{pgfonlayer}{edgelayer}
		\draw (0.center) to (1.center);
		\draw (3.center) to (4.center);
		\draw (5.center) to (6.center);
		\draw (8.center) to (7.center);
		\draw (9.center) to (7.center);
		\draw (10.center) to (8.center);
		\draw (10.center) to (9.center);
		\draw (11.center) to (12.center);
		\draw (13.center) to (14.center);
		\draw (23.center) to (22.center);
		\draw (24.center) to (22.center);
		\draw (25.center) to (23.center);
		\draw (25.center) to (24.center);
		\draw (34.center) to (35.center);
		\draw (36.center) to (37.center);
		\draw (38.center) to (39.center);
		\draw (40.center) to (41.center);
		\draw (51.center) to (52.center);
		\draw (53.center) to (54.center);
		\draw (30.center) to (70.center);
		\draw (32.center) to (72.center);
		\draw (42.center) to (92.center);
		\draw (44.center) to (93.center);
		\draw (55.center) to (100.center);
		\draw [style=trait] (49) to (57);
		\draw [style=trait] (61) to (62);
		\draw [style=trait] (68) to (67);
		\draw [style=trait] (47) to (48);
		\draw [style=trait] (59) to (60);
		\draw [style=trait] (65) to (66);
		\draw (46) to (49);
		\draw (57) to (47);
		\draw (58) to (61);
		\draw (62) to (59);
		\draw (64) to (68);
		\draw (67) to (65);
		\draw (118.center) to (119.center);
		\draw (120.center) to (121.center);
		\draw (122.center) to (123.center);
		\draw (125.center) to (126.center);
		\draw (128.center) to (127.center);
		\draw (129.center) to (127.center);
		\draw (130.center) to (128.center);
		\draw (130.center) to (129.center);
		\draw (131.center) to (132.center);
		\draw (133.center) to (134.center);
		\draw (138.center) to (137.center);
		\draw (139.center) to (137.center);
		\draw (140.center) to (138.center);
		\draw (140.center) to (139.center);
		\draw (144.center) to (145.center);
		\draw (146.center) to (147.center);
		\draw (148.center) to (149.center);
		\draw (150.center) to (151.center);
		\draw (154.center) to (155.center);
		\draw (156.center) to (157.center);
		\draw (160.center) to (161.center);
		\draw (162.center) to (163.center);
	\end{pgfonlayer}
\end{tikzpicture}\label{eq:U-decomp}
\end{equation}
where, similarly to \cite{Mottonen2006decompositions}, e.g.~$\begin{tikzpicture}
	\begin{pgfonlayer}{nodelayer}
		\node [style=none] (0) at (-3.75, 2.5) {};
		\node [style=none] (1) at (-1.25, 2.5) {};
		\node [style=none] (4) at (-3.75, -2.5) {};
		\node [style=none] (5) at (-3.75, 1.5) {};
		\node [style=none] (6) at (-3.75, -0.75) {};
		\node [style=none] (7) at (-3.75, -1.5) {};
		\node [style=none] (13) at (-3.75, -0.25) {};
		\node [style=gate] (15) at (-2.5, 2.5) {$\star$};
		\node [style={bwcontrol={}}] (16) at (-2.5, -1.5) {};
		\node [style={bwcontrol={}}] (17) at (-2.5, -2.5) {};
		\node [style={bwcontrol={}}] (18) at (-2.5, 1.5) {};
		\node [style=control] (19) at (-2.5, -0.25) {};
		\node [style=wcontrol] (20) at (-2.5, -0.75) {};
		\node [style=none] (27) at (-1.25, -2.5) {};
		\node [style=none] (28) at (-1.25, 1.5) {};
		\node [style=none] (31) at (-1.25, -0.75) {};
		\node [style=none] (32) at (-1.25, -1.5) {};
		\node [style=none] (33) at (-1.25, -0.25) {};
		\node [style=none] (36) at (-4.75, -0.25) {$\ell$};
		\node [style=none] (37) at (-4.5, -0.25) {};
		\node [style=none] (38) at (-4, -0.25) {};
		\node [style=none] (39) at (-3.75, 0.5) {};
		\node [style={bwcontrol={}}] (40) at (-2.5, 0.5) {};
		\node [style=none] (41) at (-1.25, 0.5) {};
	\end{pgfonlayer}
	\begin{pgfonlayer}{edgelayer}
		\draw (0.center) to (1.center);
		\draw (4.center) to (27.center);
		\draw (5.center) to (28.center);
		\draw (6.center) to (31.center);
		\draw (7.center) to (32.center);
		\draw (13.center) to (33.center);
		\draw [style=trait] (16) to (17);
		\draw (15) to (18);
		\draw (19) to (16);
		\draw [style={arrows={->[]}}] (37.center) to (38.center);
		\draw (39.center) to (41.center);
		\draw (40) to (19);
		\draw [style=trait] (18) to (40);
	\end{pgfonlayer}
\end{tikzpicture}$ is a syntactic sugar for the composition of all rotations on the first qubit, controlled by the $\ell$-th qubit and anti-controlled by the $\ell+1$-th qubit:
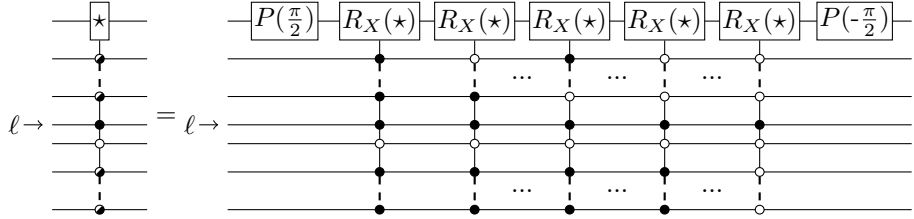
\[\begin{tikzpicture}
	\begin{pgfonlayer}{nodelayer}
		\node [style=none] (0) at (-3.75, 2.5) {};
		\node [style=none] (1) at (-1.25, 2.5) {};
		\node [style=none] (4) at (-3.75, -2.5) {};
		\node [style=none] (5) at (-3.75, 1.5) {};
		\node [style=none] (6) at (-3.75, -0.75) {};
		\node [style=none] (7) at (-3.75, -1.5) {};
		\node [style=none] (13) at (-3.75, -0.25) {};
		\node [style=gate] (15) at (-2.5, 2.5) {$\star$};
		\node [style={bwcontrol={}}] (16) at (-2.5, -1.5) {};
		\node [style={bwcontrol={}}] (17) at (-2.5, -2.5) {};
		\node [style={bwcontrol={}}] (18) at (-2.5, 1.5) {};
		\node [style=control] (19) at (-2.5, -0.25) {};
		\node [style=wcontrol] (20) at (-2.5, -0.75) {};
		\node [style=none] (27) at (-1.25, -2.5) {};
		\node [style=none] (28) at (-1.25, 1.5) {};
		\node [style=none] (31) at (-1.25, -0.75) {};
		\node [style=none] (32) at (-1.25, -1.5) {};
		\node [style=none] (33) at (-1.25, -0.25) {};
		\node [style=none] (36) at (-4.75, -0.25) {$\ell$};
		\node [style=none] (37) at (-4.5, -0.25) {};
		\node [style=none] (38) at (-4, -0.25) {};
		\node [style=none] (39) at (-3.75, 0.5) {};
		\node [style={bwcontrol={}}] (40) at (-2.5, 0.5) {};
		\node [style=none] (41) at (-1.25, 0.5) {};
	\end{pgfonlayer}
	\begin{pgfonlayer}{edgelayer}
		\draw (0.center) to (1.center);
		\draw (4.center) to (27.center);
		\draw (5.center) to (28.center);
		\draw (6.center) to (31.center);
		\draw (7.center) to (32.center);
		\draw (13.center) to (33.center);
		\draw [style=trait] (16) to (17);
		\draw (15) to (18);
		\draw (19) to (16);
		\draw [style={arrows={->[]}}] (37.center) to (38.center);
		\draw (39.center) to (41.center);
		\draw (40) to (19);
		\draw [style=trait] (18) to (40);
	\end{pgfonlayer}
\end{tikzpicture}=\begin{tikzpicture}
	\begin{pgfonlayer}{nodelayer}
		\node [style=none] (0) at (-8, 2.5) {};
		\node [style=none] (1) at (10, 2.5) {};
		\node [style=none] (2) at (-8, -2.5) {};
		\node [style=none] (3) at (-8, 1.5) {};
		\node [style=none] (4) at (-8, -0.75) {};
		\node [style=none] (5) at (-8, -1.5) {};
		\node [style=none] (6) at (-8, -0.25) {};
		\node [style=gate] (7) at (-4, 2.5) {$R_X(\star)$};
		\node [style=control] (8) at (-4, -1.5) {};
		\node [style=control] (9) at (-4, -2.5) {};
		\node [style=control] (10) at (-4, 1.5) {};
		\node [style=control] (11) at (-4, -0.25) {};
		\node [style=wcontrol] (12) at (-4, -0.75) {};
		\node [style=none] (13) at (10, -2.5) {};
		\node [style=none] (14) at (10, 1.5) {};
		\node [style=none] (15) at (10, -0.75) {};
		\node [style=none] (16) at (10, -1.5) {};
		\node [style=none] (17) at (10, -0.25) {};
		\node [style=none] (21) at (-8, 0.5) {};
		\node [style=control] (22) at (-4, 0.5) {};
		\node [style=none] (23) at (10, 0.5) {};
		\node [style=gate] (24) at (-6.5, 2.5) {$P(\frac\pi2)$};
		\node [style=gate] (25) at (-1.5, 2.5) {$R_X(\star)$};
		\node [style=control] (26) at (-1.5, -1.5) {};
		\node [style=control] (27) at (-1.5, -2.5) {};
		\node [style=wcontrol] (28) at (-1.5, 1.5) {};
		\node [style=control] (29) at (-1.5, -0.25) {};
		\node [style=wcontrol] (30) at (-1.5, -0.75) {};
		\node [style=control] (31) at (-1.5, 0.5) {};
		\node [style=gate] (32) at (1, 2.5) {$R_X(\star)$};
		\node [style=control] (33) at (1, -1.5) {};
		\node [style=control] (34) at (1, -2.5) {};
		\node [style=control] (35) at (1, 1.5) {};
		\node [style=control] (36) at (1, -0.25) {};
		\node [style=wcontrol] (37) at (1, -0.75) {};
		\node [style=wcontrol] (38) at (1, 0.5) {};
		\node [style=gate] (39) at (3.5, 2.5) {$R_X(\star)$};
		\node [style=control] (40) at (3.5, -1.5) {};
		\node [style=control] (41) at (3.5, -2.5) {};
		\node [style=wcontrol] (42) at (3.5, 1.5) {};
		\node [style=control] (43) at (3.5, -0.25) {};
		\node [style=wcontrol] (44) at (3.5, -0.75) {};
		\node [style=wcontrol] (45) at (3.5, 0.5) {};
		\node [style=gate] (46) at (6, 2.5) {$R_X(\star)$};
		\node [style=wcontrol] (47) at (6, -1.5) {};
		\node [style=wcontrol] (48) at (6, -2.5) {};
		\node [style=wcontrol] (49) at (6, 1.5) {};
		\node [style=control] (50) at (6, -0.25) {};
		\node [style=wcontrol] (51) at (6, -0.75) {};
		\node [style=wcontrol] (52) at (6, 0.5) {};
		\node [style=gate] (53) at (8.5, 2.5) {$P(\text{-}\frac\pi2)$};
		\node [style=none] (54) at (-9, -0.25) {$\ell$};
		\node [style=none] (55) at (-8.75, -0.25) {};
		\node [style=none] (56) at (-8.25, -0.25) {};
		\node [style=none] (57) at (-0.25, 1) {...};
		\node [style=none] (58) at (-0.25, -2) {...};
		\node [style=none] (59) at (2.25, 1) {...};
		\node [style=none] (60) at (2.25, -2) {...};
		\node [style=none] (61) at (4.75, 1) {...};
		\node [style=none] (62) at (4.75, -2) {...};
	\end{pgfonlayer}
	\begin{pgfonlayer}{edgelayer}
		\draw (0.center) to (1.center);
		\draw (2.center) to (13.center);
		\draw (3.center) to (14.center);
		\draw (4.center) to (15.center);
		\draw (5.center) to (16.center);
		\draw (6.center) to (17.center);
		\draw [style=trait] (8) to (9);
		\draw (7) to (10);
		\draw (11) to (8);
		\draw (21.center) to (23.center);
		\draw (22) to (11);
		\draw [style=trait] (10) to (22);
		\draw [style=trait] (26) to (27);
		\draw (25) to (28);
		\draw (29) to (26);
		\draw (31) to (29);
		\draw [style=trait] (28) to (31);
		\draw [style=trait] (33) to (34);
		\draw (32) to (35);
		\draw (36) to (33);
		\draw (38) to (36);
		\draw [style=trait] (35) to (38);
		\draw [style=trait] (40) to (41);
		\draw (39) to (42);
		\draw (43) to (40);
		\draw (45) to (43);
		\draw [style=trait] (42) to (45);
		\draw [style=trait] (47) to (48);
		\draw (46) to (49);
		\draw (50) to (47);
		\draw (52) to (50);
		\draw [style=trait] (49) to (52);
		\draw [style={arrows={->[]}}] (55.center) to (56.center);
	\end{pgfonlayer}
\end{tikzpicture}\]
(notice that each $R_X$ should be surrounded by $P(\frac\pi2)$ on the left and $P(\text{-}\frac\pi2)$ on the right, but all non-extremal ones simplify using \eqref{mctrlPaddition} and \eqref{axQCP0}.) 
By completeness of $\QC$, $C_U$ can be turned into the circuit in (\ref{eq:U-decomp}).

The first block (with $A_1$) is such that its interpretation satisfies $\forall \ket \varphi \in \C^{2^n}$, $\left(\begin{array}{@{}c@{\;}|@{\;}c@{}}I&\\\hline&A_1\end{array}\right)(\ket0\otimes \ket \varphi) = \ket0\otimes \ket \varphi$, hence using \Cref{lem:InitCtrl}, it is provably deleted by intialisation on the first qubit, and a fortiori when the $n+1$ first qubits are initialised. Similarly, the last block (with $B_1$) is deleted by the qubit initialisations.

For the second block, notice that:
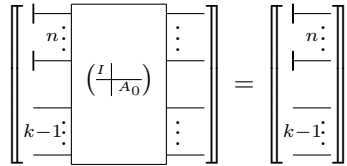
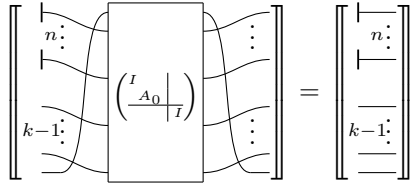
\[\interp{\begin{tikzpicture}
	\begin{pgfonlayer}{nodelayer}
		\node [style=none] (0)  at (1.25, 1.875) {};
		\node [style=none] (1)  at (2.25, 1.875) {};
		\node [style=none] (3)  at (-2.25, -1.875) {};
		\node [style=none] (5)  at (-1.25, -1.875) {};
		\node [style=ancilla] (6)  at (-2.25, 0.625) {};
		\node [style=none] (7)  at (-1.25, 0.625) {};
		\node [style=none] (8)  at (-1.25, 2.125) {};
		\node [style=none] (9)  at (-1.25, -2.125) {};
		\node [style=none] (10)  at (1.25, 2.125) {};
		\node [style=none] (11)  at (1.25, -2.125) {};
		\node [style=none] (13)  at (1.25, -1.875) {};
		\node [style=none] (14)  at (1.25, 0.625) {};
		\node [style=none] (17)  at (-2.25, -0.625) {};
		\node [style=none] (18)  at (-1.25, -0.625) {};
		\node [style=none] (20)  at (1.25, -0.625) {};
		\node [style=none] (42)  at (2.25, -1.875) {};
		\node [style=none] (43)  at (2.25, 0.625) {};
		\node [style=none] (45)  at (2.25, -0.625) {};
		\node [style=ancilla] (48)  at (-2.25, 1.875) {};
		\node [style=none] (49)  at (-1.25, 1.875) {};
		\node [style=none, font={\tiny}] (50)  at (0.0, 0.125) {$\left(\begin{array}{@{}c@{\;}|@{\;}c@{}}I&\\\hline&A_0\end{array}\right)$};
		\node [style=none] (80)  at (-1.45, 1.425) {$\vdots$};
		\node [style=none] (81)  at (1.55, 1.425) {$\vdots$};
		\node [style=none, font={\scriptsize}] (82)  at (-1.75, 1.25) {$n$};
		\node [style=none] (83)  at (-1.45, -1.075) {$\vdots$};
		\node [style=none] (84)  at (1.55, -1.075) {$\vdots$};
		\node [style=none, font={\scriptsize}] (85)  at (-2.0, -1.25) {$k{-}1$};
	\end{pgfonlayer}
	\begin{pgfonlayer}{edgelayer}
		\draw (0.center) to (1.center);
		\draw (5.center) to (3.center);
		\draw (6) to (7.center);
		\draw (9.center) to (8.center);
		\draw (10.center) to (8.center);
		\draw (11.center) to (9.center);
		\draw (11.center) to (10.center);
		\draw (13.center) to (42.center);
		\draw (14.center) to (43.center);
		\draw (17.center) to (18.center);
		\draw (20.center) to (45.center);
		\draw (48) to (49.center);
	\end{pgfonlayer}
\end{tikzpicture}}
=\interp{\begin{tikzpicture}
	\begin{pgfonlayer}{nodelayer}
		\node [style=none] (87)  at (-0.5, -1.875) {};
		\node [style=none] (88)  at (0.5, -1.875) {};
		\node [style=ancilla] (89)  at (-0.5, 0.625) {};
		\node [style=none] (90)  at (0.5, 0.625) {};
		\node [style=none] (93)  at (-0.5, -0.625) {};
		\node [style=none] (94)  at (0.5, -0.625) {};
		\node [style=ancilla] (95)  at (-0.5, 1.875) {};
		\node [style=none] (96)  at (0.5, 1.875) {};
		\node [style=none] (97)  at (0.3, 1.425) {$\vdots$};
		\node [style=none] (99)  at (0.3, -1.075) {$\vdots$};
		\node [style=none, font={\scriptsize}] (150)  at (0.0, 1.25) {$n$};
		\node [style=none, font={\scriptsize}] (153)  at (-0.25, -1.25) {$k{-}1$};
	\end{pgfonlayer}
	\begin{pgfonlayer}{edgelayer}
		\draw (88.center) to (87.center);
		\draw (89) to (90.center);
		\draw (93.center) to (94.center);
		\draw (95) to (96.center);
	\end{pgfonlayer}
\end{tikzpicture}}\]
which implies 
\[\interp{\begin{tikzpicture}
	\begin{pgfonlayer}{nodelayer}
		\node [style=none, font={\tiny}] (61)  at (0.0, -0.125) {$\left(\begin{array}{@{}c@{}c@{\;}|@{\;}c@{}}I&&\\&A_0&\\\hline&&I\end{array}\right)$};
		\node [style=none] (101)  at (1.25, 1.625) {};
		\node [style=none] (102)  at (3.0, 2.125) {};
		\node [style=none] (103)  at (-3.0, -1.625) {};
		\node [style=none] (104)  at (-1.25, -2.125) {};
		\node [style=ancilla] (105)  at (-3.0, 0.875) {};
		\node [style=none] (106)  at (-1.25, 0.375) {};
		\node [style=none] (107)  at (-1.25, 2.375) {};
		\node [style=none] (108)  at (-1.25, -2.375) {};
		\node [style=none] (109)  at (1.25, 2.375) {};
		\node [style=none] (110)  at (1.25, -2.375) {};
		\node [style=none] (111)  at (1.25, -2.125) {};
		\node [style=none] (112)  at (1.25, 0.375) {};
		\node [style=none] (113)  at (-3.0, -0.375) {};
		\node [style=none] (114)  at (-1.25, -0.875) {};
		\node [style=none] (115)  at (1.25, -0.875) {};
		\node [style=none] (116)  at (3.0, -1.625) {};
		\node [style=none] (117)  at (3.0, 0.875) {};
		\node [style=none] (118)  at (3.0, -0.375) {};
		\node [style=ancilla] (119)  at (-3.0, 2.125) {};
		\node [style=none] (120)  at (-1.25, 1.625) {};
		\node [style=none] (122)  at (-2.45, 1.675) {$\vdots$};
		\node [style=none] (123)  at (2.55, 1.675) {$\vdots$};
		\node [style=none] (125)  at (-2.45, -0.825) {$\vdots$};
		\node [style=none] (126)  at (2.55, -0.825) {$\vdots$};
		\node [style=none] (130)  at (-1.25, 2.125) {};
		\node [style=none] (131)  at (-2.5, -2.125) {};
		\node [style=none] (132)  at (-3.0, -2.125) {};
		\node [style=none] (133)  at (1.25, 2.125) {};
		\node [style=none] (134)  at (2.5, -2.125) {};
		\node [style=none] (135)  at (3.0, -2.125) {};
		\node [style=none, font={\scriptsize}] (151)  at (-2.75, 1.5) {$n$};
		\node [style=none, font={\scriptsize}] (154)  at (-3.0, -1.0) {$k{-}1$};
	\end{pgfonlayer}
	\begin{pgfonlayer}{edgelayer}
		\draw [in=180, out=0] (101.center) to (102.center);
		\draw [in=360, out=180] (104.center) to (103.center);
		\draw [in=180, out=0] (105) to (106.center);
		\draw (108.center) to (107.center);
		\draw (109.center) to (107.center);
		\draw (110.center) to (108.center);
		\draw (110.center) to (109.center);
		\draw [in=180, out=0] (111.center) to (116.center);
		\draw [in=180, out=0] (112.center) to (117.center);
		\draw [in=180, out=0] (113.center) to (114.center);
		\draw [in=180, out=0] (115.center) to (118.center);
		\draw [in=180, out=0] (119) to (120.center);
		\draw [in=180, out=0, looseness=0.50] (131.center) to (130.center);
		\draw (132.center) to (131.center);
		\draw [in=0, out=180, looseness=0.50] (134.center) to (133.center);
		\draw (135.center) to (134.center);
	\end{pgfonlayer}
\end{tikzpicture}}
=\interp{\begin{tikzpicture}
	\begin{pgfonlayer}{nodelayer}
		\node [style=none] (136)  at (-0.5, -1.625) {};
		\node [style=none] (137)  at (0.5, -1.625) {};
		\node [style=ancilla] (138)  at (-0.5, 0.875) {};
		\node [style=none] (139)  at (0.5, 0.875) {};
		\node [style=none] (140)  at (-0.5, -0.375) {};
		\node [style=none] (141)  at (0.5, -0.375) {};
		\node [style=ancilla] (142)  at (-0.5, 2.125) {};
		\node [style=none] (143)  at (0.5, 2.125) {};
		\node [style=none] (144)  at (0.3, 1.675) {$\vdots$};
		\node [style=none] (146)  at (0.3, -0.825) {$\vdots$};
		\node [style=none] (148)  at (-0.5, -2.125) {};
		\node [style=none] (149)  at (0.5, -2.125) {};
		\node [style=none, font={\scriptsize}] (152)  at (0.0, 1.5) {$n$};
		\node [style=none, font={\scriptsize}] (155)  at (-0.25, -1.0) {$k{-}1$};
	\end{pgfonlayer}
	\begin{pgfonlayer}{edgelayer}
		\draw (137.center) to (136.center);
		\draw (138) to (139.center);
		\draw (140.center) to (141.center);
		\draw (142) to (143.center);
		\draw (149.center) to (148.center);
	\end{pgfonlayer}
\end{tikzpicture}}\]
We can hence apply the induction hypothesis on $\tikzfigS{./qciso-completeness/NF-subinduct/NF-subinduct_04}$ from which we conclude that the second block is deleted by initialisations on the $n+1$ first qubits. Similary, the penultimate block (with $B_0$) is deleted by initialisations on the first $n+1$ qubits.

All the controlled rotations in the middle are deleted by the initialisations thanks to \Cref{lem:InitCtrl}. Finally, it is provable that:
\[\begin{tikzpicture}
	\begin{pgfonlayer}{nodelayer}
		\node [style=none] (103) at (-1, 0) {$=$};
		\node [style=none] (104) at (0, -2.75) {};
		\node [style=none] (105) at (1, -2.75) {};
		\node [style=ancilla] (107) at (0, 2.875) {};
		\node [style=none] (108) at (1, 2.875) {};
		\node [style=ancilla] (109) at (0, 1.875) {};
		\node [style=none] (110) at (1, 1.875) {};
		\node [style=none] (120) at (0.55, 1.3) {$\vdots$};
		\node [style=none] (121) at (0.475, -1.825) {$\vdots$};
		\node [style=ancilla] (132) at (0, -0.375) {};
		\node [style=none] (133) at (1, -0.375) {};
		\node [style=none] (134) at (0, -1.25) {};
		\node [style=none] (135) at (1, -1.25) {};
		\node [style=ancilla] (147) at (0, 0.375) {};
		\node [style=none] (148) at (1, 0.375) {};
		\node [style=none] (149) at (-17.5, -2.75) {};
		\node [style=none] (150) at (-16.5, -2.75) {};
		\node [style=none] (151) at (-11.5, 2.75) {};
		\node [style=none] (152) at (-8, 2.75) {};
		\node [style=ancilla] (153) at (-17.5, 1.75) {};
		\node [style=none] (154) at (-16.5, 1.75) {};
		\node [style=none] (155) at (-16.5, 3) {};
		\node [style=none] (156) at (-16.5, -3) {};
		\node [style=none] (157) at (-14.5, 3) {};
		\node [style=none] (158) at (-14.5, -3) {};
		\node [style=none] (159) at (-14.5, -2.75) {};
		\node [style=none] (160) at (-14, -2.75) {};
		\node [style=none] (161) at (-14.5, 1.75) {};
		\node [style=none] (162) at (-14, 1.75) {};
		\node [style=none] (163) at (-16.95, 1.175) {$\vdots$};
		\node [style=none] (164) at (-17.025, -1.825) {$\vdots$};
		\node [style=none] (165) at (-14, 3) {};
		\node [style=none] (166) at (-14, -3) {};
		\node [style=none] (167) at (-11.5, 3) {};
		\node [style=none] (168) at (-11.5, -3) {};
		\node [style=none, font={\tiny}] (169) at (-12.75, 0) {$\left(\begin{array}{@{}c@{}c@{\;}|@{\;}c@{}}I&&\\&A_0&\\\hline&&I\end{array}\right)$};
		\node [style=none] (170) at (-11.5, -2.75) {};
		\node [style=none] (171) at (-11.5, 1.75) {};
		\node [style=ancilla] (172) at (-17.5, -0.375) {};
		\node [style=none] (173) at (-16.5, -0.375) {};
		\node [style=none] (174) at (-17.5, -1.25) {};
		\node [style=none] (175) at (-16.5, -1.25) {};
		\node [style=none] (176) at (-14.5, -0.375) {};
		\node [style=none] (177) at (-14, -0.375) {};
		\node [style=none] (178) at (-14.5, -1.25) {};
		\node [style=none] (179) at (-14, -1.25) {};
		\node [style=none] (180) at (-11.5, -0.375) {};
		\node [style=none] (181) at (-11.5, -1.25) {};
		\node [style=gate] (182) at (-11, 2.75) {$\star$};
		\node [style={bwcontrol={}}] (183) at (-11, -1.25) {};
		\node [style={bwcontrol={}}] (184) at (-11, -2.75) {};
		\node [style={bwcontrol={}}] (185) at (-11, 1.75) {};
		\node [style=control] (186) at (-11, -0.375) {};
		\node [style=ancilla] (187) at (-17.5, 0.375) {};
		\node [style=none] (188) at (-16.5, 0.375) {};
		\node [style=none] (189) at (-14.5, 0.375) {};
		\node [style=none] (190) at (-14, 0.375) {};
		\node [style=none] (191) at (-11.5, 0.375) {};
		\node [style={bwcontrol={}}] (192) at (-11, 0.375) {};
		\node [style=gate] (193) at (-10.25, 2.75) {$\star$};
		\node [style={bwcontrol={}}] (194) at (-10.25, -1.25) {};
		\node [style={bwcontrol={}}] (195) at (-10.25, -2.75) {};
		\node [style={bwcontrol={}}] (196) at (-10.25, 1.75) {};
		\node [style=control] (197) at (-10.25, 0.375) {};
		\node [style=wcontrol] (198) at (-10.25, -0.375) {};
		\node [style=gate] (199) at (-8.75, 2.75) {$\star$};
		\node [style={bwcontrol={}}] (200) at (-8.75, -1.25) {};
		\node [style={bwcontrol={}}] (201) at (-8.75, -2.75) {};
		\node [style=wcontrol] (202) at (-8.75, 0.375) {};
		\node [style=control] (203) at (-8.75, 1.75) {};
		\node [style=wcontrol] (204) at (-8.75, -0.375) {};
		\node [style=none] (205) at (-8, -2.75) {};
		\node [style=none] (206) at (-8, 1.75) {};
		\node [style=none] (207) at (-8, -0.375) {};
		\node [style=none] (208) at (-8, -1.25) {};
		\node [style=none] (209) at (-8, 0.375) {};
		\node [style=none] (210) at (-9.5, 1.125) {\reflectbox{$\ddots$}};
		\node [style=none, font={\tiny}] (211) at (-15.5, 0) {$\left(\begin{array}{@{}c@{\;}|@{\;}c@{}}I&\\\hline&A_1\end{array}\right)$};
		\node [style=ancilla] (212) at (-17.5, 2.75) {};
		\node [style=none] (213) at (-16.5, 2.75) {};
		\node [style=none] (214) at (-14.5, 2.75) {};
		\node [style=none] (215) at (-14, 2.75) {};
		\node [style=none] (216) at (-2, -2.75) {};
		\node [style=none] (217) at (-3, -2.75) {};
		\node [style=none] (218) at (-2, 1.75) {};
		\node [style=none] (219) at (-3, 1.75) {};
		\node [style=none] (220) at (-3, 3) {};
		\node [style=none] (221) at (-3, -3) {};
		\node [style=none] (222) at (-5, 3) {};
		\node [style=none] (223) at (-5, -3) {};
		\node [style=none] (224) at (-5, -2.75) {};
		\node [style=none] (225) at (-5.5, -2.75) {};
		\node [style=none] (226) at (-5, 1.75) {};
		\node [style=none] (227) at (-5.5, 1.75) {};
		\node [style=none] (228) at (-2.55, 1.175) {$\vdots$};
		\node [style=none] (229) at (-2.475, -1.825) {$\vdots$};
		\node [style=none] (230) at (-5.5, 3) {};
		\node [style=none] (231) at (-5.5, -3) {};
		\node [style=none] (232) at (-8, 3) {};
		\node [style=none] (233) at (-8, -3) {};
		\node [style=none, font={\tiny}] (234) at (-6.75, 0) {$\left(\begin{array}{@{}c@{}c@{\;}|@{\;}c@{}}I&&\\&B_0&\\\hline&&I\end{array}\right)$};
		\node [style=none] (235) at (-2, -0.375) {};
		\node [style=none] (236) at (-3, -0.375) {};
		\node [style=none] (237) at (-2, -1.25) {};
		\node [style=none] (238) at (-3, -1.25) {};
		\node [style=none] (239) at (-5, -0.375) {};
		\node [style=none] (240) at (-5.5, -0.375) {};
		\node [style=none] (241) at (-5, -1.25) {};
		\node [style=none] (242) at (-5.5, -1.25) {};
		\node [style=none] (243) at (-2, 0.375) {};
		\node [style=none] (244) at (-3, 0.375) {};
		\node [style=none] (245) at (-5, 0.375) {};
		\node [style=none] (246) at (-5.5, 0.375) {};
		\node [style=none, font={\tiny}] (247) at (-4, 0) {$\left(\begin{array}{@{}c@{\;}|@{\;}c@{}}I&\\\hline&B_1\end{array}\right)$};
		\node [style=none] (248) at (-2, 2.75) {};
		\node [style=none] (249) at (-3, 2.75) {};
		\node [style=none] (250) at (-5, 2.75) {};
		\node [style=none] (251) at (-5.5, 2.75) {};
	\end{pgfonlayer}
	\begin{pgfonlayer}{edgelayer}
		\draw (104.center) to (105.center);
		\draw (107) to (108.center);
		\draw (109) to (110.center);
		\draw (132) to (133.center);
		\draw (134.center) to (135.center);
		\draw (147) to (148.center);
		\draw (149.center) to (150.center);
		\draw (151.center) to (152.center);
		\draw (153) to (154.center);
		\draw (156.center) to (155.center);
		\draw (157.center) to (155.center);
		\draw (158.center) to (156.center);
		\draw (158.center) to (157.center);
		\draw (159.center) to (160.center);
		\draw (161.center) to (162.center);
		\draw (166.center) to (165.center);
		\draw (167.center) to (165.center);
		\draw (168.center) to (166.center);
		\draw (168.center) to (167.center);
		\draw (172) to (173.center);
		\draw (174.center) to (175.center);
		\draw (176.center) to (177.center);
		\draw (178.center) to (179.center);
		\draw (187) to (188.center);
		\draw (189.center) to (190.center);
		\draw (170.center) to (205.center);
		\draw (171.center) to (206.center);
		\draw (180.center) to (207.center);
		\draw (181.center) to (208.center);
		\draw (191.center) to (209.center);
		\draw [style=trait] (185) to (192);
		\draw [style=trait] (196) to (197);
		\draw [style=trait] (203) to (202);
		\draw [style=trait] (183) to (184);
		\draw [style=trait] (194) to (195);
		\draw [style=trait] (200) to (201);
		\draw (182) to (185);
		\draw (192) to (183);
		\draw (193) to (196);
		\draw (197) to (194);
		\draw (199) to (203);
		\draw (202) to (200);
		\draw (212) to (213.center);
		\draw (214.center) to (215.center);
		\draw (216.center) to (217.center);
		\draw (218.center) to (219.center);
		\draw (221.center) to (220.center);
		\draw (222.center) to (220.center);
		\draw (223.center) to (221.center);
		\draw (223.center) to (222.center);
		\draw (224.center) to (225.center);
		\draw (226.center) to (227.center);
		\draw (231.center) to (230.center);
		\draw (232.center) to (230.center);
		\draw (233.center) to (231.center);
		\draw (233.center) to (232.center);
		\draw (235.center) to (236.center);
		\draw (237.center) to (238.center);
		\draw (239.center) to (240.center);
		\draw (241.center) to (242.center);
		\draw (243.center) to (244.center);
		\draw (245.center) to (246.center);
		\draw (248.center) to (249.center);
		\draw (250.center) to (251.center);
	\end{pgfonlayer}
\end{tikzpicture}\]
% \end{itemize}
\end{proof}

%%%%%%%%%%%%%%%%%%%%%%%%%%%%%%%%%%%%%%%%%%%%%%%%%%%%%%%%%%%%%%%%%%%%%%%%%%%%%%%
%%%%%%%%%%%%%%%%%%%%%%%%%%%%%%%%%%%%%%%%%%%%%%%%%%%%%%%%%%%%%%%%%%%%%%%%%%%%%%%
\section{Completeness of $\QCancilla$}
\subsection{Proof of Proposition \ref{prop:mctrlaltdef}}
\label{appendix:proofmctrlaltdef}
First, we derive the following equations. Equations \eqref{ctrctrlPdefTOF} and \eqref{ctrctrlctrlPdefTOF} are alternative definitions of 2-controlled and 3-controlled phase gates. Equation \eqref{ctrlPinit} tells us how we can express a simply controlled gate with Toffoli gates and one 1-qubit gate using one ancilla.
\begin{equation}\label{ctrctrlPdefTOF}\begin{array}{rrccl}\tikzfigS{./qcancilla-completeness/alt-defs/ctrlctrlPphi}&=&\tikzfigS{./qcancilla-completeness/alt-defs/ctrlctrlPphiTOF}\end{array}\end{equation}
\begin{equation}\label{ctrctrlctrlPdefTOF}\begin{array}{rrccl}\tikzfigS{./qcancilla-completeness/alt-defs/ctrlctrlctrlPphi}&=&\tikzfigS{./qcancilla-completeness/alt-defs/ctrlctrlctrlPphiTOF}\end{array}\end{equation}
\begin{equation}\label{ctrlPinit}\begin{array}{rrccl}\tikzfigS{./qcancilla-completeness/ctrlPphiinit}&=&\tikzfigS{./qcancilla-completeness/TOFPphiTOF}\end{array}\end{equation}

\begin{proof}[Proof of \cref{ctrctrlPdefTOF}]
  \begin{gather*}
    \tikzfigS{./qcancilla-completeness/alt-defs/ctrlctrlPphi}
    \eqeqref{mctrlPinducdef}\tikzfigS{./qcancilla-completeness/alt-defs/ctrctrlPdefTOF-step-1}\eqeqrefwcontrol\tikzfigS{./qcancilla-completeness/alt-defs/ctrctrlPdefTOF-step-2}\\[0.4cm]
    \eqeqrefwcontrol\tikzfigS{./qcancilla-completeness/alt-defs/ctrctrlPdefTOF-step-3}
    \eqquatreeqref{XX}{TOFTOF}{mctrlzeroid}{mctrlPaddition}\tikzfigS{./qcancilla-completeness/alt-defs/ctrctrlPdefTOF-step-4}\\[0.4cm]
    \eqeqrefwcontrol\tikzfigS{./qcancilla-completeness/alt-defs/ctrctrlPdefTOF-step-5}\eqeqrefwcontrol\tikzfigS{./qcancilla-completeness/alt-defs/ctrlctrlPphiTOF}
  \end{gather*}
\end{proof}

\begin{proof}[Proof of \cref{ctrctrlctrlPdefTOF}]
  \begin{gather*}
    \tikzfigS{./qcancilla-completeness/alt-defs/ctrlctrlctrlPphi}
    \eqeqref{mctrlPinducdef}\tikzfigS{./qcancilla-completeness/alt-defs/ctrctrlctrlPdefTOF-step-1}\eqeqrefwcontrol\tikzfigS{./qcancilla-completeness/alt-defs/ctrctrlctrlPdefTOF-step-2}\\[0.4cm]
    \eqeqrefwcontrol\tikzfigS{./qcancilla-completeness/alt-defs/ctrctrlctrlPdefTOF-step-3}\eqquatreeqref{XX}{TOFTOF}{mctrlzeroid}{mctrlPaddition}\tikzfigS{./qcancilla-completeness/alt-defs/ctrctrlctrlPdefTOF-step-4}\\[0.4cm]
    \eqeqrefwcontrol\tikzfigS{./qcancilla-completeness/alt-defs/ctrctrlctrlPdefTOF-step-5}
    \eqeqrefwcontrol\tikzfigS{./qcancilla-completeness/alt-defs/ctrlctrlctrlPphiTOF}
  \end{gather*}
\end{proof}

\begin{proof}[Proof of \cref{ctrlPinit}]
  \begin{gather*}
    \tikzfigS{./qcancilla-completeness/TOFPphiTOF}
    \eqeqref{ancillamctrlPneg}\tikzfigS{./qcancilla-completeness/ctrlPinit-step-1}
    \eqeqref{mctrlPlift}\tikzfigS{./qcancilla-completeness/ctrlPinit-step-2}\\[0.4cm]
    \eqeqref{ctrctrlPdefTOF}\tikzfigS{./qcancilla-completeness/ctrlPinit-step-3}
    \eqtroiseqref{TOFTOF}{Paddition}{axQCP0}\tikzfigS{./qcancilla-completeness/ctrlPinit-step-4}
    \eqeqref{axQCISOinitP}\tikzfigS{./qcancilla-completeness/ctrlPphiinit}
  \end{gather*}
\end{proof}

We first prove Equations \eqref{Paltdef} and \eqref{RXaltdef} by induction on the number of qubits, whose base cases contains $n=4$ qubits. The base case for Equations \eqref{Paltdef} can be derived as follows.
\begin{gather*}
  \tikzfigS{./qcancilla-completeness/alt-defs/P-bc-step-0}
  \eqeqref{ancillamctrlPneg}\tikzfigS{./qcancilla-completeness/alt-defs/P-bc-step-1}
  \eqeqref{mctrlPlift}\tikzfigS{./qcancilla-completeness/alt-defs/P-bc-step-2}
  \eqeqref{ctrctrlctrlPdefTOF}\tikzfigS{./qcancilla-completeness/alt-defs/P-bc-step-3}\\[0.4cm]
  \eqtroiseqref{TOFTOF}{Paddition}{mctrlzeroid}\tikzfigS{./qcancilla-completeness/alt-defs/P-bc-step-4}
  \eqeqref{mctrlPlift}\tikzfigS{./qcancilla-completeness/alt-defs/P-bc-step-5}
  \eqeqref{ancillamctrlPneg}\tikzfigS{./qcancilla-completeness/alt-defs/P-bc-step-6}
\end{gather*}

The base case for Equations \eqref{Paltdef} can be derived as follows.
\begin{gather*}
  \tikzfigS{./qcancilla-completeness/alt-defs/RX-bc-step-0}
  \eqeqref{mctrlPdef}\tikzfigS{./qcancilla-completeness/alt-defs/RX-bc-step-1}
  \eqeqref{Paltdef}\tikzfigS{./qcancilla-completeness/alt-defs/RX-bc-step-2}\\[0.4cm]
  \eqeqref{ctrlPinit}\tikzfigS{./qcancilla-completeness/alt-defs/RX-bc-step-3}
  \eqeqref{TOFTOF}\tikzfigS{./qcancilla-completeness/alt-defs/RX-bc-step-4}
  \eqeqref{mctrlPdef}\tikzfigS{./qcancilla-completeness/alt-defs/RX-bc-step-5}
\end{gather*}

The induction step for Equation \eqref{RXaltdef} can be derived as follows.
\begin{gather*}
  \tikzfigS{./qcancilla-completeness/alt-defs/RX-step-0}
  \eqeqref{mctrlRXSWAP}\tikzfigS{./qcancilla-completeness/alt-defs/RX-step-1}
  \eqeqref{mctrlRXdef}\tikzfigS{./qcancilla-completeness/alt-defs/RX-step-2}\\[0.4cm]
  \overset{\text{IH}}{=}\tikzfigS{./qcancilla-completeness/alt-defs/RX-step-3}
  \overset{\text{IH}}{=}\tikzfigS{./qcancilla-completeness/alt-defs/RX-step-4}\\[0.4cm]
  \eqeqref{TOFTOF}\tikzfigS{./qcancilla-completeness/alt-defs/RX-step-5}=\tikzfigS{./qcancilla-completeness/alt-defs/RX-step-6}
  \eqeqref{mctrlRXdef}\tikzfigS{./qcancilla-completeness/alt-defs/RX-step-7}
\end{gather*}

The induction step for Equation \eqref{Paltdef} can be derived as follows.
\begin{gather*}
  \tikzfigS{./qcancilla-completeness/alt-defs/P-step-0}
  \eqeqref{mctrlPdef}\tikzfigS{./qcancilla-completeness/alt-defs/P-step-1}
  \eqeqref{RXaltdef}\tikzfigS{./qcancilla-completeness/alt-defs/P-step-2}
  \overset{\text{IH}}{=}\tikzfigS{./qcancilla-completeness/alt-defs/P-step-3}\\[0.4cm]
  \eqeqref{TOFTOF}\tikzfigS{./qcancilla-completeness/alt-defs/P-step-4}
  \eqeqref{mctrlPdef}\tikzfigS{./qcancilla-completeness/alt-defs/P-step-5}
\end{gather*}

\subsection{Proof of Equation \eqref{M3} in $\QCancilla$}\label{sec:proof-M3}
The main idea of the proof of \cref{M3} is to use the \emph{Fredkin gate} (or \emph{controlled-swap gate}), defined by \cref{fredkin}.
\begin{equation}\label{fredkin}
  \begin{array}{rcl}
    \tikzfigS{./shortcut/Fredkin}&\defeq&\tikzfigS{./shortcut/Fredkindef}
  \end{array}
\end{equation}
\begin{equation}\label{M3}\tag{$\text{K}^3$}
  \begin{array}{rcl}
    \tikzfigS{./qcancilla-axioms/M3-left}&=&\tikzfigS{./qcancilla-axioms/M3-right-simp}
  \end{array}
\end{equation}
First, we derive some useful equations.

\begin{multicols}{2}

\begin{equation}\label{HHTOFHH}
  \begin{array}{rcl}
    \tikzfigS{./qcancilla-completeness/HHTOFHH-step-0}&=&\tikzfigS{./qcancilla-completeness/HHTOFHH-step-4}
  \end{array}
\end{equation}

\begin{equation}\label{HHFredkinFHH}
  \begin{array}{rcl}
    \tikzfigS{./qcancilla-completeness/HHFredkinFHH-step-0}&=&\tikzfigS{./qcancilla-completeness/HHFredkinFHH-step-3}
  \end{array}
\end{equation}

\begin{equation}\label{initTOF}
  \begin{array}{rcl}
    \tikzfigS{./qcancilla-completeness/initTOF-step-0}&=&\tikzfigS{./qcancilla-completeness/initTOF-step-5}
  \end{array}
\end{equation}

\begin{equation}\label{K1}
  \begin{array}{rcl}
    \tikzfigS{./qcancilla-completeness/K1-left}&=&\tikzfigS{./qcancilla-completeness/K1-right}
  \end{array}
\end{equation}

\begin{equation}\label{3tofs2cnots}
  \begin{array}{rcl}
    \tikzfigS{./qcancilla-completeness/3tofs2cnots-step-0}&=&\tikzfigS{./qcancilla-completeness/3tofs2cnots-step-5}
  \end{array}
\end{equation}

\begin{equation}\label{wbTOF}
  \begin{array}{rcl}
    \tikzfigS{./qcancilla-completeness/AntiCnot-is-CCNot-w-CNot/AntiCnot-is-CCNot-w-CNot_00}&=&\tikzfigS{./qcancilla-completeness/AntiCnot-is-CCNot-w-CNot/AntiCnot-is-CCNot-w-CNot_07}
  \end{array}
\end{equation}

\begin{equation}\label{5tofs}
  \begin{array}{rcl}
    \tikzfigS{./qcancilla-completeness/5tofs-step-0}&=&\tikzfigS{./qcancilla-completeness/5tofs-step-5}
  \end{array}
\end{equation}

\begin{equation}\label{TOFFredkin}
  \begin{array}{rcl}
    \tikzfigS{./qcancilla-completeness/TOFFredkin-step-0}&=&\tikzfigS{./qcancilla-completeness/TOFFredkin-step-6}
  \end{array}
\end{equation}

\begin{equation}\label{wFredkin}
  \begin{array}{rcl}
    \tikzfigS{./qcancilla-completeness/AntiCSwap/AntiCSwap_00}&=&\tikzfigS{./qcancilla-completeness/AntiCSwap/AntiCSwap_04}
  \end{array}
\end{equation}

\begin{equation}\label{wCZ-Z}
  \begin{array}{rcl}
    \tikzfigS{./qcancilla-completeness/AntiCZ-Z/AntiCZ-Z_00}&=&\tikzfigS{./qcancilla-completeness/AntiCZ-Z/AntiCZ-Z_04}
  \end{array}
\end{equation}

\begin{equation}\label{ctrlPphasegadget}
  \begin{array}{rcl}
    \tikzfigS{./qcancilla-completeness/phase-gadget-3-is-symetric/phase-gadget-3-is-symetric_00}&=&\tikzfigS{./qcancilla-completeness/phase-gadget-3-is-symetric/phase-gadget-3-is-symetric_07}
  \end{array}
\end{equation}

\begin{equation}\label{wCCZ-CZ}
  \begin{array}{rcl}
    \tikzfigS{./qcancilla-completeness/AntiCCZ-CZ/AntiCCZ-CZ_00}&=&\tikzfigS{./qcancilla-completeness/AntiCCZ-CZ/AntiCCZ-CZ_05}
  \end{array}
\end{equation}

\begin{equation}\label{wCCRX-CRX}
  \begin{array}{rcl}
    \tikzfigS{./qcancilla-completeness/AntiCCRx-CRx/AntiCCRx-CRx_00}&=&\tikzfigS{./qcancilla-completeness/AntiCCRx-CRx/AntiCCRx-CRx_04}
  \end{array}
\end{equation}

\begin{equation}\label{FredkinwbTOF}
  \begin{array}{rcl}
    \tikzfigS{./qcancilla-completeness/AntiCCNot-through-CSwap/AntiCCNot-through-CSwap_00}&=&\tikzfigS{./qcancilla-completeness/AntiCCNot-through-CSwap/AntiCCNot-through-CSwap_04}
  \end{array}
\end{equation}

\end{multicols}

\begin{proof}[Proof of \cref{HHTOFHH}]
  \begin{gather*}
    \tikzfigS{./qcancilla-completeness/HHTOFHH-step-0}
    \eqeqref{TOFdef}\tikzfigS{./qcancilla-completeness/HHTOFHH-step-1}
    \eqeqref{axQCHH}\tikzfigS{./qcancilla-completeness/HHTOFHH-step-2}
    \eqeqref{mctrlPlift}\tikzfigS{./qcancilla-completeness/HHTOFHH-step-3}
    \eqeqref{TOFdef}\tikzfigS{./qcancilla-completeness/HHTOFHH-step-4}
  \end{gather*}
\end{proof}

\begin{proof}[Proof of \cref{HHFredkinFHH}]
  \begin{gather*}
    \tikzfigS{./qcancilla-completeness/HHFredkinFHH-step-0}
    \eqeqref{fredkin}\tikzfigS{./qcancilla-completeness/HHFredkinFHH-step-1}
    \eqdeuxeqref{CNOTHH}{HHTOFHH}\tikzfigS{./qcancilla-completeness/HHFredkinFHH-step-2}=\tikzfigS{./qcancilla-completeness/HHFredkinFHH-step-3}
  \end{gather*}
\end{proof}

\begin{proof}[Proof of \cref{initTOF}]
  \begin{gather*}
    \tikzfigS{./qcancilla-completeness/initTOF-step-0}
    \eqeqref{TOFdef}\tikzfigS{./qcancilla-completeness/initTOF-step-1}
    \eqeqref{Paltdef}\tikzfigS{./qcancilla-completeness/initTOF-step-2}
    \eqeqref{mctrlPinducdef}\tikzfigS{./qcancilla-completeness/initTOF-step-3}\\[0.4cm]
    \eqeqref{axQCCZ}\tikzfigS{./qcancilla-completeness/initTOF-step-4}
    \eqeqref{axQCHH}\tikzfigS{./qcancilla-completeness/initTOF-step-5}
  \end{gather*}
\end{proof}

\begin{proof}[Proof of \cref{K1}]
  \begin{gather*}
    \tikzfigS{./qcancilla-completeness/K1-left}
    \eqeqref{axQCANCinitdest}\tikzfigS{./qcancilla-completeness/K1-step-1}
    \eqeqref{initTOF}\tikzfigS{./qcancilla-completeness/K1-step-2}
    =\tikzfigS{./qcancilla-completeness/K1-step-3}
    \eqeqref{axQC3cnots}\tikzfigS{./qcancilla-completeness/K1-step-4}\\[0.4cm]
    \eqeqref{TOFTOF}\tikzfigS{./qcancilla-completeness/K1-step-5}
    \eqeqref{initTOF}\tikzfigS{./qcancilla-completeness/K1-step-6}
    \eqeqref{axQCANCinitdest}\tikzfigS{./qcancilla-completeness/K1-right}
  \end{gather*}
\end{proof}

\begin{proof}[Proof of \cref{3tofs2cnots}]
  \begin{gather*}
    \tikzfigS{./qcancilla-completeness/3tofs2cnots-step-0}
    \eqeqref{CNOTCNOT}\tikzfigS{./qcancilla-completeness/3tofs2cnots-step-1}
    \eqdeuxeqref{CNOTHH}{HHTOFHH}\tikzfigS{./qcancilla-completeness/3tofs2cnots-step-2}
    \eqeqref{K1}\tikzfigS{./qcancilla-completeness/3tofs2cnots-step-3}\\[0.4cm]
    \eqeqref{TOFTOF}\tikzfigS{./qcancilla-completeness/3tofs2cnots-step-4}
    \eqdeuxeqref{CNOTHH}{HHTOFHH}\tikzfigS{./qcancilla-completeness/3tofs2cnots-step-5}
  \end{gather*}
\end{proof}

\begin{proof}[Proof of \cref{wbTOF}]
  \begin{gather*}
    \tikzfigS{./qcancilla-completeness/AntiCnot-is-CCNot-w-CNot/AntiCnot-is-CCNot-w-CNot_00}
    = \tikzfigS{./qcancilla-completeness/AntiCnot-is-CCNot-w-CNot/AntiCnot-is-CCNot-w-CNot_01}
    \eqdeuxeqref{axQCANCinitdest}{axQCISOinitcnot} \tikzfigS{./qcancilla-completeness/AntiCnot-is-CCNot-w-CNot/AntiCnot-is-CCNot-w-CNot_02}
    \eqeqref{CNOTXX} \tikzfigS{./qcancilla-completeness/AntiCnot-is-CCNot-w-CNot/AntiCnot-is-CCNot-w-CNot_03}
    \eqeqref{3tofs2cnots} \tikzfigS{./qcancilla-completeness/AntiCnot-is-CCNot-w-CNot/AntiCnot-is-CCNot-w-CNot_04}\\[0.4cm]
    \eqeqref{ancillaTOFpos} \tikzfigS{./qcancilla-completeness/AntiCnot-is-CCNot-w-CNot/AntiCnot-is-CCNot-w-CNot_05}
    \eqeqref{CNOTXX}\tikzfigS{./qcancilla-completeness/AntiCnot-is-CCNot-w-CNot/AntiCnot-is-CCNot-w-CNot_06}
    \eqdeuxeqref{axQCANCinitdest}{axQCISOinitcnot} \tikzfigS{./qcancilla-completeness/AntiCnot-is-CCNot-w-CNot/AntiCnot-is-CCNot-w-CNot_07}
  \end{gather*}
\end{proof}

\begin{proof}[Proof of \cref{5tofs}]
  \begin{gather*}
    \tikzfigS{./qcancilla-completeness/5tofs-step-0}
    \eqeqref{wbTOF}\tikzfigS{./qcancilla-completeness/5tofs-step-1}
    \eqeqrefwcontrol\tikzfigS{./qcancilla-completeness/5tofs-step-2}
    \eqeqref{K1}\tikzfigS{./qcancilla-completeness/5tofs-step-3}
    \eqeqrefwcontrol\tikzfigS{./qcancilla-completeness/5tofs-step-4}
    \eqeqref{wbTOF}\tikzfigS{./qcancilla-completeness/5tofs-step-5}
  \end{gather*}
\end{proof}

\begin{proof}[Proof of \cref{TOFFredkin}]
  \begin{gather*}
    \tikzfigS{./qcancilla-completeness/TOFFredkin-step-0}
    \eqeqref{fredkin}\tikzfigS{./qcancilla-completeness/TOFFredkin-step-1}
    \eqeqref{3tofs2cnots}\tikzfigS{./qcancilla-completeness/TOFFredkin-step-2}
    =\tikzfigS{./qcancilla-completeness/TOFFredkin-step-3}\\[0.4cm]
    \eqeqref{5tofs}\tikzfigS{./qcancilla-completeness/TOFFredkin-step-4}
    \eqeqref{TOFTOF}\tikzfigS{./qcancilla-completeness/TOFFredkin-step-5}
    \eqeqref{fredkin}\tikzfigS{./qcancilla-completeness/TOFFredkin-step-6}
  \end{gather*}
\end{proof}

\begin{proof}[Proof of \cref{wFredkin}]
  \begin{gather*}
    \tikzfigS{./qcancilla-completeness/AntiCSwap/AntiCSwap_00}
    =\tikzfigS{./qcancilla-completeness/AntiCSwap/AntiCSwap_01}
    \eqeqref{wbTOF}\tikzfigS{./qcancilla-completeness/AntiCSwap/AntiCSwap_02}
    \eqdeuxeqref{CNOTCNOT}{axQCswap}\tikzfigS{./qcancilla-completeness/AntiCSwap/AntiCSwap_03}
    =\tikzfigS{./qcancilla-completeness/AntiCSwap/AntiCSwap_04}
  \end{gather*}
\end{proof}

\begin{proof}[Proof of \cref{wCZ-Z}]
  \begin{gather*}
    \tikzfigS{./qcancilla-completeness/AntiCZ-Z/AntiCZ-Z_00}
    \eqeqref{mctrlPinducdef}\tikzfigS{./qcancilla-completeness/AntiCZ-Z/AntiCZ-Z_01}
    \eqdeuxeqref{Pphasegadget}{XPX}\tikzfigS{./qcancilla-completeness/AntiCZ-Z/AntiCZ-Z_02}\\
    \eqtroiseqref{XcommutCNOT}{XPX}{axQCgphaseempty}\tikzfigS{./qcancilla-completeness/AntiCZ-Z/AntiCZ-Z_03}
    \eqdeuxeqref{PcommutCNOT}{Paddition}\tikzfigS{./qcancilla-completeness/AntiCZ-Z/AntiCZ-Z_04}\eqdeuxeqref{Pphasegadget}{mctrlPinducdef}\tikzfigS{./qcancilla-completeness/AntiCZ-Z/AntiCZ-Z_05}
  \end{gather*}
\end{proof}

\begin{proof}[Proof of \cref{ctrlPphasegadget}]
  \begin{gather*}
    \tikzfigS{./qcancilla-completeness/phase-gadget-3-is-symetric/phase-gadget-3-is-symetric_00}
    \eqeqref{mctrlPinducdef}\tikzfigS{./qcancilla-completeness/phase-gadget-3-is-symetric/phase-gadget-3-is-symetric_01}
    \eqeqref{Pphasegadget}\tikzfigS{./qcancilla-completeness/phase-gadget-3-is-symetric/phase-gadget-3-is-symetric_02}\\[0.4cm]
    \eqeqref{axQC3cnots}\tikzfigS{./qcancilla-completeness/phase-gadget-3-is-symetric/phase-gadget-3-is-symetric_03}
    \eqeqref{axQC3cnots}\tikzfigS{./qcancilla-completeness/phase-gadget-3-is-symetric/phase-gadget-3-is-symetric_04}\\[0.4cm]
    \eqeqref{Pphasegadget}\tikzfigS{./qcancilla-completeness/phase-gadget-3-is-symetric/phase-gadget-3-is-symetric_05}
    \eqeqref{axQC3cnots}\tikzfigS{./qcancilla-completeness/phase-gadget-3-is-symetric/phase-gadget-3-is-symetric_06}
    \eqeqref{mctrlPinducdef}\tikzfigS{./qcancilla-completeness/phase-gadget-3-is-symetric/phase-gadget-3-is-symetric_07}
  \end{gather*}
\end{proof}

\begin{proof}[Proof of \cref{wCCZ-CZ}]
  \begin{gather*}
    \tikzfigS{./qcancilla-completeness/AntiCCZ-CZ/AntiCCZ-CZ_00}
    \eqeqref{mctrlPinducdef}\tikzfigS{./qcancilla-completeness/AntiCCZ-CZ/AntiCCZ-CZ_01}
    \eqtroiseqref{mctrlPaddition}{ctrlPphasegadget}{CNOTXX}\tikzfigS{./qcancilla-completeness/AntiCCZ-CZ/AntiCCZ-CZ_02}\\
    \eqeqref{wCZ-Z}\tikzfigS{./qcancilla-completeness/AntiCCZ-CZ/AntiCCZ-CZ_03}
    \eqtroiseqref{Paddition}{axQCP0}{mctrlPlift}\tikzfigS{./qcancilla-completeness/AntiCCZ-CZ/AntiCCZ-CZ_04}
    \eqeqref{mctrlPinducdef}\tikzfigS{./qcancilla-completeness/AntiCCZ-CZ/AntiCCZ-CZ_05}
  \end{gather*}
\end{proof}

\begin{proof}[Proof of \cref{wCCRX-CRX}]
  \begin{gather*}
    \tikzfigS{./qcancilla-completeness/AntiCCRx-CRx/AntiCCRx-CRx_00}
    \eqeqref{mctrlPdef}\tikzfigS{./qcancilla-completeness/AntiCCRx-CRx/AntiCCRx-CRx_01}\eqdeuxeqref{axQCHH}{XX}\tikzfigS{./qcancilla-completeness/AntiCCRx-CRx/AntiCCRx-CRx_02}\\
    \eqdeuxeqref{wCZ-Z}{wCCZ-CZ}\tikzfigS{./qcancilla-completeness/AntiCCRx-CRx/AntiCCRx-CRx_03}\eqeqref{mctrlPdef}\tikzfigS{./qcancilla-completeness/AntiCCRx-CRx/AntiCCRx-CRx_04}
  \end{gather*}
\end{proof}

\begin{proof}[Proof of \cref{FredkinwbTOF}]
  \begin{gather*}
    \tikzfigS{./qcancilla-completeness/AntiCCNot-through-CSwap/AntiCCNot-through-CSwap_00}
    \eqeqref{XX}\tikzfigS{./qcancilla-completeness/AntiCCNot-through-CSwap/AntiCCNot-through-CSwap_01}
    \eqeqref{wFredkin}\tikzfigS{./qcancilla-completeness/AntiCCNot-through-CSwap/AntiCCNot-through-CSwap_02}
    \eqeqref{TOFFredkin}\tikzfigS{./qcancilla-completeness/AntiCCNot-through-CSwap/AntiCCNot-through-CSwap_03}
    \eqdeuxeqref{XX}{wFredkin}\tikzfigS{./qcancilla-completeness/AntiCCNot-through-CSwap/AntiCCNot-through-CSwap_04}
  \end{gather*}
\end{proof}

The idea of the proof of \cref{M3} is to start from the LHS circuit of \eqref{M3}, use Equations \eqref{PthroughFredkin},\eqref{ctrlPthroughFredkin} and \eqref{ctrlRXthroughFredkin} to build an instance of the LHS circuit of \eqref{M2} on two ancillae, apply \eqref{M2} and then rebuild the RHS circuit of \eqref{M3} using the same equations.

\begin{equation}\label{PthroughFredkin}
  \begin{array}{rcl}
    \tikzfigS{./qcancilla-completeness/PthroughFredkin-step-0}&=&\tikzfigS{./qcancilla-completeness/PthroughFredkin-step-8}
  \end{array}
\end{equation}

\begin{equation}\label{ctrlPthroughFredkin}
  \begin{array}{rcl}
    \tikzfigS{./qcancilla-completeness/ctrlPthroughFredkin-step-0}&=&\tikzfigS{./qcancilla-completeness/ctrlPthroughFredkin-step-18}
  \end{array}
\end{equation}

\begin{equation}\label{ctrlRXthroughFredkin}
  \begin{array}{rcl}
    \tikzfigS{./qcancilla-completeness/ctrlRXthroughFredkin-step-0}&=&\tikzfigS{./qcancilla-completeness/ctrlRXthroughFredkin-step-33}
  \end{array}
\end{equation}

\begin{proof}[Proof of \cref{PthroughFredkin}]
  \begin{gather*}
    \tikzfigS{./qcancilla-completeness/PthroughFredkin-step-0}
    \eqeqref{ctrlPinit}\tikzfigS{./qcancilla-completeness/PthroughFredkin-step-1}
    \eqeqref{TOFFredkin}\tikzfigS{./qcancilla-completeness/PthroughFredkin-step-2}
    \eqeqref{ctrlPinit}\tikzfigS{./qcancilla-completeness/PthroughFredkin-step-3}
    \eqeqref{wCZ-Z}\tikzfigS{./qcancilla-completeness/PthroughFredkin-step-4}\\[0.4cm]
    \eqeqref{ctrlPinit}\tikzfigS{./qcancilla-completeness/PthroughFredkin-step-5}
    \eqeqref{FredkinwbTOF}\tikzfigS{./qcancilla-completeness/PthroughFredkin-step-6}
    \eqeqref{ctrlPinit}\tikzfigS{./qcancilla-completeness/PthroughFredkin-step-7}
    \eqtroiseqref{mctrlPlift}{ancillamctrlPneg}{XX}\tikzfigS{./qcancilla-completeness/PthroughFredkin-step-8}
  \end{gather*}
\end{proof}

\begin{proof}[Proof of \cref{ctrlPthroughFredkin}]
  \begin{gather*}
    \tikzfigS{./qcancilla-completeness/ctrlPthroughFredkin-step-0}
    \eqeqref{Paltdef}\tikzfigS{./qcancilla-completeness/ctrlPthroughFredkin-step-1}
    \eqeqref{TOFFredkin}\tikzfigS{./qcancilla-completeness/ctrlPthroughFredkin-step-2}
    \eqeqref{Paltdef}\tikzfigS{./qcancilla-completeness/ctrlPthroughFredkin-step-3}
    \eqeqref{mctrlPlift}\tikzfigS{./qcancilla-completeness/ctrlPthroughFredkin-step-4}\\[0.4cm]
    \eqeqref{Paltdef}\tikzfigS{./qcancilla-completeness/ctrlPthroughFredkin-step-5}
    \eqeqref{TOFFredkin}\tikzfigS{./qcancilla-completeness/ctrlPthroughFredkin-step-6}
    \eqeqref{Paltdef}\tikzfigS{./qcancilla-completeness/ctrlPthroughFredkin-step-7}
    \eqeqref{mctrlPlift}\tikzfigS{./qcancilla-completeness/ctrlPthroughFredkin-step-8}\\[0.4cm]
    \eqeqref{wCCZ-CZ}\tikzfigS{./qcancilla-completeness/ctrlPthroughFredkin-step-9}
    \eqeqref{mctrlPlift}\tikzfigS{./qcancilla-completeness/ctrlPthroughFredkin-step-10}
    \eqeqref{Paltdef}\tikzfigS{./qcancilla-completeness/ctrlPthroughFredkin-step-11}\\[0.4cm]
    \eqeqref{FredkinwbTOF}\tikzfigS{./qcancilla-completeness/ctrlPthroughFredkin-step-12}
    \eqeqref{Paltdef}\tikzfigS{./qcancilla-completeness/ctrlPthroughFredkin-step-13}
    \eqeqref{mctrlPlift}\tikzfigS{./qcancilla-completeness/ctrlPthroughFredkin-step-14}\\[0.4cm]
    \eqeqref{Paltdef}\tikzfigS{./qcancilla-completeness/ctrlPthroughFredkin-step-15}
    \eqeqref{FredkinwbTOF}\tikzfigS{./qcancilla-completeness/ctrlPthroughFredkin-step-16}
    \eqeqref{Paltdef}\tikzfigS{./qcancilla-completeness/ctrlPthroughFredkin-step-17}
    \eqeqref{ancillamctrlPneg}\tikzfigS{./qcancilla-completeness/ctrlPthroughFredkin-step-18}
  \end{gather*}
\end{proof}

\begin{proof}[Proof of \cref{ctrlRXthroughFredkin}]
  \begin{gather*}
    \tikzfigS{./qcancilla-completeness/ctrlRXthroughFredkin-step-0}
    \eqeqref{mctrlPdef}\tikzfigS{./qcancilla-completeness/ctrlRXthroughFredkin-step-1}
    \eqeqref{HHFredkinFHH}\tikzfigS{./qcancilla-completeness/ctrlRXthroughFredkin-step-2}\\[0.4cm]
    \eqeqref{Paltdef}\tikzfigS{./qcancilla-completeness/ctrlRXthroughFredkin-step-3}
    \eqeqref{TOFFredkin}\tikzfigS{./qcancilla-completeness/ctrlRXthroughFredkin-step-4}\\[0.4cm]
    \eqeqref{Paltdef}\tikzfigS{./qcancilla-completeness/ctrlRXthroughFredkin-step-5}
    \eqeqref{mctrlPlift}\tikzfigS{./qcancilla-completeness/ctrlRXthroughFredkin-step-6}\\[0.4cm]
    \eqeqref{Paltdef}\tikzfigS{./qcancilla-completeness/ctrlRXthroughFredkin-step-7}
    \eqeqref{TOFFredkin}\tikzfigS{./qcancilla-completeness/ctrlRXthroughFredkin-step-8}\\[0.4cm]
    \eqeqref{Paltdef}\tikzfigS{./qcancilla-completeness/ctrlRXthroughFredkin-step-9}
    \eqeqref{HHFredkinFHH}\tikzfigS{./qcancilla-completeness/ctrlRXthroughFredkin-step-10}
    \eqeqref{axQCHH}\tikzfigS{./qcancilla-completeness/ctrlRXthroughFredkin-step-11}\\[0.4cm]
    \eqdeuxeqref{ctrlPinit}{mctrlPlift}\tikzfigS{./qcancilla-completeness/ctrlRXthroughFredkin-step-12}
    \eqeqref{TOFFredkin}\tikzfigS{./qcancilla-completeness/ctrlRXthroughFredkin-step-13}
    \eqeqref{ctrlPinit}\tikzfigS{./qcancilla-completeness/ctrlRXthroughFredkin-step-14}\\[0.4cm]
    \eqeqref{mctrlPdef}\tikzfigS{./qcancilla-completeness/ctrlRXthroughFredkin-step-15}
    \eqeqref{wCCRX-CRX}\tikzfigS{./qcancilla-completeness/ctrlRXthroughFredkin-step-16}
    \eqeqref{mctrlPdef}\tikzfigS{./qcancilla-completeness/ctrlRXthroughFredkin-step-17}\\[0.4cm]
    \eqeqrefwcontrol\tikzfigS{./qcancilla-completeness/ctrlRXthroughFredkin-step-18}
    \eqeqref{ctrlPinit}\tikzfigS{./qcancilla-completeness/ctrlRXthroughFredkin-step-19}\\[0.4cm]
    \eqeqref{FredkinwbTOF}\tikzfigS{./qcancilla-completeness/ctrlRXthroughFredkin-step-20}
    \eqeqref{ctrlPinit}\tikzfigS{./qcancilla-completeness/ctrlRXthroughFredkin-step-21}\\[0.4cm]
    \eqtroiseqref{mctrlPlift}{ancillamctrlPneg}{XX}\tikzfigS{./qcancilla-completeness/ctrlRXthroughFredkin-step-22}
    \eqeqref{HHFredkinFHH}\tikzfigS{./qcancilla-completeness/ctrlRXthroughFredkin-step-23}\\[0.4cm]
    \eqeqref{mctrlPlift}\tikzfigS{./qcancilla-completeness/ctrlRXthroughFredkin-step-24}
    \eqeqref{Paltdef}\tikzfigS{./qcancilla-completeness/ctrlRXthroughFredkin-step-25}\\[0.4cm]
    \eqeqref{FredkinwbTOF}\tikzfigS{./qcancilla-completeness/ctrlRXthroughFredkin-step-26}
    \eqeqref{Paltdef}\tikzfigS{./qcancilla-completeness/ctrlRXthroughFredkin-step-27}\\[0.4cm]
    \eqeqref{mctrlPlift}\tikzfigS{./qcancilla-completeness/ctrlRXthroughFredkin-step-28}
    \eqeqref{Paltdef}\tikzfigS{./qcancilla-completeness/ctrlRXthroughFredkin-step-29}\\[0.4cm]
    \eqeqref{FredkinwbTOF}\tikzfigS{./qcancilla-completeness/ctrlRXthroughFredkin-step-30}
    \eqeqref{Paltdef}\tikzfigS{./qcancilla-completeness/ctrlRXthroughFredkin-step-31}\\[0.4cm]
    \eqeqref{ancillamctrlPneg}\tikzfigS{./qcancilla-completeness/ctrlRXthroughFredkin-step-32}
    \eqeqref{HHFredkinFHH}\tikzfigS{./qcancilla-completeness/ctrlRXthroughFredkin-step-33}
  \end{gather*}
\end{proof}

\begin{proof}[Proof of \cref{M3}]
  \begin{gather*}
    \tikzfigS{./qcancilla-completeness/showing-M3/showing-M3_00}
    \eqtroiseqref{axQCANCinitdest}{CNOTCNOT}{TOFTOF}\tikzfigS{./qcancilla-completeness/showing-M3/showing-M3_01}\\[0.4cm]
    \eqeqref{ctrlRXthroughFredkin}\tikzfigS{./qcancilla-completeness/showing-M3/showing-M3_02}
    \eqeqref{ctrlRXthroughFredkin}\tikzfigS{./qcancilla-completeness/showing-M3/showing-M3_03}\\[0.4cm]
    \eqeqref{ctrlPthroughFredkin}\tikzfigS{./qcancilla-completeness/showing-M3/showing-M3_04}
    \eqeqref{ctrlRXthroughFredkin}\tikzfigS{./qcancilla-completeness/showing-M3/showing-M3_05}\\[0.4cm]
    \eqeqref{M2}\tikzfigS{./qcancilla-completeness/showing-M3/showing-M3_06-simp}\\[0.4cm]
    \eqeqref{ctrlPthroughFredkin}\tikzfigS{./qcancilla-completeness/showing-M3/showing-M3_07-simp}\\[0.4cm]
    \eqeqref{PthroughFredkin}\tikzfigS{./qcancilla-completeness/showing-M3/showing-M3_08-simp}\\[0.4cm]
    \eqeqref{ctrlRXthroughFredkin}\tikzfigS{./qcancilla-completeness/showing-M3/showing-M3_09-simp}\\[0.4cm]
    \eqeqref{ctrlRXthroughFredkin}\tikzfigS{./qcancilla-completeness/showing-M3/showing-M3_10-simp}\\[0.4cm]
    \eqeqref{ctrlPthroughFredkin}\tikzfigS{./qcancilla-completeness/showing-M3/showing-M3_11-simp}\\[0.4cm]
    \eqeqref{ctrlRXthroughFredkin}\tikzfigS{./qcancilla-completeness/showing-M3/showing-M3_12-simp}\\[0.4cm]
    \eqeqref{ctrlPthroughFredkin}\tikzfigS{./qcancilla-completeness/showing-M3/showing-M3_13-simp}\\[0.4cm]
    \eqeqref{PthroughFredkin}\tikzfigS{./qcancilla-completeness/showing-M3/showing-M3_14-simp}\\[0.4cm]
    %&\eqeqref{PthroughFredkin}\tikzfigS{./qcancilla-completeness/showing-M3/showing-M3_15}\\[0.4cm]
    \eqtroiseqref{axQCANCinitdest}{CNOTCNOT}{TOFTOF}\tikzfigS{./qcancilla-completeness/showing-M3/showing-M3_16-simp}
    \end{gather*}
\end{proof}

\subsection{Proof of Equation \eqref{Mstar} in $\QCancilla$}
\label{appendix:inductionMstar}
\begin{gather*}
  \tikzfigS{./induction-mstar/step-0}
  \eqdeuxeqref{axQCANCinitdest}{RXaltdef}\tikzfigS{./induction-mstar/step-1}\\[0.4cm]
  \eqeqref{RXaltdef}\tikzfigS{./induction-mstar/step-2}
  \eqeqref{Paltdef}\tikzfigS{./induction-mstar/step-3}\\[0.4cm]
  \eqeqref{RXaltdef}\tikzfigS{./induction-mstar/step-4}
  \eqeqref{TOFTOF}\tikzfigS{./induction-mstar/step-5}\\[0.4cm]
  \overset{\text{IH}}{=}\tikzfigS{./induction-mstar/step-6}\\[0.4cm]
  \eqeqref{TOFTOF}\tikzfigS{./induction-mstar/step-7}\\[0.4cm]
  \eqeqref{Paltdef}\tikzfigS{./induction-mstar/step-8}\\[0.4cm]
  \eqeqref{Paltdef}\tikzfigS{./induction-mstar/step-9}\\[0.4cm]
  \eqeqref{RXaltdef}\tikzfigS{./induction-mstar/step-10}\\[0.4cm]
  \eqeqref{RXaltdef}\tikzfigS{./induction-mstar/step-11}\\[0.4cm]
  \eqeqref{Paltdef}\tikzfigS{./induction-mstar/step-12}\\[0.4cm]
  \eqeqref{RXaltdef}\tikzfigS{./induction-mstar/step-13}\\[0.4cm]
  \eqeqref{Paltdef}\tikzfigS{./induction-mstar/step-14}\\[0.4cm]
  \eqdeuxeqref{axQCANCinitdest}{Paltdef}\tikzfigS{./induction-mstar/step-15}
\end{gather*}

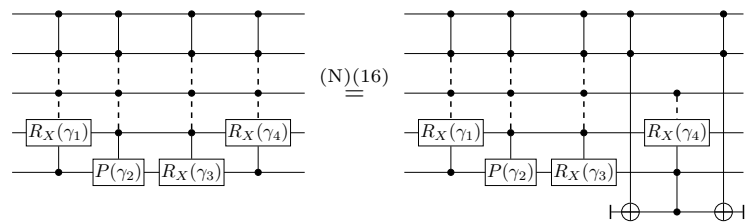
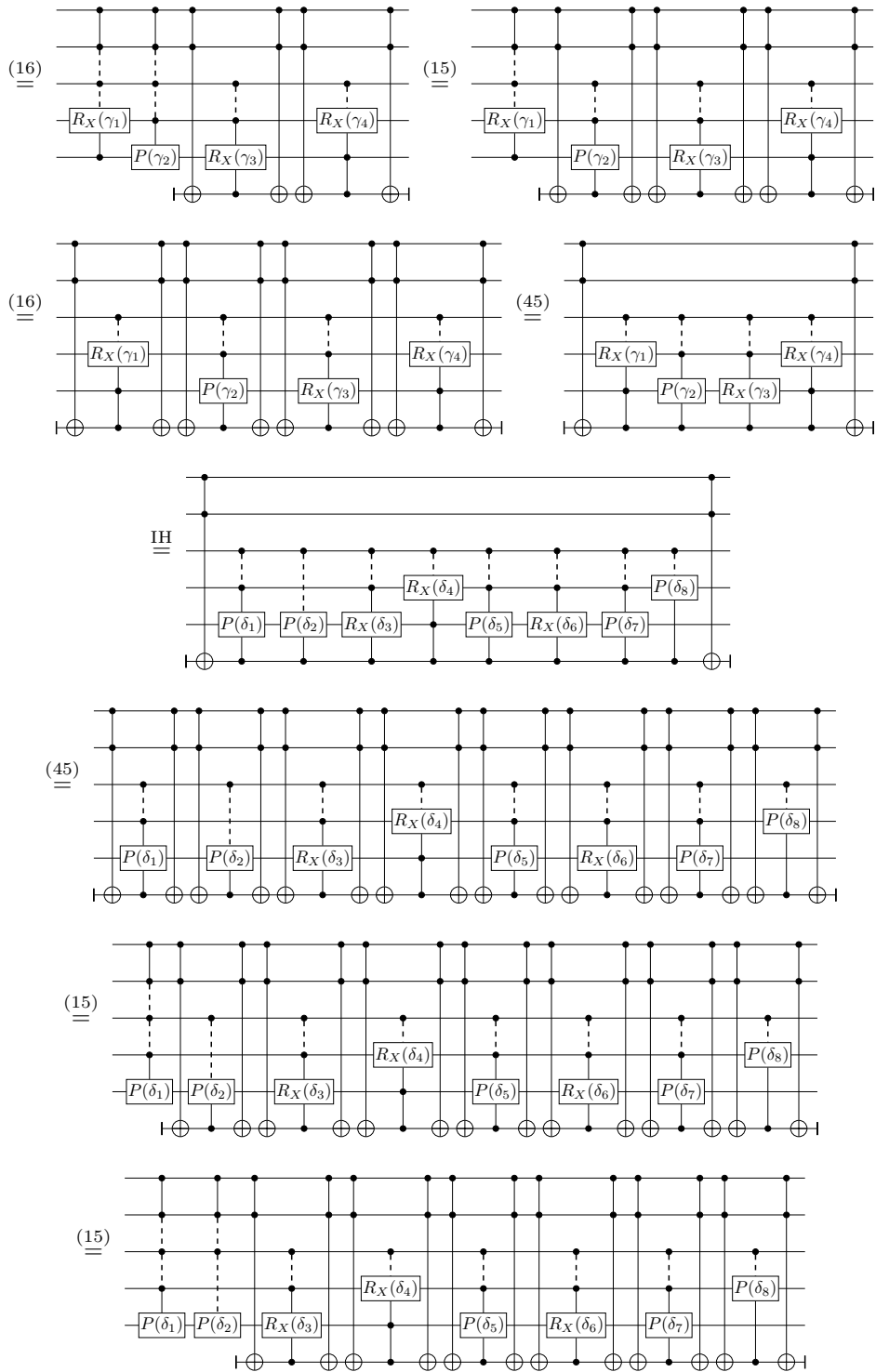
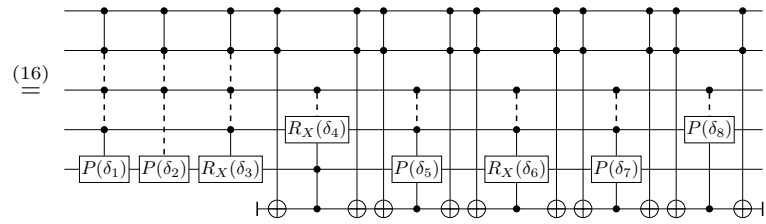
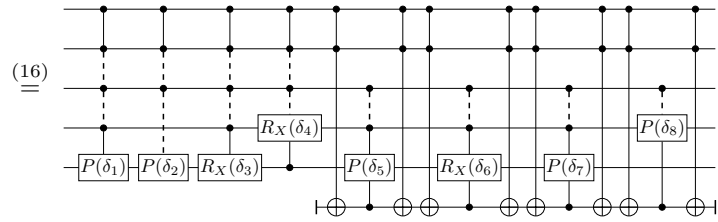
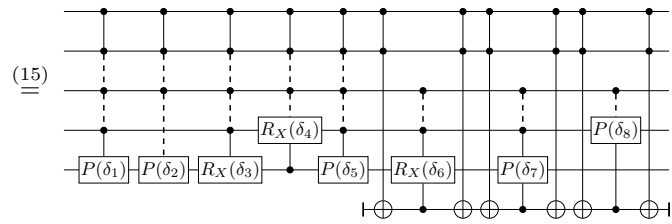
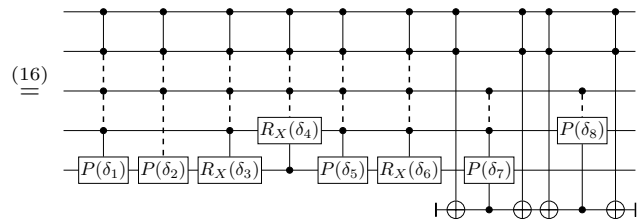
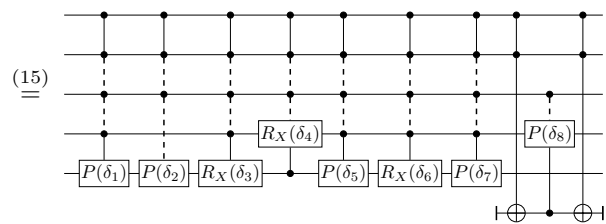
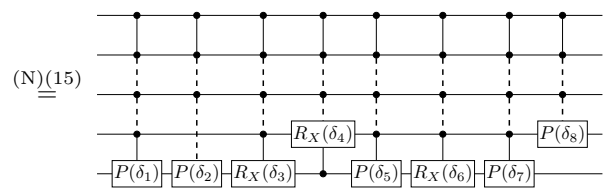
\end{document}